\documentclass{article}
\usepackage[utf8]{inputenc}
\usepackage{amsthm}
\usepackage[margin=0.5in]{geometry}
\usepackage{amsmath}
\usepackage{amsfonts}
\usepackage{graphicx}
\usepackage{enumitem}
\usepackage{caption}
\usepackage{amssymb}
\usepackage{mathtools}
\usepackage{tikz,,pgf}
\usepackage{bbm}
\usepackage{subfig}
\usepackage{url}
\usepackage{calc}
\usetikzlibrary{decorations.markings}

\usepackage{fullpage}
\usepackage{graphicx}
\usepackage{titlesec}
\usepackage[many]{tcolorbox}
\usepackage{centernot}
\tcbuselibrary{listings}

\usepackage [autostyle, english = american]{csquotes}
\MakeOuterQuote{"}

\newtcblisting[auto counter]{mytcblisting}[2][]{
  enhanced,
  breakable,
  arc=0pt,
  outer arc=0pt,
  listing only,
  colback=white,
  colframe=black,
  fonttitle=\normalfont,
  colbacktitle=white,
  coltitle=black,
  before skip=6pt,
  after skip=6pt,
  listing options={
    columns=fullflexible,
    basicstyle=\ttfamily\small,
    numbers=left,
    numbersep=10pt,
    xleftmargin=15pt
  },
  title={ #2},
  #1
}

\titleformat{\section}[block]{\Large\bfseries\filcenter}{}{1em}{}
\titleformat{\subsection}[hang]{\bfseries\filcenter}{}{1em}{}
\titleformat{\subsubsection}[hang]{\bfseries\filcenter}{}{1em}{}
\setlist[enumerate]{label*=\arabic*.}

\newtheorem{thm}{Theorem}
\newtheorem*{assumption*}{Assumption}
\newtheorem*{question*}{Question}
\newtheorem{prop}{Proposition}[subsection]
\newtheorem{dfn}{Definition}[subsection]
\newtheorem{lem}{Lemma}[subsection]
\newtheorem{cor}{Corollary}[subsection]

\newtheorem{question}{Question}[subsection]

\theoremstyle{definition}
\newtheorem{construction}{Construction}

\theoremstyle{definition}
\newtheorem{concquestion}{Question}

\title{Payment $\centernot \implies$ Consensus}
\author{Thomas Orton \\[0.5cm]{\small Senior Thesis, Harvard University.}\\[0.5cm]{\small Thesis Supervisor: Boaz Barak.}
\\[0.5cm]{\small Submitted March 2019. Revised May 2021.}}
\date{}

\begin{document}
\maketitle

\begin{abstract}
Decentralized payment systems such as Bitcoin have become massively popular in the last few years, yet there is still much to be done in understanding their formal properties. The vast majority of decentralized payment systems work by achieving consensus on the state of the network; a natural question to therefore ask is whether this consensus is necessary. In this paper, we formally define a model of payment systems, and present two main results. In Theorem \ref{thm:no-black-box}, we show that even though there exists a single step black box reduction from Payment Systems to Byzantine Broadcast, there does not exist any black box reduction in the other direction which is significantly better than a trivial reduction.  In Theorem \ref{thrm:short-cyc}, we show how to construct Payment Systems which only require a very small number of messages to be sent per transaction. In particular, global consensus about which transactions have occurred is not necessary for payments in this model. We then show a relation between the construction in Theorem \ref{thrm:short-cyc} and the Lightning Network, relating the formal model constructions we have given to a practical algorithm proposed by the cryptocurrency community. 
\end{abstract}

\newpage

\tableofcontents

\newpage

\part{Introduction}
\section{Motivation} 

Decentralized payment systems, commonly known as "cryptocurrencies", solve the following problem: provide a group of participants the ability to "send" and "receive" virtual money to each other, such that no small group of individuals can violate the security of the monetary system. For example, participants should not be able to spend a virtual $\$100$ note twice, nor should they be able to "steal" money from other participants. Since the launch and successful operation of Bitcoin in 2009, we have seen a significant increase in interest, funding and research into understanding decentralized payment systems. For example, at its peak in December 2017, the market capitalization of Bitcoin alone had grown from nothing to over $300$ billion USD in just nine years  \cite{coinmarketcap}. Throughout this timeperiod, it is undeniable that research in distributed consensus algorithms pioneered by Lamport and others \cite{LSP82} in the 1980s has been significantly influential in shaping how both researchers and practitioners think about and design distributed payment systems. Algorithms for consensus aim to solve the problem of coordinating a group of participants to communicate in a way such that after some time, all participants agree on some fact, and there is no way for a small group of individuals to prevent this from happening. The connection to decentralized payment systems is that if everyone can use a consensus protocol to agree on how much money each person has spent, then we can ensure that no-one spends their $\$100$ twice or steals money from someone else. To give a few examples of the extent of this influence:

\begin{enumerate}
    \item All of the top 10 cryptocurrencies\footnote{As of March 2019} by market capitalization provide payment system functionality by reaching consensus on the current monetary state of the network. \cite{coinmarketcap} \footnote{The author makes no claim about the thousands of relatively unknown cyrotocurrencies which often do not have well understood security guarantees.}
    
    \item Prominent researchers in this field introduce the distributed payment system problem as one in which it is crucial all participants agree on which transactions have occurred. \cite{youtubevid-micali}. 
    
    \item Current implementations of cryptocurrencies developed by the academic community are directly based on consensus solutions for Byzantine Agreement \cite{CGMV18} or built on top of consensus protocols \cite{SCG+14} which reach agreement on each transaction.  
    
    \item Prominent members of the cryptocurrency community explicitly and repeatedly frame cryptocurrencies a solution for solving a global consensus problem about the monetary state of the network. \footnote{The development of Bitcoin \cite{Bit} is claimed to have been significantly influenced by known solutions and impossibility results for Byzantine Agreement. For example, it is claimed the author of Bitcoin (whose real identity is unknown) would introduce Bitcoin as a solution to the Byzantine Agreement problem on chat forums. Consequentially, the thousands of cryptocurrencies which are built on top of Bitcoin's Blockchain architecture are all derived from solving the consensus problem as well. The founder of Ethereum, which is considered one of the foundational cryptocurrencies within the cryptocurrency community, maintains a website at \url{https://vitalik.ca/} where he gives his perspectives on designing distributed payment systems. The arguments he gives make heavy use of language and ideas drawn from distributed consensus.}
\end{enumerate}

Based on the above, it seems that the following assumption is implicit in a significant part of the work and discussion in the cryptocurrency community:

\begin{assumption*}
Distributed payment systems cannot exist without achieving regular global consensus about which payments have occurred.
\end{assumption*}

\newpage

Understanding the truth of falsehood of this implicit assumption is of central importance: 

\begin{enumerate}
    \item If we can show in which precise way the assumption is true, we will have a clearer understanding of the canonical structure of a decentralized payment system: any such implementation would need to use tools from consensus literature, and known impossibility results would apply. 
    
    \item If proven false, we will have a better understanding of exactly where the equivalence between consensus and payment systems breaks down. Exploiting the point at which the two problems diverge may lead to new algorithms which break lower bound results inherent in any payment system implementation based on consensus of payments. 
\end{enumerate}

The purpose of this thesis is to closely scrutinize this implicit assumption. By the end of this paper, we hope to convince the reader that the latter scenario is closer to the truth: while consensus certainly implies a payment system, the reverse implication is more nuanced. The two main results of this paper are as follows: In Theorem \ref{thm:no-black-box}, we show that even though there exists a single step black box reduction from a certain model of payment systems to Byzantine Broadcast, there does not exist any black box reduction in the other direction which is significantly better than a trivial reduction.  In Theorem \ref{thrm:short-cyc}, we show that in a certain trust model we can construct payment systems where in the best case (when all participants behave honestly), only  $\mathcal{O}(\log_{d}(N))$ participants exchange messages per transaction. In particular, under certain conditions, transactions do not require global consensus to occur in the model we give. 

\section{Organization and Contributions}

We begin in part \ref{part:prelim} by giving a brief survey of the consensus literature and implementations of modern cryptocurrencies most relevant to the question at hand. The goal of this section is to (a) give context to the current understanding of consensus and cryptocurrencies, and (b) familiarize the reader with known results which will be used in subsequent parts. Results whose solutions provide useful intuition will be proved. We will also discuss how this thesis relates to prior work on trying to remove the need for frequent consensus in Bitcoin, by contrasting the Constructions in Part \ref{part:trusted-parties} with the Lightning Network \cite{lightning-network}. 

\ 

The comparison between payment systems and consensus begins in part \ref{part:payment-systems}. We start by formally defining the minimal functionality any distributed payment system should satisfy in the fault tolerant model. This will lead to the definition of the "Marker Problem", a toy problem designed to encapsulate the key ideas of a distributed payment system. After formally defining a model for this problem, the first set of results we show towards resolving the central question of this thesis is that (section \ref{sec:consensus})

\begin{enumerate}
    \item There exists a single step, black box reduction from the Marker Problem to Byzantine Broadcast. (Proposition \ref{prop:PS-to-BG})
    
    \item There exists no black box reduction from Byzantine Broadcast to the Marker Problem which is significantly better than a trivial reduction.  (Theorem \ref{thm:no-black-box})
\end{enumerate}

Already, this suggests that in the particular model we have chosen, achieving consensus is not inherently required in a payment system. 

\ 

We continue by giving some concrete constructions of solutions to payment systems in the proposed model. In section \ref{sec:message-complexity} we play closer attention to the best case message complexity\footnote{the number of messages sent when when all participants happen to behave honestly} of payment systems, in an attempt to break the inherent lower bound of $\Omega(N)$ messages per transaction in any consensus-based payment system which achieves consensus on transactions\footnote{Such consensus based solutions form the backbone of all cryptocurrencies based on a blockchain construction.}, where $N$ is the number of participants:

\begin{enumerate}
    \item We show a lower bound of $\Omega(Nf)$ on the "total" best case message complexity of payment systems in a weak adversarial model, where $f$ is the number of faults (Proposition \ref{message-complexity:prop:mclb}). We give solutions to the Marker Problem showing this lower bound is tight.  
    
    \item We use the proof of Proposition \ref{message-complexity:prop:mclb} to construct a solution to the Marker Problem "cycle coin", which has the curious locality property that certain payments require only $\mathcal{O}(1)$ messages in the best case (Construction \ref{construction:cycle-coin}). 
\end{enumerate}

In part \ref{part:trusted-parties}, we argue for reasonable trust extensions to the fault tolerant model which are natural for distributed payment systems. By building on the ideas in Construction \ref{construction:cycle-coin}, we then show that under this model, and under certain regularity assumptions about the distribution of transactions:

\begin{enumerate}
    \item There exist payment systems with best case message complexity $\mathcal{O}(\log_{d}(N))$ per transaction, where $d \geq 4$ can be chosen to any function of $N$ if one is willing to assume all participants have initial income $\Omega(d)$ (Theorem $\ref{thrm:short-cyc}$) \footnote{In a semi-egalitarian society where every process has initial income $\Omega(N^{c})$ for $c>0$, then the best case message complexity is $\mathcal{O}(1)$}. This Construction breaks the $\Omega(N)$ best case message complexity lower bound inherent to any consensus based solution which achieves consensus on transactions. 
    
    \item Using similar ideas, there exist simple and realistically plausible ways for participants to coordinate to further reduce the best case message complexity per transaction. 
\end{enumerate}

We comment that the most valuable component of part \ref{part:trusted-parties} is likely the conceptual idea of how to bootstrap multiple solutions of Construction \ref{construction:cycle-coin} in a certain trust model to create highly local transactions with small message complexity. The income regularity conditions are unnecessarily strong and not naturally motivated: there is likely significant room for improvement for developing more sophisticated randomized constructions with naturally motivated transaction distribution assumptions, and this is perhaps an interesting problem for future work. 

\ 

We conclude in part \ref{part:conclusion} by summarizing the key ideas of this paper, commenting on the practical consequence of the results obtained in an idealized model, and proposing problems for future exploration. By the end, we hope that the reader considers the equivalence between distributed payment systems and consensus to be less trivial than initially believed. 

\newpage
\part{Preliminaries}\label{part:prelim}

\section{Fault Tolerance}

\subsection{Introduction}

There has been a substantial amount of research in designing fault tolerant systems for networked processes. The purpose of this subsection is to give a brief introduction and context to this area of research. We give a survey of different model settings in which this problem is considered, state some known results of this area, and give proofs of results we rely on in later chapters. 

\

One of the first formal definitions for consensus over a distributed network was given by Lamport et al. in the 1980s \cite{LSP82}. The motivation given was similar to the following story: Imagine there are a collection of $N$ Byzantine Generals who are currently camping outside of enemy territory, and would like to decide on a plan to attack the enemy the following day. For simplicity, assume they can either decide to attack (denoted $1$) or retreat (denoted $0$). Most importantly, they need to make sure that they all agree on the same plan, or risk dividing the army in half and being defeated by the enemy. All this would be rather simple if the generals could sit together at a table and vote on a decision; however, our generals are rather shy and refuse to leave their tents. Each general $i$ will only communicate with another general via sending letters. To complicate matters further, it is known that up to $f<N$ of the generals are working in secret with the enemy, and will try to ruin any plans of the Byzantines. Each Byzantine has an opinion (either $0$ or $1$) of what the decision should be for the next day. We would like to come up with a strategy the generals can follow such that if all non-traitorous, honest Byzantines think the decision should be to attack (resp. retreat), then all honest Byzantines reach a consensus on this decision, even when the dishonest generals behave maliciously. Even if not all honest Byzantines have the same opinion, we still want to ensure that every honest Byzantine agrees on the same decision at the end (whichever that might be), to avoid half the army attacking and the other half retreating. 

\

To begin solving this problem, we need to formalize how the generals communicate and behave. We imagine a collection of \textit{processes/nodes} $P_{1},\dots,P_{N}$ (i.e. generals) which function within a \textit{network}. At each time step $t \in \mathbb{N}$, each process may send and receive messages to other processes in the network. The reliability of the network to deliver these messages is either \textit{asynchronous}, \textit{partially synchronous}, or \textit{synchronous} \cite{SJS+18}. While an asynchronous network may arbitrarily delay (but eventually deliver) a sent message, a synchronous network is guaranteed to reliably deliver every sent message by the next time period. Partially synchronous networks \cite{DLS88} model a region between these two extremes, where there are some (unknown to each process) guarantees on the delay of messages being delivered. This paper will focus on the synchronous network case. Formally, we have the following definition:

\begin{dfn}\label{prelim:dfn:sync-network:}
A synchronous network consists of a collection of \textit{processes/nodes} $P_1, \dots, P_N$. If you like, you can imagine each $P_n$ as being an algorithm running on node $n$. At each time step $t \in \mathbb{N}$, the following occurs for process $P_n$, $n \in [N]$:

\begin{enumerate}
    \item $P_n$ receives all the messages that were sent to it from another processes at time $t-1$. In particular, $P_n$ sees a list containing elements of the form $(m,n')$, where $m$ is the sent message and $n'$ is the sender. 
    
    \item Based on the contents of all received messages up until time $t$, $P_n$ can send messages to any other process. Formally, $P_n$ is a deterministic function of all past received messages. \label{prelim:dfn:sync-network:decision-step} 
\end{enumerate}
\end{dfn}

If we were to stop here, we would be considering a model for the Byzantine Generals problem in the \textit{unauthenticated} case. In this paper, we will primarily be concerned\footnote{all communication will be assumed to be authenticated unless explicitly stated.} with the \textit{authenticated} case, where we give our generals some extra help: we imagine that each process $P_{n}$ can \textit{sign} a message $m$ to produce the string $SIGN_{P_n}(m):=(m)_{P_n}$. Such a signature is assumed to be unforgeable and tamperproof: no other process can produce a substring of the form $(m)_{P_n}$ unless they copied it from a message originally signed by $P_n$. Such a formalism is designed to model real cryptographic signatures which have similar properties. Without loss of generality, we assume that in the \textit{authenticated} model, all processes sign their messages before sending them. 

\ 

We imagine that at time $t=0$, each process $P_n$ is given an initial value $v_n$ in some finite set $\mathcal{V} \supset \{0,1\}$, which is their opinion of how to attack the next day. We would like our processes to agree on some value at the end. Now, how do the enemy generals behave? What do the agreement strategies we construct have to protect against? Conceptually, it will be useful to imagine a single "adversary" which can corrupt honest generals to make them dishonest, and can coordinate the enemy generals against the honest Byzantine generals. Many different kinds of adversaries have been studied, depending on the kinds of applications considered. For example, \textit{fail-stop} models include adversaries which can cause honest processes to terminate during network execution, and \textit{failure-omission} models allow adversaries to cause some messages to be omitted from delivery. This paper will primarily be concerned with \textit{byzantine} adversaries. Formally, an $f$-Adversary is one which, at $t=0$, can look at all the initial values $\{v_n\}_{n \in [N]}$, and \text{knows} the deterministic strategy all the processes will follow. It can then pick up to $f$ processes to corrupt, making them dishonest. From this point onward, the dishonest processes can behave arbitrarily, while the honest processes behave according to some specified strategy. While dishonest processes cannot forge the signature of any honest process in the authenticated model, each dishonest process is allowed to forge the signature of another dishonest process.

\

Having formally described how both honest and dishonest generals behave, we now need to say what it means for them to reach agreement on a decision. There are two closely related problems which model this. Recall that each process $P_{n}$ is given an initial value $v_n$ at time $t=0$. At some point in the future, $P_{n}$ \textit{decides} on a value $d_n$, the decision it will ultimately follow about the battle the next day. Our solution should consist of a collection of \textit{protocols}, i.e. rules or deterministic functions, which each honest process $P_n$ follows protocol $p_n$. Even though we have not mentioned any randomness, we will give a definition which allows some probability of failure in anticipation of a future model: 

\ 

\begin{dfn}
(Byzantine Agreement):

We say that a collection of protocols $\{p_i\}_{i \in [N]}$ is a solution to the byzantine agreement problem in the presence of an $f-$Adversary with error probability $\epsilon$ if the following conditions hold with probability at least $1-\epsilon$: 

\begin{enumerate}
    \item Consistency: For any two honest processes $P_i,P_j$, we have $d_i=d_j$. 
    
    \item Validity: If $v_i=v$ for all honest processes, then $d_i=v$ for all honest processes. \label{def:byz-consistency}
    
    \item Termination: Each honest processes decides in finite time.
\end{enumerate} 

If the set of initial values is $\mathcal{V}=\{0,1\}$, we call this problem binary Byzantine Agreement.
\end{dfn}

Thus, \textit{deterministic} solutions to Byzantine Agreement, which are protocols which function within the deterministic model we have built, have error probability $\epsilon=0$.

\

A closely related problem is Byzantine Broadcast: instead of each process $P_{n}$ receiving an initial value $v_n$, only the specially selected process $P_{1}$ (the "general leader") receives an initial value $v_1$. If the general leader is honest, then all honest generals should agree with the leader's decision. If the leader is dishonest, all generals should still agree on the same value.

\newpage

\begin{dfn}
(Byzantine Broadcast):

We say that a collection of protocols $\{p_i\}_{i \in [N]}$ is a solution to the byzantine broadcast problem in the presence of an $f-$Adversary with error probability $\epsilon$ if the following conditions hold with probability at least $1-\epsilon$:

\begin{enumerate}
    \item Consistency: For any two honest processes $P_i,P_j$, we have $d_i=d_j$. 
    
    \item Validity: If $P_1$ is honest, then $d_i=v_1$ for all honest processes.
    
    \item Termination: Each honest processes decides in finite time. 
\end{enumerate} 

If the initial set of values $\mathcal{V}=\{0,1\}$, we call this problem Binary Byzantine broadcast.

\end{dfn}

If we have a solution to these problems, it will be useful to quantify exactly how good the solution is. Towards this aim, we give the following definitions: 

\begin{dfn}
Given any solution to a Byzantine problem, we define the following metrics:

\begin{enumerate}
    \item \textbf{Message Complexity}: The total number of messages sent across the network by all honest processors until consensus is reached (i.e. all processes decide). Note that we explicitly do not take message length into account\footnote{This is because we will later solely focus on the number of messages sent in an attempt to differentiate consensus solutions from payment system solutions. The literature often also considers the number of bits per message.}. 
    
    \item \textbf{Signature Complexity}: The total number of signatures sent across the network by all honest processors during network operation. Note that multiple signatures can occur in a single message, and furthermore we count a signature multiple times if it is sent in multiple messages. If a process signs and sends a signed message, i.e. $(("attack")_{P_1})_{P_2}$, this counts as multiple signatures. 
    
    \item \textbf{Round Complexity}: The maximum number of time steps taken until consensus is reached.
\end{enumerate}
\end{dfn}

Before diving into the details, we briefly give an overview of some known results for these problems. 

\ 

After introducing and defining the Byzantine problems, \cite{LSP82} showed that for deterministic processes $\{P_n\}_{n \in [N]}$ operating over a \textit{synchronous} network in the \textit{unauthenticated} case, the the Agreement and Broadcast problems are solvable if and only if $3f<N$. In contrast, in the \textit{authenticated} case, Byzantine Broadcast solvable for all $f \leq N$ \cite{DS83}, but Byzantine Agreement still only remains solvable iff $3f<N$. In an \textit{asynchronous} network, an important result from \cite{FLP85} showed that consensus is impossible even in the weak fail stop adversarial model allowing only a single process to arbitrarily terminate. It was shown in \cite{DLS88} that we can recover from this impossibility result and still reach consensus in a partially synchronous network. 

\ 

The first solutions to the Byzantine agreement problem required sending messages with a combined bit length which was exponential in $f$, namely $\mathcal{O}(N^{f+2})$, and a round complexity of $f+1$. It was later shown in \cite{FL82} that this round length is optimal, and later solutions for Byzantine agreement were given which gave polynomial message complexity \cite{DS83}. In \cite{DR85}, a lower bound of $\Omega(N(f+1))$ for the signature complexity in the authenticated model was given and shown to be tight. What about the difference in difficulty of Byzantine consensus when the set of possible initial values $\mathcal{V}$ is large, verses the binary case $\mathcal{V}=\{0,1\}$? By giving a black box reduction from the multivalued case to the binary case, \cite{TC84} showed that binary and multivalued Byzantine Agreement are essentially equivalent. We will therefore often think of the binary and multi-valued problems as being "the same". 

\ 

Following a categorization of the complexity of the Byzantine problems for deterministic solutions, researchers turned to randomness in an attempt to break these lower bounds. Many of these solutions were based on the idea of using randomness to create "public coins" which could be used to facilitate consensus \cite{R83}. Recently, by making use of a shared random string, a random oracle and cryptographic signatures, \cite{M18} built on these ideas to give a solution \textit{BBA*} to the agreement problem which runs in $\mathcal{O}(1)$ rounds in expectation, has $\mathcal{O}(N^2)$ message complexity in expectation, and tolerates $3f<N$ failures by a computationally bound adaptive adversary. Other randomized solutions have been given which tolerate up to $2f<N$ faults and also run in $\mathcal{O}(1)$ rounds in expectation \cite{KK09}. From the lower bounds side, it has been shown that any randomized solution to the agreement problem has a probability of failing which decreases at best exponentially in the number of rounds in the non-adaptive fail-stop model \cite{CMS89}, \cite{AH10}. 

\ 

The precise model in which a randomized solution to Byzantine Agreement/Broadcast is formulated relies on technical definitions of random oracles, digital signatures, and computationally bounded adversaries which are not central to the ideas of this paper. Despite this, randomized consensus algorithms are used in practical designs of modern cryptocurrencies, and the reductions we give in subsequent sections are fairly agnostic to whether we are in the deterministic or randomized model. To strike a balance, while still being able to give a comparison between payment systems and randomized solutions to consensus, we will simply state results which also apply to randomized consensus, and refer the interested reader to \cite{M18} for more details on what the model for randomized consensus looks like. We informally describe the differences between the randomized and deterministic models here, though these details will not be relevant in this paper:

\begin{enumerate}
    \item Protocols can now be randomized, i.e. they may flip random coins to decide what to do next.
    
    \item Processes have access to certain cryptographic tools, such as digital signatures through a public key interface, shared random oracles, and a public shared random string.
    
    \item The adversary can only run in polynomial time. At the beginning of each time step, it can view the entire network, and can choose some processes to behave dishonestly, up to a total of $f$ over all time. If a process is corrupted by an adversary at some time, we refer to it as dishonest (even before the time it has been corrupted). The adversary then directly controls each process.
\end{enumerate}

If a solution for a problem occurs in this model, we will call it a \textit{randomized} solution. These solutions may have error probability $\epsilon>0$. Because the randomized setting only gives processes access to more primitives and requires the adversary to be computationally bounded, a solution in the deterministic model is automatically a solution in the randomized model. 

\newpage

\

\subsection{Constructive Results}

We now present some specific solutions and relations between Byzantine Broadcast and Byzantine Agreement. 

\ 

\begin{prop}\label{prelim:prop:BB-det}
For any $f < N-1$, there exists a deterministic solution to the Byzantine Broadcast problem which can tolerate up to $f$ corruptions. 
\end{prop}

\begin{proof}
(From \cite{DS83})

\ 

Recall that we let $SIGN_{P_i}(v)=(v)_{P_i}$ denote the string signed by $P_i$. Likewise, we let $((v)_{P_{i_{1}}}\dots)_{P_{i_k}}$ denote the string obtained by $P_{i_1}$ signing $v$, and then $P_{i_2}$ signing the resulting message, and so forth. If $i_i=1$ (i.e. the first signature is from the broadcasting process $P_1$), and $P_{i_j} \not = P_{i_w}$ for $j \not = w$ (all the signature identities are distinct), we call a message of this form \textit{proper} of length $k$, and we can refer to $v$ as the \textit{value} of the proper message. We also arbitrarily designate $f+1$ "relay" processes $P_{j_1},\dots,P_{j_{f+1}}$, where none of these are the broadcaster $P_1$. 

\

If $P_1$ is honest, it signs its initial value by computing $m=(v_1)_{P_1}$, and sends this message to all processes. 

\

Now consider the following protocol for an honest process $P_n$. Each processor keeps a list $L$ of values it has seen before. At the beginning of time step $i+1$, $P_n$ lexicographical orders all its received messages during the previous time step. $P_n$ then iterates through each message $m$ in order and does the following:

\begin{enumerate}
    \item if $m$ is not proper, or is proper but is not of length $i$, discard $m$. 
    
    \item if $value(m) \in L$, discard $m$. 
    
    \item if $|L| \geq 2$, discard $m$. 
    
    \item Otherwise, add $value(m)$ to $L$. We say that $P_n$ \textit{extracts} $value(v)$ at round $i+1$. Sign $m$ to produce a proper string $m'=(m)_{P_n}$ of length $i+1$. If $P_n$ is a relay processor, send a copy of $m'$ to every other processor. If $P_n$ is not a relay processor, and send a copy of $m'$ to every relay processor.
\end{enumerate}

At time $f+2$, if $|L|\not = 1$, $P_n$ decides the value $0$ ("sender fault"). Otherwise $P_n$ decides the unique value in $L$. 

\ 

We claim that such a construction solves the Byzantine Broadcast problem. First, suppose the broadcaster $P_1$ is honest. Then at step $1$, all honest nodes $P_n$ extract value $v_1$. Moreover, since $P_1$ never produces a signature of the form $(v')_{P_1}$ for $v' \not = v_1$, no proper message with a different value is ever sent across the network. Consequently, after $f+2$ steps, each $P_n$ has $|L|=1$, and all honest processors decide $v_1$. 

\ 

Next, we claim that all honest processors decide on the same value. In particular, we claim: suppose an honest process $P_n$ extracts value $v$ within $f+3$ steps. Then any other honest process $P_{n'}$ has either extracted $v$ or has extracted two distinct values. It then follows that either (a) all processes extract exactly $0$ or $2$ values, or (b) all processes extract the exact same value. In both cases, all honest processes agree on the same value after $f+3$ steps. 

\ 

To prove the claim, suppose that $P_{n}$ has extracted $v$ after $f+3$ steps, but some $P_{n'}$ has not. let $m$ be the proper message of length $i$ with value $v$ from which $P_{n}$ extracted $v$. Choose $n$ such that $i$ is the smallest such integer (i.e. $i+1$ was the earliest round number in which $v$ was extracted by any honest process $P_{n}$). Then we must have $i<f+1$. Otherwise, $m$ has been signed by at least one honest process at an earlier time (and hence $v$ was extracted at an earlier time), contradicting the choice of $i$. Thus $P_{n}$ extracted $v$ by time $f+1$. If $P_{n}$ is a relay process, it then transmits a proper message of length $i+1$ to every other honest process with value $v$: hence, every honest process will either extract $v$ at round $i+1 \leq f+2$, or it will not because it has already extracted two values. If $P_{n}$ is not a relay process, then since at least one relay process is honest, $P_{n}$ sends a proper message of length $i+1$ to some honest relay process $P_{n''}$ with value $v$. There are now two cases: either $P_{n''}$ extracts $v$ in round $i+1$, and we reduce to a previous case, concluding that all processes have either extracted $v$ or two values by time $i+2 \leq f+3$. If $P_{n''}$ does not extract $v$ in round $i+1$, then it has already extracted two values by round $i+1$, and we reduce to the previous case again, where all processes will have extracted two values by round $i+2 \leq f+3$. This completes the claim.

\end{proof}

Notice that in the construction above, the total number of messages sent by all honest processes is $\mathcal{O}(Nf)$. Each honest process, besides $P_1$ and the relay processors, sends at most two groups of messages to relay processors. Each relay processor sends at most 2 groups of messages to non-relay processors. This makes the total number of messages equal to 

$$(\textit{messages from process 1})+(\textit{messages from relay processes})+(\textit{messages from non-relay processes})$$

$$\leq N+2\times (f+1)\times N+2N\times (f+1)=\mathcal{O}(Nf)$$

Since each message contains $\mathcal{O}(f)$ signatures, the total number of signatures sent is $\mathcal{O}(Nf^2)$. Likewise, the number of time steps until a decision is made is $\mathcal{O}(f)$. These will be useful metrics to remember for two reasons: firstly it gives us a sense of how efficient our solution is. Secondly, there are known lower bounds of e.g. how many time steps are needed to achieve consensus in certain models, and we will make use of these in the future.

\

While we will not use it, we also mention that there is a similar result for the Byzantine Agreement problem: 

\begin{prop}
For any $f \in \mathbb{N}$ such that $3f<N$, there exists a deterministic solution to the Byzantine Agreement problem which can tolerate up to $f$ corruptions. The message and signature complexity is polynomial in $N,f$, and the round complexity is $\mathcal{O}(f)$. 
\end{prop}

We comment again that, in contrast to the Broadcast problem, there is no solution for the Byzantine Agreement problem for any $f \in\mathbb{N}$ when $3f \geq N$\cite{LSP82}. \footnote{For example, consider the case $3f=N,f=1$}. 

\ 

For randomized Byzantine Agreement, we mention the following result:

\ 

\begin{prop}\label{prelim:prop:bin-BA-rand}
(From \cite{M18}) For any $f \in \mathbb{N}$ such that $3f<N$, and any $\gamma \geq 1$, there exists a randomized solution to the Byzantine Agreement problem for $\mathcal{V}=\{0,1\}$ which can tolerate up to $f$ corruptions. The total number of messages sent and signatures made is $\mathcal{O}(\gamma N^2)$, the number of time steps taken is $\mathcal{O}(\gamma)$ in expectation, and the probability of error is $2^{-\Omega(\gamma)}$. 
\end{prop}

We comment that there exist solutions tolerating $f$ faults with $2f<N$ and similar complexity given in \cite{KK09}. However, the construction given in \ref{prelim:prop:bin-BA-rand} uses ideas which are perhaps more directly relevant to the implementation of modern cryptocurrencies \cite{CGMV18}. 

\

We can strengthen proposition \ref{prelim:prop:bin-BA-rand} to handle the full Byzantine Agreement problem with arbitrary initial values using the following lemma:

\newpage

\begin{lem}\label{prelim:lem:multi-to-binary}
There exists a black box reduction, using only two extra rounds and $\mathcal{O}(N^2)$ extra messages and signatures, from multi-valued Byzantine Agreement tolerating $f<\frac{N}{3}$ faults to Binary Byzantine Agreement tolerating $f$ faults, i.e. these two problems are essentially equivalent.
\end{lem}

\begin{proof}
(From \cite{TC84})

\ 

Consider the following construction: At the first round, all honest processes send their initial value $v_i$ to every other process. We call an honest process \textit{perplexed} if during this round, at least $\frac{1}{2}(N-f)$ of the values it receives are different from its own (if it receives no value from a process, it assumes the default value $0$ was sent); otherwise we say the honest process is \textit{content}. At the second round, every honest perplexed process sends a message to every other process saying "I am perplexed". 

\ 

Now for each honest process $P_n$, define two arrays $Value_{n}[i], Perplexed_{n}[i]$. Set $Value_{n}[n]=v_n$, and $Value_{n}[j]=$the value process $P_j$ claimed to have during round $1$. Likewise, set $Perplexed_{n}[n]=True$ if $P_n$ is perplexed, and $Perplexed_{n}[j]=True$ iff process $P_j$ claimed to be perplexed in round 2. Lastly, define $Alert_n=True$ if at least $N-2f$ of the elements of $Perplexed_{n}$ are $True$, and $False$ otherwise.

\ 

Lastly, have each honest process $P_{n}$ now run the Binary Byzantine Agreement protocol with initial value $Alert_{n}$. Eventually all honest processes decide on a common value $Alert$. If $Alert=True$, then all honest processes decide on the default initial value $0$. If $Alert=False$, then $P_{n}$ decides as follows: $P_{n}$ initializes a list $L_{n}$: for each $j$ such that $Perplexed_{n}[j]=False$, $P_{n}$ adds $Value_{n}[j]$ to $L_{n}$. $P_{n}$ then decides on the most frequently occurring value in $L_{n}$. 

\ 

We claim this construction gives the required behavior:

\begin{enumerate}
    \item Termination: If the Binary Byzantine Agreement protocol terminates in $t$ steps, then this construction terminates in $t+2$ steps. 
    
    \item Validity: Suppose all honest processes have the same initial value $v$. Since $3f<N$, each process receives at most $f+1<\frac{1}{2}(N-f)$ distinct values at round $1$, and so no honest process is perplexed. Thus at most $f<N-2f$ elements of $Perplexed_{n}$ are true, and $Alert_{n}=False$ for every honest process. By validity of the Binary Byzantine Agreement protocol, all honest processes agree on the value $Alert=False$. Moreover, the value $v$ occurs at least $N-f>\frac{N}{2}$ times in $L_{n}$, so $P_{n}$ decides correctly. 
    
    \item Consistency: If $Alert=True$, all honest processes decide the same value $0$. It remains to consider the case $Alert=False$. Let $P_{n}$ be a content process with initial value $v_n$, and let $v^{*}$ be a most frequently occurring value of the initial values of correct processes. Suppose $v_n \not = v^{*}$. Then $P_{n}$ receives at least $\frac{1}{2}(N-f)$ values different from its own $v_i$ in step $1$, contradicting that $P_{n}$ is content. It follows that $v^{*}$ is unique, and that the initial value of every content process equals $v^{*}$. Since $Alert=False$, at least $f+1$ honest processes are content; otherwise, at least $N-2f$ honest processes would be perplexed, and all correct processes would have $Alert_{n}=True$ at the end of the second round, contradicting the validity of the Binary Byzantine Agreement protocol. Thus each correct process $P_{n}$ has their list $L_{n}$ consisting of at least $f+1$ copies of $v^{*}$ from honest content processes, and at most $f$ other values from dishonest processes claiming to be content. It follows that $v^{*}$ is the unique majority value and all $P_{n}$ decide on the same value. 
\end{enumerate}
\end{proof}

Another way to simplify the number of definitions we have is to notice that a solution for Byzantine Agreement implies a solution of similar complexity for Byzantine Broadcast:

\newpage

\begin{lem}\label{prelim:lem:BB-to-BA}
Suppose there exists a solution to the Byzantine Agreement problem which tolerates $f$ faults. Then there exists a solution to the Byzantine Broadcast problem which tolerates $f$ faults, takes $1$ extra round, and sends $\mathcal{O}(N)$ more messages and signatures than the original solution. In particular, we construct this solution via a black box reduction. 
\end{lem}

\begin{proof}
Have processor $1$ send a signed copy of its initial value $v_1$ to all processors. At time step $1$, if an honest processor sees a single value $v_1$ signed and sent from $P_1$, it takes this to be its value in the Byzantine Agreement game. Otherwise it chooses a default initial value $0 \in \mathcal{V}$. Now all honest processors run the assumed solution for Byzantine Agreement, and eventually decide on a value. This construction uses one extra round and an additional $\mathcal{O}(N)$ messages and signatures. 

\ 

If $P_1$ is honest, all honest processes start with the same initial value for Byzantine Agreement, and by assumption will all decide on value $v_1$. If $P_1$ is dishonest, regardless of the initial values chosen by honest processors, because there are at most $f$ faults, all honest processors will come to a consensus on the same value by assumption of the correctness of the Byzantine Agreement solution. 
\end{proof}

\begin{cor}\label{prelim:cor:BB-rand}
For any $f \in \mathbb{N}$ such that $3f<N$, and any $\gamma \geq 1$, there exists a randomized solution to the Byzantine Broadcast problem which can tolerate up to $f$ corruptions. The total number of messages sent and signatures made is $\mathcal{O}(\gamma N^2)$, the number of time steps taken is $\mathcal{O}(\gamma)$, and the probability of error is $2^{-\Omega(\gamma)}$. 
\end{cor}

\begin{proof}
This follows from Propositions \ref{prelim:prop:bin-BA-rand}, \ref{prelim:lem:multi-to-binary} and Lemma \ref{prelim:lem:BB-to-BA}. 
\end{proof}

\

\newpage

\subsection{Lowerbound Results} 

\ 

Having given some positive constructions, we now survey some impossibility results for these problems. The first is the following:

\begin{prop}\label{prelim:prop:round-complexity-lowerbound}
Any deterministic solution to the Byzantine Generals problem tolerating $f$ faults requires at least $f+1$ time steps in the worst case. In particular, there exists a strategy the adversary can follow which forces the number of steps taken to be $f+1$. 
\end{prop}

We refer to \cite{DS83} for a proof of this fact, which is based on a generalization of a similar theorem given in \cite{FL82}. Note that this is a lower bound result: not every possible network computation requires $f+1$ steps to reach consensus, but there is always some set of choices the adversary can make if it really wants to force $f+1$ steps to be made. For example, if the broadcaster is honest and sends the signed value $(0)_{P_1}$ to all processes, and each process $P_n$ happens to send $((0)_{P_1})_{P_n}$ to every other process in time step 2, then all processes can infer that agreement has been reached and terminate in $2$ steps. By combining the round complexity lower bound (Proposition \ref{prelim:prop:round-complexity-lowerbound}) together with the black box reduction from Byzantine Broadcast to Byzantine Agreement (Lemma \ref{prelim:lem:BB-to-BA}), we get the immediate corollary:

\begin{cor}
Any deterministic solution to the Byzantine Agreement problem tolerating $f$ faults requires at least $f$ time steps in the worst cast. In particular, there exists a strategy the adversary can follow which forces the number of steps taken to be $f$. 
\end{cor}

A later paper \cite{DR85} gives lower bounds for the message complexity and signature complexity of deterministic Byzantine Broadcast: 

\begin{prop}\label{prelim:prop:sig-complexity-lowerbound}
(From \cite{DR85}) Any deterministic solution to the Byzantine Broadcast problem tolerating $f<N-1$ faults has signature complexity $N(f+1)/4$, \textit{even when all processes behave honestly}.
\end{prop}

Proposition \ref{prelim:prop:sig-complexity-lowerbound} is a strong result: it says that even when all processes behave correctly, $N(f+1)/4$ signatures are still exchanged by any protocol tolerating $f$ faults. We note again that we are assuming all sent messages are authenticated with a signature during the sending process. 

\begin{proof}
Consider two executions of the network: in the first execution $H$, all processes are honest and the broadcaster $P_1$ has initial value $v_1=0$. In the second execution $G$, all processes are honest and the broadcaster $P_1$ has the initial value $v_1=1$. Let $M_{H,a,b},M_{G,a,b}$ be the sets of messages with their associated send times, sent by process $a$ to process $b$ during histories $H,G$ respectively. Let $A(n)$ be the set of processes which either (a) received a signature from $n$ in histories $G$ or $H$, or (b) sent their own signature to $n$ in histories $G$ or $H$. Now if $\forall n \in [N]$, we have $|A(n)| \geq f+1$, we are done, since one of history $G$ or $H$ involves sending at least $\frac{1}{2}\sum_{n \in [N]} |A(n)|/2 \geq N(f+1)/4$ signatures. Suppose for the sake of contradiction that $\exists n \in [N]$ such that $|A(n)|\leq f$. We now define a new history $H'$ which proceeds as follows: we make all the processes in $A(n)$ faulty. During the execution of $H'$, we make each process $P_{n'} \in A(n)$ send the messages $M_{H,n',n}$ to $P_n$ at the appropriate times. Towards all other processes $n'' \not = n, n'' \not \in A(n)$, we make $P_{n'} \in A(n)$ send the messages $M_{G,n',n''}$ to $P_{n''}$ at the appropriate times. 

\ 

We need to verify two properties of this construction. The first is that our construction is valid: the messages we require the dishonest processes to send do not violate the integrity of the signatures honest processes during the execution of history $H'$. Note that $n' \in A(n)$ only sends messages to $P_{n}$ which contain signatures from dishonest nodes $A(n)$, and so the adversary is able to "forge" any signatures required for messages that need to be sent to $n$. For $n'' \not = n, n'' \not \in A(n)$, note that $n''$ never receives a signature from $P_n$. Thus for $n' \in A(n)$, all the signatures in a message $m \in M_{G,n',n''}$ are either from signatures already received by $n'$, or from signatures from processes in $A(n)$ (which can be forged). 

\ 

Lastly, note that the received messages of honest process $P_n$ look identical to those in history $H$, so $P_n$ will decide $0$. However, there is at least one honest process $n'' \not = n, n'' \not \in A(n)$ whose received messages look identical to those in history $G$, and so will decide $1$, violating the consistency condition of Byzantine Broadcast. 
\end{proof}

\ 

Again by Lemma \ref{prelim:lem:BB-to-BA}, the analogous result holds for Byzantine Agreement. We also have a similar lower bound on the message complexity: 

\begin{prop}\label{prelim:prop:message-complexity-lowerbound}
(From \cite{DR85}) Any deterministic solution to the Byzantine Generals problem tolerating $f$ faults has message complexity at least $\max((N-1)/2,(1+\frac{f}{2})^2)$ in the worst case. In particular, there exists a strategy the adversary can follow which forces the number of messages sent to be $\max((N-1)/2,(1+\frac{f}{2})^2)$. 
\end{prop}

We comment that the adversarial strategy for proposition \ref{prelim:prop:message-complexity-lowerbound} is rather weak: the adversary simply needs to ignore some of its received messages, and behave honestly otherwise.

\newpage

\subsection{Concatenating Protocols}

A common technique for building solutions to larger problems is to use protocols for smaller problems (Byzantine Generals, secret sharing) as building blocks, and we will frequently use this technique. For example, one may wish to have processes run \textit{two} copies of Byzantine Broadcast in parallel or sequence, and then use the decisions from each consensus protocol and combine them in a particular way. However, naively combining protocols can lead to serious flaws in the concatenated protocol, and so is worth mentioning here briefly. For example, imagine simulating two "copies" of Byzantine Consensus in the authenticated setting, one after the other. In the first simulation, processes exchange authenticated messages and achieve consensus. In the second simulation, processes receive new input values and re-run the consensus protocol again. However, if the consensus protocol is blindly re-run, there is no longer a guarantee of consensus for the second round\footnote{For example, see \cite{LLR06} for impossibility results in this direction; note that these results apply to stateless composition of protocols. The same paper shows that if we include the round number of a protocol in a message, we can arbitrarily compose solutions to Byzantine Agreement}. This is because in the second simulation, dishonest processes can reuse the signatures of honest processes from the first simulation (which they would not have been able to acquire otherwise). However, this issue is easily overcome by including a \textit{nonce} in all messages which uniquely identify which simulation the message belongs to. 

\ 

The next few statements are difficult to state formally in a way which captures their full generality without introducing substantial notation, even though the ideas are very simple. Instead we choose to make these claims as high level statements, where the proof will make clear exactly when it is valid to apply them.

\

\begin{dfn}
We say a protocol solution $\{p'_n\}_{n \in [N]}$ \textbf{simulates} a collection of protocols solutions $\{(p_1)_{n}\}_{n \in [N]},\dots,\{(p_K)_{n}\}_{n \in [N]}$ with unique nonces\footnote{we assume these nonces have never been used before in the network execution} $nonce_1,\dots,nonce_k$ if  for all $n$, $p'_n$ stipulates running copies of protocols $(p_1)_n,\dots,(p_K)_n$ with their associated nonces. Recall that in the authenticated setting, we assume all sent messages are signed as $m'=(m)_{P_{j}}$. Formally, $p'$ behaves as follows:

\begin{enumerate}
    
    \item Whenever $(p')_n$ receives a message $m'$, it checks to see that all signatures $(s)_{P_j}$ contained in $m'$ are of the form $s=m \cdot nonce_i$ \footnote{here $\cdot$ denotes string concatenation with a unique symbol between the two concatenated strings} for some fixed $i$. We say such messages \textit{belong} to simulation $i$. If so, it passes the message $m'$ to the simulation $(p_i)_n$. Otherwise it ignores the message. 
    
    \item When the simulation $(p_i)_n$ receives message $m'$, it pretends that it can't see any of the $\cdot nonce_i$ components and behaves as usual. If $(p_i)_n$ wants to copy or sign the signatures of other processes and combine them in a new message, we implicitly assume that it includes the nonce identifier for the $i$th simulation. 
    
    \item Suppose protocol $(p_i)_n$ wants to send a message $m'$ on process $P_n$. $m'$ consists of collections of signed messages where each signature $(s)_{P_j}$ is of the form $s=m \cdot nonce_i$ for some fixed $i$. $p'_n$ simply sends this message over the network.
\end{enumerate}

Note that the signature and message complexity of $\{p'_n\}_{n \in [N]}$ is equal to the sums of the complexities of each simulation. The round complexity remains the same.
\end{dfn}

\
 
\begin{prop}\label{prelim:prop:concat-det} (Informal)
Suppose $p_1,\dots,p_k$ are protocol solutions which individually \textit{succeed} against an $f$-Adversary in the deterministic setting. Then the simulation of these protocols combined, $p'$, also succeeds (all of $p_1,\dots,p_K$ succeed together) against the $f$-Adversary. In the randomized setting, if we combine the Byzantine Agreement protocols $p_1,\dots,p_K$, where each $p_i$ has error $\epsilon$ and is the algorithm in \cite{M18} used to prove proposition \ref{prelim:prop:bin-BA-rand}, then $p'$ also succeeds with error $K\epsilon$. 
\end{prop}

\

\begin{proof}
(Informal:) We give the proof for the deterministic model. Suppose that the combined simulation did not succeed, so that without loss of generality, the particular simulation $(p_1)_n$ did not succeed running on honest process $P_n$ when the other simulations were run in conjunction. In particular, $(p_1)_n$ behaved (decided incorrectly, sent an incorrect message etc) in an unintended way at some time $t$. For $j \in [N]$, let $R_{j,t},S_{j,t}$ be the set of all messages received/sent resp. by $P_j$ during time $t$ which belong to simulation $1$, and let $D$ be the set of dishonest processes. By fixing any initial values to the problem in question\footnote{e.g. the initial values $\{v_i\}_{n \in [N]}$ in Byzantine Broadcast}, for an honest process $P_j$, $(p_1)_j$'s actions only depend on $R_j$. But now consider a new network execution $E$ where all honest processes only run protocol solution $\{(p_1)_{j}\}_{j \in [N]}$ corresponding to simulation $1$. Have the adversary choose the same set $D$ of dishonest processes, and mimic the execution of the $1$st simulation in $E$: Inductively, we claim that at step $t$, the sent and received messages $S_{E,j,t},R_{E,j,t}$ of any process $P_j$ in $E$ at time $t$ is equal to those of $R_{j,t},S_{j,t}$ at time $t$, modulo the identifier $nonce_i$. The case $t=0$ is immediate. At the beginning of time $t$, all processes $P_j$ in $E$ receive the same messages as in $R_{j,t}$ by induction, since these are just the sent messages of the prior round. We need to show all processes in $E$ also send the same messages at time $t$. For honest processes, this follows immediately, because they are deterministic functions of their received messages. For a dishonest process $P_j \in D$ during the execution of $E$, we \textit{stipulate} that it sends the same messages belonging to simulation $1$ that were sent by $P_j$ at this time in the original simulation. We can do this only if $P_j$ is not forging any signatures of honest processes by sending a message $m' \in S_{j,t}$. But by construction, any signatures from honest processes appearing in $m'$ belong to simulation $1$, and so must appear in $R_{j,t}$ (otherwise the adversary would not have been able to send $m'$ in the original simulation). By induction these signatures appear in $R_{E,j,t}$ (modulo $nonce_i$), and so the adversary can comply with this stipulation. It follows that all honest processes in $E$ receive the same messages as in the original simulation. Thus $P_n$ behaves incorrectly in $E$, contradicting that $\{(p_1)_i\}_{i \in [N]}$ succeeds individually.

\end{proof}

\

\newpage

\section{Cryptocurrencies}\label{sec:crypto}

\subsection{Digital Money and Decentralization}

The idea of digital money has been previously studied by cryptographers, mainly with the concerns of privacy and security in mind. For example, \cite{CFN88} gives a construction of a protocol which would allow individuals to interact with a bank in an anonymous way: Alice will be able to spend money from her account without the bank being able to tell where she is spending it. However, if Alice ever tries to spend the same digital coin twice, then she ends up revealing her identity to the bank, and the bank can prove that Alice double spent a coin (and consequently follow with legal action). Similar constructions of this kind include Alice being able to prove, for example, whether the bank is being honest or stealing her money. 

\ 

These solutions make sense when two conditions are met: (a) when there is a specially designated individual, such as a bank, who can be relied on to behave in a certain way because of a regulatory environment, and (b) when there is a realistic threat of legal action if such behavior is not observed. But what if the "bank" Alice is using is an anonymous individual on the internet? Or even if the bank is a known start-up, what if it operates in a foreign country? Even if Alice can prove that the bank is cheating, whether Alice can reasonably follow up with punitive action is a non-trivial concern. One of the key problems decentralized payment systems solve is being able to co-ordinate a large number of individuals to form a payment system, even in the absence of a strong regulatory environment. This makes such a payment system highly robust and accessible to anyone with the minimal ability to send messages across the internet.\footnote{There are a number of other security advantages which are often argued: for example, there is now no longer a central bank the government can use to change the money supply.} Because we cannot rely on any fixed subset of individuals to behave in a particular way in this context, the notion of fault tolerance against any $f$ failures (i.e. bad behavior by any $f$ participants) is therefore certainly a \textit{necessary} requirement for any decentralized payment system to have. Decentralized payment systems therefore traditionally focus on providing a protocol individuals can follow, so that even if any $f$ individuals behave badly, the payment system will still function correctly.

\

In contrast to distributed consensus in formal models, our theoretical understanding of cryptocurrencies is still relatively underdeveloped and an active area of research. For example, a significant portion of current research is focused on just understanding and formalizing Bitcoin's \cite{Bit} particular implementation of a distributed payment system \cite{SJS+18}, which operates via \textit{blockchain} consensus. Other researchers are working on adapting known solutions to Byzantine Agreement to work over the internet as distributed payment systems \cite{CGMV18}, while some members of the cryptocurrency community attempt to more informally generalize the ideas behind the blockchain protocol to achieve greater transaction speeds via "tangles" \cite{tangle}. This research area is very new and constantly evolving. We will give a brief summary of the key model differences between a practical payment system which works on the internet, and a protocol which might operate in a fault tolerant model of the previous section. We will then briefly outline the high level idea of how Bitcoin facilitates a payment system. Since the vast majority of prominent cryptocurrencies operate on similar principles, this will be a faithful representation to keep in mind when thinking about current implementations of practical payment systems. The details of this representation are listed purely to give context, but are not needed for the rest of this paper.

\ 

The key model differences between the internet and the "fault tolerant" model given in the previous section are as follows: 
\

\begin{enumerate}
    \item All participants do not necessarily know each other (the "communication graph" is not fully connected); instead, processes only know of and can message a few neighbouring processes. They therefore communicate to others by "gossiping" to neighbouring processes. For example, if $P_{a}$ wants to send the message "I pay $P_{b}$ $\$1$", $P_{a}$ will send this message to its neighbours, and request that the message be inductively forwarded to their neighbours. Fault tolerance in different network communication typologies has been studied \cite{LSP82}, but it is perhaps unclear how to model the connectivity of arbitrary participants on the internet in a robust way.  
    
    \item Processes are usually allowed to be offline: for example, they can "opt out" of participating in a protocol at arbitrary times, and then rejoin later. Defining a notion of fault tolerance in networks where certain nodes can be "online" or "offline" has been worked on in \cite{BBRP07} and \cite{PS17}.
    
    \item There are substantial financial incentives for processes to behave in non-trivial ways.  Modeling the incentives of processes is actively being studied from a game theory perspective, particularly in the context of Bitcoin \cite{TE18}, \cite{LTKS15}. For example, a number of results have showed that Bitcoin is not \textit{incentive compatible}, in the sense that even two thirds of the participants are honest, it can be more financially profitable for processes to behave dishonestly. \cite{EE18}, \cite{CKWN16}, \cite{SBBR16}, \cite{SSZ16}, \cite{NKMS16} 
\end{enumerate}

\ 

Despite these differences, Bitcoin is still an empirically successful algorithm at achieving distributed consensus between collections of anonymous individuals. At a high level, Bitcoin works as follows: at each round $i$, all honest participants will reach agreement on a \textit{block} of new transactions. This block is then appended to the list of transactions which have been agreed on previously; thus all honest participants have a consensus about who has paid whom and by how much. Thus when someone wants to make a new transaction, all honest participants can check the list of transactions they have already agreed on to verify that there is sufficient balance for the transaction to go through. 

\ 

Define a \textit{ledger} $L_0$, which consists of all transactions currently processed by the network at time $t=0$. Practically, this might consist of a single entry $(pay,NULL,Alice,100)_{Null}$, indicating that $Alice$ starts out with $100$ bitcoins at $t=0$. Now, we imagine a sequence of \textit{rounds} $i=0,1,\dots$. During each round, any number of unknown participants may try to send and receive payments. At $i=0$, Alice is the only one with a positive balance, so only she can make a payment. 

\ 

At round $i$, a participant $Bob$ with non-zero balance might want to pay Alice $1$ Bitcoin. He does this by gossiping the signed message $(pay,Bob,Alice,1,id)_{Bob}$\footnote{Here id is a unique identifier} to his neighbours, hoping that everyone will eventually receive this message. At the end of round $i$, a random leader is elected from the set of all online participants. The ability to elect a random leader is one of the central ideas in being able to extend classical solutions for Byzantine Agreement to those which work when the participants are unknown. For now, let's take it on faith that at the end of round $i$, all participants agree on a leader Charlie for round $i$. Charlie, if he is honest, will look at all the transactions he has received through gossiping. He will then try to put them all together in an extension \textit{block} $E_{i}$, where $E_{i}$ contains the signed messages of all the payments in round $i$. If he cannot include a particular transaction (because maybe Bob tried to pay Alice a bitcoin when he didn't have any balance), Charlie simply ignores this transaction. Finally, Charlie \textit{links} $E_{i}$ with the ledger for the payments in all previous rounds $L_{i-1}$, forming a \textit{chain} $L_i=E_{i} \rightarrow E_{i-1} \rightarrow \dots \rightarrow L_0$. Charlie then signs and publishes $L_i$, and everyone agrees that everyone's balance at the beginning of round $i+1$ is as reflected by the payments listed in $L_i$, provided the extension block $E_i$ is valid\footnote{all the signatures are correct, no-one has negative balance etc}. 

\ 

We now briefly try to motivate why such a construction works, without getting tied down by details. Firstly, notice that regardless of how Charlie behaves, Charlie can never cause Bob to pay Alice an amount Bob did not intend to pay: this is because Bob needs to sign any payment before it can be included in an extension block. Thus, even if Charlie is dishonest, the most damage he can do is block all transactions by not including anyone's transaction in the next block. If we assume that $2f<N$, then if we elect a random Charlie at each round, at least half of the time we will have an honest Charlie which will allow transactions to be appended to the chain. Thus we will always make some progress in processing transactions over time.  

\ 

The non-trivial part of Blockchain is electing a random leader: these details make the description just given slightly less clean. Bitcoin does this by allowing any process which would like to be the leader try to solve a random puzzle. As a concrete example\footnote{there are less computationally intensive ways to achieve random leader election}, imagine all processes have access to a random function $H$; if you can find a value $v$ such that the last $k$ digits of $H(v||E_{i} \rightarrow L_{i-1})$ are all $0$, then you can publish $v$, together with an extension $E_i$, to all processes to prove you are a leader for extending $L_{i-1}$ at round $i$. \footnote{processes who opt into this role, \textit{miners}, are given incentives to do so by receiving a monetary reward in Bitcoin if they become the leader} The assumption is that the only way to find such a $v$ is to try different random values, eventually finding such a $v$ after $2^{k}$ steps in expectation . The time at which the next leader finds such a $v$ is random. Moreover, there may be two distinct processes $P_{a},P_{b}$, each with two distinct proposed extensions $E_{a},E_{b}$, which both find valid values $v_{a},v_{b}$ for round $i$. In this case, both extensions are accepted, and the network is currently uncertain whether $L_{a}=E_{a} \rightarrow L_{i-1}$ or $L_{b}=E_{b} \rightarrow L_{i-1}$ reflects the true balance. However, at round $i+1$ a new leader $P_{c}$ is elected. $P_{c}$ needs to choose which chain and extension they would like to propose; the value $v$ which $P_{c}$ finds is, with high probability, only valid for one of $L_{a},L_{b}$, and so $P_{c}$ only extends one of these. We stipulate that honest leaders should only try to extend the longest chain, and moreover that the longest chain is the one which specifies the "true" balance of all participants. Under certain assumptions, one can show that after a few extensions, it is always clear whether a certain extension will continue to stay in the longest chain or will be forever rejected \cite{GKL15}. If we can be sure $E_i$ is always in the longest chain, then we can consider all the transactions in $E_i$ has having been \textit{confirmed}.

\ 

In the particular example we have given, we are uniformally electing a leader proportional to how many times they choose to evaluate the function $H$ in an attempt to find a valid value $v$, and we assume that no subset of participants utilizing a total of one half of the computational power of the network is all dishonest. This is like the condition $2f<N$, where $N$ now represents the total computational power of the network. Presumably an adversary cannot maliciously coordinate so much computational power\footnote{This has been shown to be a questionable assumption in practice, because large companies specialize in monopolizing computational power for mining bitcoin due to economies of scale. There are other choices for what $N$ can represent, for example the total money in the system. Then the condition $2f<N$ says that an adversary cannot coordinate more than $\frac{1}{2}$ of the total wealth of the payment system to behave in an adversarial way.}, and when they cannot, Bitcoin is in some sense secure. But suppose that $f>\frac{1}{2}N$. Then the adversary can launch the notorious $51\%$ attack: imagine Alice pays Bob $\$1$ in block $E_1$. After a few more blocks, where everyone behaves honestly, Bob sees the chain $E_5 \rightarrow \dots \rightarrow E_1 \rightarrow E_0$. Since $E_1$ is so far down in the chain, ordinarily, if more than say two thirds of the computational power is honest, $E_1$ would stay in the longest chain forever with high probability. Bob sees that Alice's transaction is therefore confirmed, so he sends the physical goods Alice purchased. However, after receiving the goods from Bob, Alice uses her $51\%$ computational power to produce a new block $E_1'$ with no transactions, and creates the chain $E_1'\rightarrow E_0$. Nothing has gone wrong yet; $E_5 \rightarrow \dots \rightarrow E_1 \rightarrow E_0$ is still the longest chain, so everyone agrees that $E_1$ (and hence Alice's payment to Bob) is confirmed, so Bob's balance is still $\$1$. But now Alice chooses to \textit{only} extend the chain $E_1' \rightarrow E_0$. Because Alice has the majority of the computational power, after some amount of time she will be able to produce a long chain $E_k \rightarrow \dots \rightarrow E_1' \rightarrow E_0$ which is, with high probability, longer than any other chain that would have been produced by everyone else \textit{even if everyone else was only extending the chain} $E_5 \rightarrow \dots \rightarrow E_1 \rightarrow E_0$. Since the new longest chain no longer contains Alice's payment to Bob, Bob loses his $\$1$. Thus Bitcoin is not secure if $f>\frac{N}{2}$. 

\newpage

\subsection{Prior Work on Reducing Consensus} 

A major drawback Bitcoin suffers from is the time it takes to create a new block extension\footnote{This is a trade-off between security and efficiency: longer block times mean more stability, but slower transactions.}, and the large space required to store all transactions on a chain. As a result, Bitcoin can only handle on the order of ten transactions per second, compared to the tens of thousands per second achieved by modern credit card services. 

\ 

It is within this context that members of the cryptocurrency community have been trying to reduce the amount of information which needs to be agreed on through consensus in order for a transaction to occur: given the use of a consensus mechanism for transactions, how can we make it more efficient? Note that there's a nuanced difference between this question, and the one which starts by asking whether global consensus on transactions is needed at all. 

\ 

It is challenging to give a complete and accurate account of the efforts which have been pursued in this area. For example, there are over 2000 cryptocurrencies which are registered on CoinMarketCap alone\footnote{As of March 2019.}. Many of these are slight variations of the blockchain protocol, tailored for a particular use case. Moreover, even when a purportedly novel solution to consensus is presented, it is usually done so in the following manner. A short whitepaper will be produced, sketching the software engineering details of how the protocol works. Sometimes there might be some discussion about various attacks against the protocol, and how they will not succeed. In rare cases, authors might heuristically argue that some statistical method will guarantee security by doing some calculations. But in almost all instances of practical releases of cryptocurrencies by the "non-academic" community, there are no proofs of security or correctness. Indeed, even the security of Bitcoin is a relatively open problem. It therefore makes it very difficult to assess which claims should be taken seriously: a software developer might publish a claim that they have a protocol with minimal use of consensus, but whether their protocol is provably secure in an adversarial model is another matter. 

\

Perhaps the most serious attempt to reduce the amount of consensus needed for transactions to occur is the lightning network \cite{lightning-network}. This is a solution designed to reduce the number of transactions which need to be globally published on the blockchain ledger. It is still a controversial solution within the blockchain community\footnote{For example, Roger Ver and other prominent cryptocurrency figures are vocal skeptics \cite{youtubevid-ver}}, but it is currently being experimented with in a semi-live setting\footnote{However, its launch has also repeatedly been delayed ever since its initial proposal in 2016.}. By pursuing the question of consensus, we will be led to constructing a payment system in part \ref{part:trusted-parties} which will turn out to have similar characteristics to the lightning network. We therefore include in the appendix a description of how the lightning network works, and how the constructions in Part \ref{part:trusted-parties} relate to this prior work. This can be examined after reading Part \ref{part:trusted-parties}, and we will draw the reader's attention to the appendix when this point comes.

\newpage
\part{Payment Systems}\label{part:payment-systems}

\section{Payment Systems}

Inspired by the task of trying to compare the consensus problem to the problem of sending payments over a distributed system, we now give a model and definition for a payment system. In particular, we would like our model to capture the minimal functionality any reasonable payment system should possess. Similar to the motivation of Byzantine Generals, we imagine $N$ participants who would like to decide on a set of rules such that, by only passing signed messages between them, they can create a system which allows people to send payments between each other. Such a system needs to be robust to the many $f$ individuals who would like to break the system due to financial incentives. 

\ 

Real currency has worked in the past as follows: individuals have unforgeable, or difficult to replicate, discrete physical objects (paper notes, shells, precious metals etc) which act as a marker that an individual has some value. These markers themselves need not have inherent value; what is important is the inductive belief that if Alice has such a marker, a second party Bob will accept the marker as payment. The only reason Bob accepts the marker as payment is because he too believes that he may find a third party Charlie who will accept the marker as payment from Bob, and so on. More formally, our definition of payment system should encapsulate the following notions

\begin{enumerate}
    \item Participants can hold some notion of "marker".
    
    \item If a participant holds a marker, this means that in the future, they can transfer the marker to someone else, and consequently lose the marker. 
    
    \item If a second party receives a marker, they know who they received it from.
\end{enumerate}

\ 

Decentralized payment systems are concerned with providing these marking functionalities in the context where some (unspecified) participants may actively try to cheat. Physical currencies have the advantage that they cannot easily be replicated. In contrast, digital decentralized currencies are less obvious to implement because if one has a digital coin, it can easily be "copied". Concretely, if Alice has some protocol she can follow which involves interacting with some network participants $P$ and constitutes paying a coin to Bob, nothing stops her from repeating the same protocol with a disjoint set of participants $P'$ and paying the same coin to Charlie afterwards. Since $P,P'$ are distinct, they have no way of knowing that Alice spent her coin twice. Note however in our analogy, that if Alice gives Bob a shell as payment, there is no a priori need for Charlie to also agree that Alice gave Bob a shell (which would be the case if consensus about transactions is reached), unless perhaps Bob later chooses to pay Charlie. 

\ 

We now formalize these notions in the same fault tolerant model as Byzantine Consensus. Formally, we consider a set of processes $P_1,\dots,P_N$ which interact via point to point messages in the \textit{synchronous}, \textit{authenticated} setting of Definition \ref{prelim:dfn:sync-network:}. We break the time steps $t=0,\dots$ into $K$ \textit{rounds}, each round consisting of $T$ \textit{steps}. Thus after $i$ rounds, $iT$ steps have passed. We imagine each process $P_n$ as starting with an initial balance of $v_{n,0} \in \mathbb{N}$ markers at the beginning of round $0$, and we will notate $v_{n,i}$ to be the number of markers $P_n$ has at the beginning of round $i$. In general, we will only be able to talk about $v_{n,i}$ for honest processes $P_{n} \in \mathcal{H}$. At the beginning of round $i$, each honest process $P_n$ receives an input $I_{n,i} \in [N]$. If $P_n$ is honest and $v_{n,i}=0$, then we stipulate $I_{n,i}=n$. Semantically, this means "$P_{n}$ would like to send a coin to $I_{n,i}$ in round $i$". For the remaining $T$ steps of round $i$, honest processors $P_{n}$ follow some protocol $p_{n}$. We will let $\mathcal{H}$ denote the set of all honest processes, and we assume that the adversary knows all future inputs of all processes. At the end of round $i$, each honest process \textit{decides} on 

\begin{enumerate}
    \item A value $v_{n,i+1} \in \mathbb{N}, v_{n,i+1} \geq 0$, its balance at the end of round $i$.
    
    \item A list of senders $S_{n,i} \subset [N]\cup\{\perp\}$ which $P_{n}$ believes sent it coins in round $i$, where we allow duplicate elements and $\perp$ is a dummy sentinel value. 
\end{enumerate}

\begin{dfn}\label{def:full-PS}
We say that a collection of protocols $\{p_i\}_{i \in [N]}$ is a solution to the payment system problem in the presence of an $f-$Adversary with error probability $\epsilon$ if, with probability at least $1-\epsilon$, for all rounds $i \in [K]$ we have

\begin{enumerate}
    \item (\textbf{safety S1, non-duplication}) $$\sum_{P_n \in \mathcal{H}} v_{n,i}\leq \sum_{P_n \in \mathcal{H}} v_{n,0}$$  and $$\forall n \in [N], v_{n,i} \geq 0$$
    
    \item (\textbf{safety S2, non impersonation}) Suppose $P_{n_1},P_{n_2} \in \mathcal{H}$, and $P_{n_1} \in S_{n_2,i}$. Then $I_{n_1,i}=n_2$. 
    
    \item (\textbf{safety S3, self-consistency}) Suppose $P_n \in \mathcal{H}$. Let $\delta=1$ if $v_{n,i}>0$, $0$ otherwise. Then $v_{n,i+1} = v_{n,i}+|S_{n,i}|-\delta$.
    
    \item (\textbf{liveness L1}) Suppose $P_{n_1} \in \mathcal{H}$. Then $R_{n_1,i}\subset \{P_{n_2} \in \mathcal{H}|I_{n_2,i}=n_1 \land v_{n_2,i}>0\}$
    
\end{enumerate}

These properties should hold over all input sequences over $K$ rounds, and over all initial value distributions $\{v_{n,0}\}_{n \in [N]}$. If at most one honest process $P_{n}$ receives an input $I_{n,i} \not = n$ per round, we call this the \textit{single transaction per round} model. Otherwise we refer to the \textit{multi-transaction per round} model.
\end{dfn}

The best way to get an intuition for these conditions is to consider the special case where $v_{1,0}=1$, and $v_{n,0}=0$ for $n \not = 1$.  In this case, after each round, call the (at most one) honest process which decides it has non-zero value the \textit{marked} process. S1 simply says at most one honest process $P_{n}$ decides it is marked. S2 says that if $P_{n}$ thinks $P_{n'}$ was the marked process in the previous round, and $P_{n'}$ is honest, then $P_{n'}$ was indeed so. S3 and L1 together say that if $P_{n'}$ is marked and receives input $n$, then $P_{n}$ will decide it is marked in the next round, and that the previous marked process was $P_{n'}$. In this special case, we such a problem the \textbf{marker game/marker problem}. Note that this is necessarily a single transaction per round model. More formally, we have:

\begin{dfn}\label{payment-systems:df:marker-game}
(\textbf{The Marker Problem}) In the same setting as before, processes communicate for $K$ rounds and have to tolerate byzantine adversaries: 
\begin{enumerate}
    \item If process $P_1$ is honest, it starts with a marker at the beginning of the first round. We say it is "Marked". Otherwise, the marked process is $\perp$ (a dummy value indicating no marked process for this round). 
    
    \item At the beginning of each round, the marked process $M \in \{P_{n}\}_{n \in [N]} \cup \{\perp\}$ is given an input $I_{i}=n \in [N]$. At the end of this round:
    
    \begin{enumerate}
        \item (consistency) At most one honest process $P_n \in \mathcal{H}$ should decide that it has the marker and become the "marked node". 
        
        \item (liveness) If $M,P_{n} \in \mathcal{H}$, $I_{i}=n$, then $P_{n}$ should decide that it now has the marker, and that the marked node in the previous round was $M$.
        
        \item  (non-impersonation) If $P_{i} \in \mathcal{H}$ becomes marked at the end of this round and $M=\perp$, then $P_{i}$ should decide that the previous marked node was $d \in (\{P_{n}\}_{n \in [N]}\cup \{\perp\})\setminus \mathcal{H}$. 
    \end{enumerate}
    
\end{enumerate}

\end{dfn}

\newpage

Similarly to the byzantine generals problem, we can define analogous metrics of solution complexity:

\begin{enumerate}
    \item \textbf{Message Complexity Per Round}: The total number of messages sent across the network by all honest processors per round. We let the \textit{amortized} round complexity be this value divided by the number of honest processors which made a payment in the round in question. 
    
    \item \textbf{Signature Complexity Per Round}: The total number of signatures sent across the network by all honest processors per round. 
    
    \item \textbf{Round Complexity}: The number of time steps taken per round, i.e. $T$. 
\end{enumerate}

The first proposition we prove is that the marker game is as general as a full payment system: if we have a solution to the marker game, then we can construct one for a payment system as well with minimal overhead. It will therefore suffice to focus on understanding this simpler problem first.

\begin{prop}\label{prop:maker-is-general}
Suppose a collection of protocols $\{p_n\}_{n \in [N]}$ deterministically solve the marker game with round length $T$ for $K$ rounds. Then there exists a solution $\{p'_n\}_{n \in [N]}$ which solves the payment problem with round length $T$ for $K$ rounds for any initial marker distribution $\{v_{n,0}\}_{n \in [N]}$. 
\end{prop}

\begin{proof}

Let $V=\sum_{n \in [N]} v_{n,0}$. The key idea is to have $p'_{n}$ simulate $V$ copies of the marker game solution. By permuting process labels, we can assume without loss of generality that for any $n \in [N]$, we have a solution to the marker game where $P_{n}$ starts as the marked process. For each $n \in [N]$, we obtain a group $g_n$ of $v_{n,0}$ copies of a solution to the marker game $\{p_{n'}\}_{n' \in [N]}$, where $P_n$ starts as the marked process. We define $p'_{n}$ to be the process which simulates the $V$ copies of the marker game for each group $g_{n'}, n' \in [N]$. By proposition \ref{prelim:prop:concat-det}, we know that we can assume each individual simulation has the same guarantees the marker game provides. 

\

$p'_{n}$ interacts with its simulations as follows: At each round $i$, $p'_{n}$ \textit{decides} that it has value equal to the number of times it is marked in each of its simulations, i.e. the number of \textit{marked simulations} it has. If a marked simulation $p_{n'}$ of $p'_{n}$ decides that $s$ was the marked process in round $i$, then $p'_{n}$ includes $s$ in $S_{n,i}$. Now suppose $p'_{n}$ receives input $I_{n,i}$ in round $i$. $p'_{n}$ does the following: First it tries to choose a marked simulation, otherwise it chooses no simulation (If $v_{n,i}>0$, such a marked simulation can be chosen by definition). It then sets the input register of simulation $p_{n}$ to be $I_{n,i}$, and continues to simulate all $V$ simulations for $T$ steps. This completes the description of $p'_{n}$. 

\ 

For correctness, it suffices to check each condition:

\begin{enumerate}
    \item (\textbf{safety S1, non-duplication}) $$\sum_{P_n \in \mathcal{H}} v_{n,i}= \textit{ the total number of marked simulations across all processes }  \leq V = \sum_{P_n \in \mathcal{H}} v_{n,0}$$ where the inequality follows by consistency of each simulation. $\forall n \in [N], v_{n,i} \geq 0$ by construction.  
    
    \item (\textbf{safety S2, non impersonation}) Suppose $P_{n_1},P_{n_2} \in \mathcal{H}$, $P_{n_1} \in S_{n_2,i}$. This means $P_{n_1},P_{n_2}$ run each of their simulations honestly. Then some marked simulation $k \in [V]$ of $P_{n_2}$ decided that $P_{n_1}$ sent its marker to $P_{n_2}$ in round $i$. Suppose for the sake of contradiction that $I_{n_1,i} \not = n_2$. Then $P_{n_1}$ never set $n_2$ as input to simulation $k$ in round $i$. If $P_{n_1}$ was marked in simulation $k$ at the beginning of round $i$, then by liveness $I_{n_1,i} \not = n_2$ is marked in the subsequent round, but this contradicts consistency of simulation $k$ at round $i+1$. If $P_{n_1}$ was not marked in simulation $k$ at the beginning of round $i$, then either non-impersonation (if there is no marked process in simulation $k$ at the beginning of round $i$) or liveness and consistency (if there is some marked process in simulation $k$ at round $i$ not equal to $P_{n_1}$) of the $k$th simulation will give as the required contradiction.
    
    \item (\textbf{safety S3, self-consistency}) Suppose $P_n \in \mathcal{H}$. $v_{n,i}$ is equal to the number of $P_{n}$'s marked simulations at the beginning of round $i$. By liveness and by construction, this is equal to $v_{n,i-1}$ plus the number of newly marked simulations in round $i-1$, minus one if $P_{n}$ sends a marker in some simulation in round $i-1$ (which by construction occurs if $P_{n}$ has non-zero balance).  
    
    \item (\textbf{liveness L1}) Suppose $P_{n_1},P_{n_2} \in \mathcal{H}$. Then $I_{n_2,i}=n_1$ implies (by liveness) that by the end of the round, a marker game simulation belonging to $P_{n_1}$ decides $P_{n_2}$ was the previous marked process in this simulation, so $P_{n_2} \in R_{n_1,i}$ by construction. 
    
\end{enumerate}

\end{proof}

One can apply a similar idea to show that a particular solution to the marker game implies one for a full payment system in the randomized model, with error probability at most $V\epsilon$ (union bound). Note that in this construction, the round complexity is preserved. If in round $i$ this construction makes payments $\{I_{n,i}=r_n\}_{n \in [N]}$, then the associated message and signature complexities of this construction are equal to the sums of the corresponding complexities in each of the marker simulations which initiate these payments.

\ 

Before moving to analyzing the marker problem in more detail, we preemptively comment on some concerns the reader may have about the definitions given

\begin{enumerate}
    \item Why do we only support sending discrete markers instead of general numerical values as payment? Would we not need to send $10^6$ markers just to make a single payment in some instances? Notice that traditional hard currency systems solve this problem by giving different discrete objects (i.e. notes) different values. By issuing a wide variety of different denominations\footnote{For example, imagine a monetary system consisting of the notes $1,2,4,\dots,2^n$. Then we can pay any integer value in the range $[0,2^{n+1}-1]$ using at nost $n$ notes}, people are able to pay a wide range of values in cash with only a few notes/coins. In the construction of proposition \ref{prop:maker-is-general}, it is relatively simple to see how we might mark different simulations as corresponding to different values.
    
    \item Why do we require \textit{both} the sender $P_{n_1}$ and the recipient $P_{n_2}$ to be honest for a payment to go through? We have two justifications for this: firstly, in a payment between two parties, it is generally implicit that both parties desire the payment to go through ($P_{n_1}$ wants to exchange the payment for goods, and $P_{n_2}$ would like to receive the payment). Thus each party can only harm their interests by behaving dishonestly. The second reason is semantics: one might be concerned by the notion that $P_{n_1}$ is \textit{unable} to pay $P_{n_2}$ unless $P_{n_2}$ adheres to a strict set of rules, but this seems to undermine the ability of $P_{n_1}$ to be able to spend its money any way it chooses. Such a concern may be justified if we were dealing with hard currency backed by a central bank, where all merchants must e.g. accept paper currency as payment as decreed and enforced by rule of law. However, in the context of a decentralized payment system, where there is no central authority, there is no means by which to enforce that participants are \textit{forced} to accept e.g. Bitcoin as a means of payment. We require $P_{n_2}$ to be honest in order for $P_{n_1}$ to be able to make a payment to $P_{n_2}$, just as much as we require $P_{n_2}$ to believe that a cryptographic string of $0$'s and $1$'s is something it should provide goods to $P_{n_1}$ in exchange for. Nothing stops $P_{n_2}$ from simply turning off their laptop and deciding never to accept Bitcoin as a means of payment again. 
    
    \item Why do we require non-impersonation? Why would there a priori be a problem with someone making a payment while claiming to be someone else? We give two reasons: firstly, this definition is easy to enforce: if $P_{n_1}$ wants to send a payment to $P_{n_2}$, simply require that they first indicate this fact to $P_{n_2}$ by sending a signed statement of this intention. The second reason is that this condition gives us the convenience of being able to easily reduce unauthenticated communication to payment systems\footnote{we will say more on this in the next section}.

    \item Why is the fault tolerant model the correct model of security for payment systems? This is an excellent question. Because crytocurrencies have been so strongly influenced by consensus work, the fault tolerant model is the trust model which has been traditionally assumed by the cryptocurrency community when analyzing the security of distributed payment systems. The primary reason we define payment systems in this model is therefore so that we can analyze the consensus assumption in the same context it is traditionally thought to hold in. Recall in section \ref{sec:crypto} that we justified why such trust robustness is \textit{necessary} in the context of the internet: when any particular subset of participants might be anonymous, or cannot reliably be legally challenged by law enforcement, we cannot rely on any particular subset of individuals to behave honestly. However, we never argued why such trust robustness is \textit{sufficient}. For example, consider the following pathological example:  For $P_1$ to send a marker to $P_N$, each of $P_{2},\dots,P_{N-1}$ is given two choices, either (a) or (b). If at least one of these processes choose (a), then $P_1$'s marker is sent through successfully. If all of $P_2,\dots,P_{N-1}$ choose (b), then the marker is randomly distributed to one of the participants $P_2,\dots,P_{N-1}$. If we prescribe all honest processes must choose option (a), then the payment system just described is fault tolerant for any $N-3$ faults! Yet it is very unclear why we would expect any self-interested party to choose option (a) instead of option (b); said differently, it is unclear why honesty is the "default" behavior\footnote{In contrast, "honesty" is a far more natural default state in the context of more traditional problems thought to be solved by consensus: For example, one common traditional motivation for Byzantine Agreement is that if you have a collection of processors in an airplane which communicate together, you would like to make sure that they still collectively reach a valid decision even when there are some hardware faults. In this context, it is clear why honesty is the default behavior -- the processors were designed to be honest in the first place.}. Thankfully, the payment systems we primarily consider will tolerate any $N$ faults, and will not suffer such pathological cases. Regardless, we encourage the reader to think about precisely what the fault tolerant model of trust means in any particular payment system. 
\end{enumerate}


\newpage

\section{Payment Systems in the Fault Tolerant Model}\label{sec:consensus}

In this section, we will study the relation between the marker problem and the byzantine consensus problem in the synchronous fault tolerant model. While this setting is less natural for real world payment systems, it has been well studied for Byzantine consensus, and provides a benchmark on which to compare the problems of consensus and payment. 

\subsection{Reductions}\label{subsec:reductions}

We begin by analyzing reductions between the two problems. The first observation is that a solution to the Byzantine Broadcast problem immediately implies a solution for the marker problem through a black box reduction: in particular, Byzantine Broadcast is at least as hard as the marker problem. 

\begin{prop}\label{prop:PS-to-BG}
There exists a black box reduction from the marker game to Byzantine Broadcast: Suppose $\{p_n\}_{n \in [N]}$ is a solution which solves the Byzantine Broadcast problem with error $\epsilon$, tolerating $f$ faults, with round complexity $T$, and message and signature complexities $mc,sc$. Then for any $K \in \mathbb{N}$, there exists a $K$ round solution to the marker problem with message and signature complexity $mc,sc$ per round, error $K\epsilon$, and where each round uses $T$ steps of communication. 
\end{prop}

\begin{proof}
We concatenate together $K$ rounds of Byzantine Broadcast, where at each round the new marked process is the leader of the next round. The marked process $M$ of round $i$ broadcasts their input $I_{M,i}$ to the network at the beginning of the round, and all honest participants run Byzantine Broadcast with $M$ as the leader. At the end of the round, all honest processes decide that $I_{M,i}$ is the new marked process, and continue inductively. We initialize process $1$ as being the marked in the first round. We need to check that such a construction satisfies the conditions of definition \ref{payment-systems:df:marker-game}. Note by proposition \ref{prelim:prop:concat-det}, we can assume that with probability at least $1-K\epsilon$, all simulations of Byzantine Broadcast runs successfully. Conditioning on this event, we can assume that after the $i$th round, all honest processes reach consensus on a new marked process for round $i+1$. Let $M$ be the marked process in round $i$, and let $I_{M,i}=n$. 

\ 

\begin{enumerate}
    \item consistency: All honest processes reach consensus on a single value at round $i$, so at most one honest processes decides it is marked at round $i$. 
    
    \item liveness: if $M$ is honest, then all honest processes decide on the proposed value $n$. If $n$ is honest, then $P_{n}$ decides $n$ at the end of round $i$, so liveness is satisfied. Since $P_{n}$ is honest, it knows that $M$ was the previous marked process because it decided this in the previous round.
    
    \item non-impersonation: If $P_{n}$ is honest, then inductively $P_{n}$ decided on the leader/marked process of round $i$ at the end of round $i-1$. If $M \in \mathcal{H}$, then $n$ decides $M$. If $M \not \in \mathcal{H}$, then $n$ decides $M \not \in \mathcal{H}$. 
\end{enumerate}

\end{proof}

This black box reduction gives two immediate corollaries:

\ 

\begin{cor}
For any $f<N, K \in \mathbb{N}$, there exists a deterministic solution to the $K$ round marker game which tolerates up to $f$ faults, has message complexity $\mathcal{O}(Nf)$ per round, signature complexity $\mathcal{O}(Nf^2)$ per round, and $T=\mathcal{O}(f)$ steps per round. 
\end{cor}

\begin{proof}
This follows from known solutions to Byzantine Broadcast (proposition \ref{prelim:prop:BB-det}) and the black box reduction (proposition \ref{prop:PS-to-BG}).
\end{proof}

\begin{cor}\label{relations:cor:reduction-rand-sol}
For any $f$ such that $3f<N$ and any $K,\gamma \in \mathbb{N}$, there exists a randomized solution to the $K$ round marker problem which tolerates up to $f$ faults, has message complexity $\mathcal{O}(\gamma \log(K)N^2)$ per round, and $T=\mathcal{O}(\gamma \log(K))$ steps per round, with error $2^{-\Omega(\gamma)}$. 
\end{cor}

\begin{proof}
By corollary \ref{prelim:cor:BB-rand}, there exists a randomized solution to Byzantine Broadcast tolerating $f$ faults with signature and message complexity $\mathcal{O}(\gamma \log(K)N^2)$, running in $\mathcal{O}(\gamma \log(K))$ steps, and with error probability $\frac{1}{K}2^{\Omega(\gamma)}$. We now apply proposition \ref{prelim:prop:concat-det} to concatenate the Byzantine Agreement protocols together. 
\end{proof}

We comment that this black box reduction from payment systems to Byzantine Broadcast leads to a kind of "global consensus" about the state of the network at the end of each round, and is very reminiscent of the idea of a public ledger incorporated in cryptocurrencies which use blockchains \cite{ZXD+17}.  On the other hand, while a Byzantine Broadcast solution can easily solve the marker problem, it is less clear how a solution to the marker problem relates to Byzantine Broadcast. In particular, the safety and liveness assumptions of the marker problem are far more local. For example, even in  $N-1$ payments, it does not follow that by the end of each round, honest processes will be in agreement with which process has been marked. Indeed, to get a sense of this locality, consider the following proposition:

\begin{prop}\label{prop:payment-system-does-not-imply-messages}
There exists a deterministic solution to the marker problem which tolerates $f$ faults for $f<N$, such that even when all nodes behave honestly, after $N-1$ rounds where nodes $1,\dots,N-1$ have each received a payment, node $N$ does not receive a single message during all rounds. 
\end{prop}

We will prove proposition \ref{prop:payment-system-does-not-imply-messages} later in section \ref{sec:message-complexity}. This proposition alone doesn't necessarily say anything: maybe there is some way we can combine solutions to the marker problem which efficiently solves Byzantine Broadcast, and shows us that payment systems are "just as hard" to construct as solutions to the Byzantine Broadcast problem. We have just given a black box reduction from the Marker Problem to Byzantine Broadcast; a natural question is therefore whether there exists a reduction the other way:

\begin{question}\label{reduction:question-1}
Does there exist a reduction from Byzantine Broadcast to the Marker Problem?
\end{question}

To answer this question, we first need to formalize what we mean by a reduction. In particular, notice that we will allow ourselves even more power than is allowed by a black box reduction: our reduction will be able to look at the "state" of each process and its sent and received messages and possibly take advantage of the inner workings of a solution to the marker problem. If it is indeed true that every solution of a payment system essentially makes use of consensus, then by controlling the payment system and looking at its operation, we should be able to make all processes reach agreement on some value. We formalize this idea with the following definition:

\begin{dfn}
A reduction from Byzantine Broadcast to a solution $\{p_n\}_{n \in [N]}$ of the Marker Game which tolerates $f$ faults is a protocol solution $\{p'_n\}_{n \in [N]}$ to the Byzantine Broadcast problem tolerating $f$ faults, with the following exceptions 

\begin{enumerate}
    
    \item Processes can no longer directly send messages across the network, but can only interact through access to $W$ uniquely identifiable simulations of the Marker Game solution $\{p_n\}_{n \in [N]}$ (by permuting labels, we allow different simulations to correspond to different processes starting with the marker). Each reduction step $i$ corresponds to a single round execution of the marker problem solutions. \label{condition:no-msg}
    
    \item At the beginning of reduction step $i$ (round $i$ in the marker problem), process $P_n$ can write to the input register $I_{n}$ of the $j\in [W]$th simulation of the marker game, and can view what the $j$th simulation decides. 
    
    \item At step $i$, process $P_n$ can also look at the inner workings of simulation $w \in [W]$: for example, the memory state of the simulation, and all sent and received messages for $P_n$ indirectly made by running the $w$th simulation. 
\end{enumerate}

We call such a solution a \textit{strong} reduction. If we remove condition \ref{condition:no-msg} and allow processes to send messages across the network as well, we call this a \text{weak} reduction.

\ 

We say that Byzantine Broadcast is strongly/weakly reducible to the Marker Problem with fault tolerance $f$ if Byzantine Broadcast with fault tolerance $f$ is reducible to every solution $\{p_n\}_{n \in [N]}$ of the Marker Problem tolerating $f$ faults. Again, we comment that if a black box reduction from Byzantine Broadcast to the Marker Problem exists, then such a reduction would show Byzantine Broadcast is strongly reducible to the Marker Problem. Thus this definition of reduction is weaker than that of a black box reduction.
\end{dfn}

For strong reductions, we restrict processes to only be able to use the interface of the payment system to avoid the trivial case where processes simply ignore the payment system and implement Byzantine Broadcast by sending messages to each other. What should we be aiming for in such a reduction? Since the reduction from the Marker Problem to Byzantine Broadcast only used a single black box call to Byzantine Broadcast per round, and terminated immediately, at best we might hope that conversely there is a clever way to send around $\mathcal{O}(1)$ coins in the network in $\mathcal{O}(1)$ steps such that we force every honest process to essentially agree on something. Certainly, we should be able to do this with any payment system which works by reaching consensus at the end of every round. On the other extreme, focus on the easier problem of reducing just \textit{binary} Byzantine agreement to the Marker Problem. We know there exist solutions to Byzantine Agreement \textit{in the unauthenticated setting} (and without any help of a marker solution black box) which run in $\mathcal{O}(f)$ steps and send $\mathcal{O}(f^2)$ messages when $f=\mathcal{O}(N)$ \cite{DF+82}. These solutions imply the existence of naive reductions of similar complexity which make no use of the essential properties of a payment system whatsoever. Indeed, notice if we were to have access to a payment system which can send arbitrary values\footnote{By a previous footnote, certainly we can send an encoding of value $v \in \mathbb{N}$ with $\mathcal{O}(\log(v))$ markers.}, we can naively, via a payment system, simulate having access to a network which can send unauthenticated messages: if all processes start out with an arbitrarily large amount of money, an honest process $P_{n_1}$ can send an unauthenticated message $m$\footnote{the only values we really need to send over the network for the solution of binary Byzantine Broadcast are the names of different processes and the values $0,1$} to $P_{n_2}$ by sending value $ENCODE(m)$ to $P_{n_2}$, where $ENCODE(m)$ converts $m$ into an integer. In the next step, $P_{n_2}$ receives this value and determines that the sender was $P_{n_1}$. 

\

It turns out that the answer to question \ref{reduction:question-1} is no; in particular, given simulation access to a payment system, in general we can't do much better at achieving consensus than we would have, had we not been given a payment system at all. Payment systems don't significantly help us solve consensus: 

\begin{thm}\label{thm:no-black-box}

\

For any $f \in \mathbb{N}$ such that $3f<N$, Byzantine Broadcast is not \textit{strongly} reducible to the Marker Problem in fewer than $\frac{f+1}{3}$ simulation steps. If $W$ copies of the Marker Problem are used in the reduction, then the number of steps of any reduction is at least $\max(\frac{f+1}{3},\Omega(\frac{N}{fW}))$. Consequently, there also does not exist any black box reduction from Byzantine Broadcast to the Marker Problem in fewer than this many black box steps.

\ 

For any $f \in \mathbb{N}$ such that $3f<N$, Byzantine Broadcast is not \textit{weakly} reducible to the Marker Problem in fewer than $\frac{f+1}{3}$ simulation steps. If $W$ copies of the Marker Problem are used and $M$ messages are sent in the reduction, then the number of steps of any reduction is at least $\max(\frac{f+1}{3},\Omega(\frac{N-M}{fW}))$. Consequently, there also does not exist any black box reduction, even when allowed to send additional authenticated messages across the network, from Byzantine Broadcast to the Marker Problem in fewer than this many black box steps. 
\end{thm}

\ 

We note that since Byzantine Broadcast is reducible to Byzantine Agreement (proposition \ref{prelim:lem:BB-to-BA}), the analogous statements hold for Byzantine Agreement as well. We prove this theorem in two parts. First, we give a solution to the marker problem with round complexity $3$, $\mathcal{O}(f^2)$ signature complexity per round, and $\mathcal{O}(f)$ message complexity per round. 
 
\ 

\newpage

\begin{construction}\label{construction:simple-det}

\begin{enumerate}

    \item Arbitrarily designate $3f+1$ "broadcast" processes.
    
    \item (Step 1): At round $i$, the marked process $M$ sends the signed message $m_{n}=$"I want to pay process $n$; here is a proof $p$ I am the marked process at round $i$" to each of the broadcast processes.
    
    \item (Step 2): If a broadcast process receives a valid message of the form $m_{n}$, it signs and sends a message of the form $r_{n}="M$ pays $n$ at round $i"$ to process $P_{n}$. If it receives multiple valid messages of this form, it sends at most one message to a potential recipient during this round.
    
    \item (Step 3) If a process $P_{n}$ receives at least $2f+1$ signed messages of the form $r_{n}$, then it decides it has the marker and $M$ was the sender. Otherwise it decides that it does not have the marker. The $2f+1$ signed messages constitute the "proof" that $P_{n}$ is the marked process at round $i+1$. (During the first round, the "proof" that $P_1$ is the leader is just an empty string). 
\end{enumerate}

\end{construction}

\begin{prop}\label{reduction:proof-simple-det}
For any $K \in \mathbb{N}$ and any $f$ such that $3f<N$, construction \ref{construction:simple-det} gives a solution to the marker problem which tolerates $f$ faults, has round complexity $T=3$, and message and signature complexity per round $\mathcal{O}(f),\mathcal{O}(f^2)$ resp.
\end{prop}

\begin{proof}
Construction \ref{construction:simple-det} uses $T=3$ per round, and honest processes send at most $2\times (3f+1)$ messages and $\mathcal{O}((2f+1)(3f+1)+(3f+1))=\mathcal{O}(f^2)$ signatures per round. We claim this gives a valid construction. To see this, first note that at any round $i$, at most one process $P_{n}$ receives $2f+1$ messages of the form $r_n$. Suppose not, and two distinct processes $P_{n},P_{n'}$ have this property. There are at least $f+1$ signers in common between the broadcasters who signed the messages for $P_{n}$, and the broadcasters who signed the messages for $P_{n'}$. At least one of these broadcasters is honest, contradicting that each honest broadcaster send at most one message during any particular round. 

\begin{enumerate}
    \item Consistency: At most one honest process receives the required number of signatures to decide it is the marked node.
    
    \item Liveness: If $M,P_{n} \in \mathcal{H}$, $M$ sends a message to all broadcast processes. At least $3f+1-f=2f+1$ broadcast processes forward a message to $P_{n}$, and $P_{n}$ then decides it is the new marked process and that the previous marked process is $M$. Because only $P_{n}$ has a proof it is the marked process, it can elect a new marked process in the next round if it behaves honestly. 
    
    \item Non-impersonation: If $M=\perp$, then no honest process has a proof that it is marked at round $i$. Consequentially, no honest $P_{n} \in \mathcal{H}$ which decides it is marked will decide that the previous marked process was in $\mathcal{H}$. 
\end{enumerate}

\end{proof}

We comment that solutions for the deterministic marker problem and randomized Byzantine Broadcast both have  expected round length $T=\mathcal{O}(1)$. Construction \ref{construction:simple-det} for the marker problem in the randomized setting is an improvement of the $\log(K)$ round complexity derived in corollary \ref{relations:cor:reduction-rand-sol}. 

\

We can now give a proof of Theorem \ref{thm:no-black-box}:

\begin{proof}
The key idea is to leverage the lower bound results from propositions \ref{prelim:prop:round-complexity-lowerbound}, \ref{prelim:prop:sig-complexity-lowerbound} and \ref{prelim:prop:message-complexity-lowerbound} to derive a contradiction. 

\ 

First consider the strong case, and suppose there is a reduction using fewer than $\frac{f+1}{3}$ simulation steps. By proposition \ref{reduction:proof-simple-det}, this implies a solution to Byzantine Broadcast which simulates multiple copies of a marker game, where each simulation step corresponds to $3$ network time steps. This solution tolerates $f$ faults and runs in fewer than $\frac{f+1}{3}\times 3=f+1$ steps, contradicting the lower bound in proposition \ref{prelim:prop:round-complexity-lowerbound}. Now suppose that we only use $o(\frac{N}{fW})$ rounds. By proposition \ref{reduction:proof-simple-det}, this implies a solution to Byzantine Broadcast with signature complexity $o(\frac{N}{fW} \times f^2 W)=o(Nf)$, contradicting proposition \ref{prelim:prop:sig-complexity-lowerbound} (we note we could have derived the same bound by considering the message complexity). 

\ 

For the weak case, we can repeat the argument which considers the round complexity of Byzantine Broadcast to get that the number of simulation steps is at least $\frac{f+1}{3}$. Suppose for the sake of contradiction that the number of simulation steps is $o(\frac{N-M}{fW})$. Then by proposition \ref{reduction:proof-simple-det}, this implies a solution to Byzantine Broadcast with message complexity $o(N)$, contradicting the lower bound of proposition \ref{prelim:prop:message-complexity-lowerbound}.  
\end{proof}

\

\subsection{Best Case Message Complexity}\label{sec:message-complexity}

The result from the previous section tells us that, at least from the perspective of reductions in the model we have given, consensus and payment are fairly different problems. This gives us hope that we might be able to construct practical payment systems which behave fundamentally differently compared to any consensus based solution. With this goal in mind, the candidate property we choose to focus on for the rest of this paper is the \textit{best case message complexity} of a payment system, which we define to be the number of messages sent per transaction in the event when all processes behave honestly.\footnote{Note that the payment system is still robust to $f$ processes behaving dishonestly, but the number of messages sent may increase in this case.} In particular, we will do our best, ignoring all other complexity considerations, to focus on building a payment system with good best case message complexity under reasonable conditions. We make this choice for the following reasons:

\begin{enumerate}
    \item While round complexity has strong lower bound results in the fault tolerant model, discrete time steps and round lower bounds don't realistically translate to protocols over the internet in practice. 
    
    \item The signature complexity of a solution in the fault tolerant model does not necessarily translate to practice either: there are cryptographic techniques to "compress" chains of signatures in a single message into a short summary string. 
    
    \item Notice that a message complexity of $\Omega(N)$ per transaction \textit{in the best case} is an innate property of any consensus based solution which reaches agreement on all transactions. This is because all participants need to agree on the updated state of the network after a transaction, and so need to all send/receive at least one message. By focusing on reducing the message complexity as much as possible, we will be forced to end up with solutions to the payment problem which are inherently local: for example, if only $\mathcal{O}(1)$ messages are typically sent per transaction, this would intuitively need to rely on a method very different to global consensus of transactions. 
    
    \item Best case complexity can be a realistic benchmark for practical payment systems. For example, suppose we only send large amounts of messages when processes are behaving dishonestly. In practice, if a single individual is causing excessive network stress and this can be detected, they can be ignored/removed from the network. Moreover, we can naturally build into a payment system financial costs or \textit{fees} which are associated with how many messages are sent over the network. 
\end{enumerate}

With this goal in mind, we begin by studying the best case complexity of payment systems. Recall from proposition \ref{prelim:prop:sig-complexity-lowerbound} that every solution to Byzantine Broadcast requires sending $N(f+1)/4$ signatures even when all processes behave honestly. The message complexity of the marker problem, and payment systems in general, is slightly more more nuanced. Recall we showed a construction for a payment system which always sent $\Theta(f)$ messages per payment. It turns out that there are payment systems which might only send any number between $1$ and $N$ messages per payment. Moreover, if we make certain assumptions about the distribution of income in a payment system, we can get away with sending significantly less than $\Omega(N)$ messages per payment. We start by giving an analogous result for some metric of the message complexity of any deterministic payment system which closely mirrors that of proposition \ref{prelim:prop:sig-complexity-lowerbound} for Byzantine Broadcast:

\newpage

\begin{prop}\label{message-complexity:prop:mclb}
For any $K=1$ round, $T$ step deterministic solution to the marker problem tolerating $f$ faults, let $z_{n}$ denote the number of messages sent by all processes after $T$ steps when the marked process $M=P_1$ receives input $I_{1}=n$ and all processes behave honestly. Then 

$$\sum_{n \in [N]} z_{n} = \Omega(Nf)$$
\end{prop}

We note that this result holds even in the case of a weak adversarial model, where the adversary is only able to simulate two copies of honest protocols (it does not take advantage of being able to see the entire network state). To prove this, first  we will give a lemma, which in itself will be instructive in constructing a new solution to the marker problem. Let $X(n)$ be the set of processes which either send or receive a message during the first $T$ steps when $P_1$ receives input $I_{1}=n$ (so $z_n\geq \frac{1}{2}|X(n)|$), and all processes behave honestly.

\begin{lem}\label{message-complexity:lem:message-complexity-intersection}
Let $P_{n_1},P_{n_2} \in \{P_{n}\}_{n \in [N]}\setminus\{P_{1}\}$, and $D:= X(n_1) \cap X(n_2)$. Then either (i) $P_{n_1}$ is contained in $X(n_2)$ or $P_{n_2}$ is contained in $X(n_1)$, or (ii) $|D|\geq f-1$.
\end{lem}

\begin{proof}
Suppose the claim were false: namely $P_{n_1},P_{n_2} \not \in X(n_1) \cap X(n_2)$, but $|D| \leq f-1$. We now give a dishonest protocol each of the processes in $Z:=D\cup\{P_1\}$ can follow ($|Z|\leq f$) which allows $P_{1}$ to pay both processes $P_{n_1}$ and $P_{n_2}$ its coin (i.e. "double spend"), resulting in both processes deciding they are the marked process for the next round, contradicting the consistency property of definition \ref{payment-systems:df:marker-game}. The scenario is as follows:

\ 
\begin{figure}
\begin{center}
\begin{tikzpicture}[fill=gray]
\scope
\clip (-4,-4) rectangle (4,4)
      (2,0) circle (2);
\fill (0,0) circle (2);
\endscope
\scope
\clip (-4,-4) rectangle (4,4)
      (0,0) circle (2);
\fill (2,0) circle (2);
\endscope
\draw (0,0) circle (2) (0,2)  node [text=black, above left] {$X(n_1)$}
      (2,0) circle (2) (2,2)  node [text=black,above right] {$X(n_2)$}
      (-4,-4) rectangle (6,4) node [text=black,above] {};
      \draw (3,0) node[circle,fill=red, inner sep=0pt,minimum size=0.5pt, opacity=.2,text opacity=1] (n2) {$n_2$};
      \draw (1,0) node[circle, inner sep=0pt,minimum size=0.5pt, opacity=.2,text opacity=1] (Z) {$Z$};
      \draw (-1,0) node[circle,fill=red, inner sep=0pt,minimum size=0.5pt, opacity=.2,text opacity=1] (n1) {$n_1$};
\end{tikzpicture}
\end{center}
\caption{Lemma \ref{message-complexity:lem:message-complexity-intersection}}
\end{figure}

Have each node in $Z$ run two copies of the honest protocol (without nonces); the first copy corresponds to a simulation where $P_1$ pays $P_{n_1}$, and the second simulation corresponds to a simulation where $P_{1}$ pays $P_{n_2}$. When $P_{n_{a}},P_{n_{b}} \in Z$ send messages to each other, they prepend a label $\in \{0,1\}$ to their messages to indicate which simulation copy the message belongs to. When they send messages to honest processes not in $Z$, they omit this label. When they receive messages from processes not in $Z$, they can infer which simulation the message belongs to because the sets $X(n_1)\setminus Z, X(n_2)\setminus Z$ are disjoint. To start off the simulation, we have $P_{1} \in Z$ run two copies of honest protocols for $P_{1}$, one where it "imagines" its input being $I(1)=n_1$, and the other where it imagines $I(1)=n_2$.  

\ 

Notice the following: after $T$ steps, the sent and received messages for all processes in $X(n_1)-Z$ are identical to those in the case where all processes are honest and $P_{1}$ has input $I(1)=n_1$. Thus $P_{n_1}$ decides it is marked after $T$ steps. But the same argument tells us that $P_{n_2}$ decides it is marked as well, leading to a contradiction. 
\end{proof}

Notice that in Construction \ref{construction:simple-det}, we always have $|D| \geq 3f+1$, and so this construction represents the "one extreme" of the condition in Lemma \ref{message-complexity:lem:message-complexity-intersection}, namely we always have $|D|\geq f-1$. One might therefore wonder if it is possible to construct a solution which always satisfies the other extreme. Indeed, after finishing the proof of Proposition \ref{message-complexity:prop:mclb}, we will construct a deterministic solution to a payment system which respects the "other extreme" of this condition, namely it will always be the case that either $P_{n_1}$ is contained in $X(n_2)$ or $P_{n_2}$ is contained in $X(n_1)$. For now, we complete the proof:

\begin{proof}

By lemma \ref{message-complexity:lem:message-complexity-intersection}, $\forall P_{n_1},P_{n_2} \in \{P_{n}\}_{n \in [N]} \setminus \{P_1\}$, either (i) $P_{n_1} \in X(n_2)$ or $P_{n_2} \in X(n_1)$, or (ii)  $|X(n_1)\cap X(n_2)|\geq f-1$. Define 

$$K_1=\{n \in [N]\setminus\{1\}||X(n)| \geq f-1\}$$ 

and its relative complement

$$K_2=[N]\setminus (\{1\} \cup K_1)$$

By condition (ii), it is easy to see that 

$$\sum_{n \in K_2} |X(n)| \geq \frac{|K_2|(|K_2|-1)}{2}$$

because for each unordered pair $(n_1,n_2)$ with $n_1,n_2 \in K_2$ we have either $P_{n_1} \in X(n_2)$ or $P_{n_2} \in X(n_1)$. Thus we have

$$2\sum_{n \in [N]\setminus\{1\}} z_n \geq \sum_{n \in [N]\setminus\{1\}}|X(n)|$$

$$\geq \min_{|K_2| \in \mathbb{R}_{+}} \frac{|K_2|(|K_2|-1)}{2}+|K_1|(f-1)$$

$$=\min_{|K_2| \in \mathbb{R}_{+}} \frac{|K_2|(|K_2|-1)}{2}+(N-1-|K_2|)(f-1)$$

$$=f(N-\frac{1}{2})-N+\frac{7}{8}-\frac{1}{2}f^2=\Omega(Nf)$$

\end{proof}

\ 

As a sanity check, note that Construction \ref{construction:simple-det} satisfies this lower bound and is tight: we have $\sum_{n \in [n]} z_n = N \times 2 \times (3f+1)=N(6f+2)$. 

\

The intuition for Lemma \ref{message-complexity:lem:message-complexity-intersection} is as follows: either the intersection $X(n_1)\cap X(n_2)$ is $\Omega(f)$, so we can guarantee that there is some honest process within this intersection who will prevent $P_{1}$ from double spending. This is very similar in spirit to usual Byzantine Consensus reasoning. If this isn't the case, then the only other way we can guarantee $P_{1}$ can't double spend after paying $P_{n_1}$ is if $P_{n_2}$ itself \textit{saw} $P_{1}$ spending its coin already, and hence $P_{n_2} \in X(n_1)$; this is the point at which the reasoning for payment systems formally differs from that of Byzantine Consensus in the fault tolerant model, because we stipulate that $P_{n_2}$ must be honest if it is to have any protection from this double spending\footnote{wheres in Byzantine Consensus, we wouldn't be able to assume that this particular fixed process behaves honestly.}. 

\ 

How might we construct a payment system where the property $P_{n_2} \in X(n_1)$ or $P_{n_1} \in X(n_2)$ always holds (in general, even when $P_1$ isn't the starting marked process)? One (and the only) way to do this is to connect all processes together in a directed cycle, and stipulate that payments can only move "across" the cycle. Since any two paths on a directed cycle originating from a common point always meet at some endpoint on one of the paths, this will give us the desired property. Using this idea, we now try to  give another solution to the marker problem where the number of messages sent per transaction is highly variable -- some transactions only require $O(1)$ messages in the best case, but only under highly unrealistic assumptions about the distribution of payments. In part \ref{part:trusted-parties}, we will then use this construction as a building block to outline how we can achieve $\mathcal{O}(\log(N))$ messages per transaction in the best case under more reasonable assumptions about the distribution of payments. 

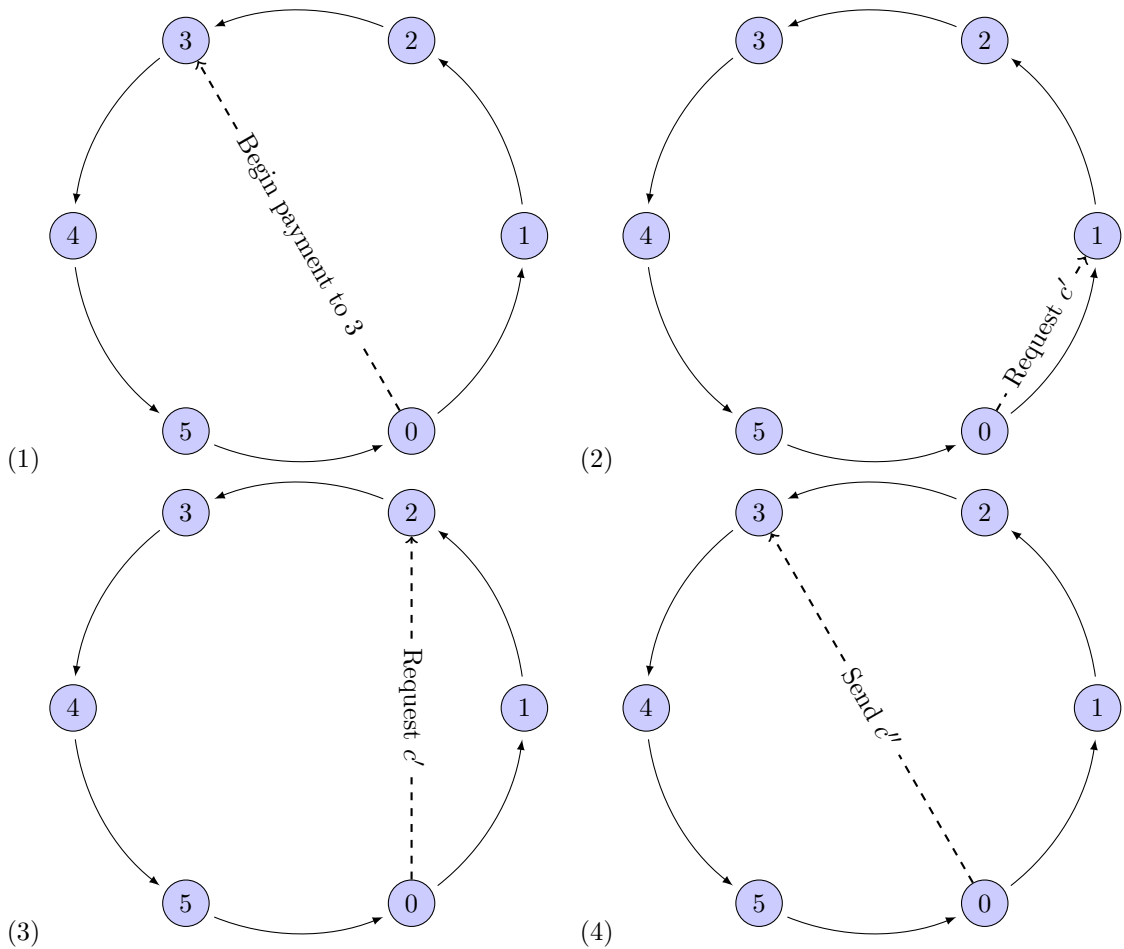
\begin{figure}[!h]
\begin{center}
\begin{tabular}{cc}
(1)
 \begin{tikzpicture}
\def \n {5}
\def \radius {3cm}
\def \margin {8}
\foreach \s in {0,...,\n}
{
  \node[draw, circle, fill=blue!20] (\s) at ({360/(\n+1) * (\s - 1)}:\radius) {$\s$};
  \draw[->, >=latex] ({360/(\n+1) * (\s - 1)+\margin}:\radius) 
    arc ({360/(\n+1) * (\s - 1)+\margin}:{360/(\n+1) * (\s)-\margin}:\radius);
}
    
\path (0) -- node[sloped] (text) {Begin payment to $3$} (3);
  \draw[->,thick,dashed] (0)--(text)--(3);
\end{tikzpicture}  & (2) \begin{tikzpicture}
\def \n {5}
\def \radius {3cm}
\def \margin {8}
\foreach \s in {0,...,\n}
{
  \node[draw, circle, fill=blue!20] (\s) at ({360/(\n+1) * (\s - 1)}:\radius) {$\s$};
  \draw[->, >=latex] ({360/(\n+1) * (\s - 1)+\margin}:\radius) 
    arc ({360/(\n+1) * (\s - 1)+\margin}:{360/(\n+1) * (\s)-\margin}:\radius);
}
\path (0) -- node[sloped] (text) {Request $c'$} (1);
  \draw[->,thick,dashed] (0)--(text)--(1);
  
\end{tikzpicture} \\

 (3) \begin{tikzpicture}
\def \n {5}
\def \radius {3cm}
\def \margin {8}
\foreach \s in {0,...,\n}
{
  \node[draw, circle, fill=blue!20] (\s) at ({360/(\n+1) * (\s - 1)}:\radius) {\s};
  \draw[->, >=latex] ({360/(\n+1) * (\s - 1)+\margin}:\radius) 
    arc ({360/(\n+1) * (\s - 1)+\margin}:{360/(\n+1) * (\s)-\margin}:\radius);
}
  
\path (2) -- node[sloped] (text) {Request $c'$} (0);
  \draw[->,thick,dashed] (0)--(text)--(2);
\end{tikzpicture} &  (4) \begin{tikzpicture}
\def \n {5}
\def \radius {3cm}
\def \margin {8}
\foreach \s in {0,...,\n}
{
  \node[draw, circle, fill=blue!20] (\s) at ({360/(\n+1) * (\s - 1)}:\radius) {$\s$};
  \draw[->, >=latex] ({360/(\n+1) * (\s - 1)+\margin}:\radius) 
    arc ({360/(\n+1) * (\s - 1)+\margin}:{360/(\n+1) * (\s)-\margin}:\radius);
}
\path (0) -- node[sloped] (text) {Send $c''$} (3);
  \draw[->,thick,dashed] (0)--(text)--(3);
\end{tikzpicture}\\
\end{tabular}
\end{center}
\caption{Illustration of a payment from process $P_0$ to process $P_3$ in construction \ref{construction:cycle-coin}.}
\end{figure}

\newpage

\begin{construction}\label{construction:cycle-coin}
(Cycle Coin): 

For simplicity, we will first describe a construction which \textit{almost} works, and requires all processes (regardless of whether they are faulty or not) respond in a particular way which is detectable. We will then describe how to remove this assumption. 

\ 

We number the $N$ processes from $0,\dots,N-1$. Begin by connecting all processes into a cycle, with $P_{n}$ directionally connected and leading to $P_{n+1 \text{ mod } N}$. We let $P_{a,b}$ denote the path along this graph which connects $P_{a}$ to $P_{b}$, with $P_{a}$ included and $P_{b}$ excluded. If $a=b$, then we let $P_{a,b}$ be the empty path. Recall if we have values $v_i$, we let $(\dots ((v \cdot v_1)_{P_{i_1}}\cdot v_{2})_{P_{i_2}} \dots \cdot v_{i_k})_{P_{i_k}}$ denote the string produced when $v \cdot v_1$ is signed by $P_{i_1}$, the resulting string is concatenated with $v_2$ and then signed by $P_{i_2}$ and so on. Notionally, we rewrite this as $(v)[v_1 P_{i_1},\dots,v_k P_{i_k}]$. We define a \text{chain} of length $0$ \textit{ending} at $P_0$ and \textit{extended} by $P_0$ to be the string $("")_{P_0}=()[P_0]=[P_0]$. We define a chain of length $i$ \textit{ending} at $P_{end}$ and \textit{extended} by $P_{ext}$ to be a string of the form $(c)[P_{i_1},xP_{ext},P_{i_2},xP_{ext},P_{i_3},xP_{ext},\dots, P_{i_k},yP_{ext}]$, where $c$ is a chain of length $i-1$ ending at $P_{ext}$, and $P_{ext}=P_{i_1},\dots,P_{i_k}$ are the processes along the ordered path $P_{P_{ext},P_{end}}$. We let $(n_1,...,n_k)$ denote a chain $c$ of length $k$ constructed by a chain of length $1$ ending at $P_{n_1}$, extended by $P_{n_1}$ to a chain of length $2$ ending at $P_{n_2}$, and so forth, and finally ending at $P_{n_k}$ (note that there is a bijective correspondence between a chain $c$ and its representation $(n_1,...,n_k)$). As some concrete examples where $N=6$:

$$(0)=[P_0,(yP_0)]$$

$$(1)=[P_0,(P_0,yP_0)]$$

$$(2)=[P_0,(P_0,xP_0,P_1,yP_0)]$$

$$(3)=[P_0,(P_0,xP_0,P_1,xP_0,P_2,yP_0)]$$

$$(1,3)=[P_0,(P_0,xP_0),(P_1,xP_1,P_2,yP_1)]$$

The chain of length $0$ corresponds to the empty tuple $()$. Given a chain $c=(n_1,...,n_k)$, we let $weight(c)$ be the length of the walk consisting of the paths $P_{0} \rightarrow P_{n_1} \rightarrow \dots \rightarrow P_{n_k}$ along the cycle (where we interpret $P_{a} \rightarrow P_{a}$ as the empty path). A partially constructed chain $c'$ is a string of the form $c'=(n_1,\dots,P_{ext})[P_{i_1},xP_{ext},P_{i_2},xP_{ext},P_{i_3},xP_{ext},\dots P_{i_k},xP_{ext}]$ with $P_{i_1}\rightarrow \dots \rightarrow P_{i_k}$ being some cycle path starting from $P_{ext}$. Given a partially constructed chain $c'$, we let $weight(c')$ denote the weight of the (complete) chain produced by replacing the last $xP_{ext}$ with $yP_{ext}$. 

\ 

The key property of the protocol is the following: an honest protocol $P_{n}$ will decide it is marked iff it receives a chain ending at $P_{n}$, and it \textit{knows} there does not exist a chain of greater weight. $P_{n}$ can be certain of this fact because if there were ever a chain of greater weight, $P_{n}$ would have had to sign it. Likewise, $P_{n}$ can always prevent chains of greater weight from being produced once it has decided it has been marked, by refusing to extend chains. Inductively, $P_{n}$ can then extend the chain of greatest weight to a new process $P_{end}$ as a form of payment. The details are as follows:

\newpage

\begin{enumerate}
    \item At round $0$, the marked process $P_{0}$ begins with the empty chain $()$ of length $0$. 
    
    \item At round $i=0,\dots,K-1$, the marked process $M$ has a chain $c$ of length $i$ ending at $M$: we say $c$ "marks" $M$. It then does the following on receiving input $I_{M}=n$:
    
    \begin{enumerate}
        
        \item $M$ communicates with all processes on the path $P_{M,P_{n}}$ to create a chain $c''$, where $c''$ is extended by $M$, contains $c$ as a prefix, and ends at $P_{n}$. If $P_{n} \not = M$, $M$ does this by beginning with the partial chain $c'=((c)_{M}\cdot x)_{M}$: for each process on the path $P_{M,P_{n}}$ in order (excluding $M$), $M$ messages process $P_{i_j}$ with query $c'$ and asks it to send back $(c')_{P_{i_j}}$. $M$ then updates $c'=((c')_{P_{i_j}}\cdot x)_{M}$ and moves to the next process in the path, or updates $c'=((c')_{P_{i_j}}\cdot y)_{M}$ if $P_{i_j}$ is the last process on the path. If all processes on $P_{M,P_{n}}$ comply, this takes at most $\mathcal{O}(|P_{M,P_{n}}|)$ messages. 
        
        \item $M$ then sends the completely constructed chain $c''$ to process $P_{n}$
    \end{enumerate}
    
    \item At round $i=0,\dots,K-1$, an honest process $P_{n}$ reacts the following way to a request from $P_{a}$ to help extend a chain $c$ to a chain $c''$. We assume $P_{n}$ is sent a partially constructed chain $c'$ which it is asked to append its signature to. 
    
    \begin{enumerate}
        \item Checks that $c$ is a chain of length $i$, and that $c$ ends at $P_{a}$, and that the form of $c'$ is valid. These properties are required for $P_{a}$'s request to be valid, and these properties can be verified by any process. 
        
        \item If $P_{n}$ has previously signed a partial chain $\gamma$ with $weight(\gamma)\geq weight(c')$, then $P_{n}$ refuses to extend $c'$ and replies with $(\gamma)_{P_{n}}$ instead.
        
        \item If $P_{n}$ has received a chain $\gamma$ ending at $P_{n}$ with $weight(\gamma)\geq weight(c')$, then $P_{n}$ refuses to extend $c'$ and replies with $\gamma$ instead.
    
        \item Otherwise $P_{n}$ responds to $P_{a}$ with $(c')_{P_{n}}$. $P_{a}$ can then create a partial chain $((c')_{P_{n}} \cdot x)_{P_{a}}$ or chain $((c')_{P_{n}} \cdot y)_{P_{a}}$ with weight one greater than $c'$.
        
    \end{enumerate}
    
    \item At round $i=0,\dots,K-1$, an honest process $P_{n}$ reacts the following way when receiving a chain $c$ of length $i+1$ ending at $P_{n}$:
    
    \begin{enumerate}
        \item If $P_{n}$ has not signed a partial chain $\gamma$ with $weight(\gamma)=weight(c)$, and $P_{n}$ has not previously received a chain $\gamma$ ending at $P_{n}$ with $weight(\gamma)=weight(c)$, then $P_{n}$ decides it is the marked process for the next round, and that the marked process in the current round is the extender of $c$.
    \end{enumerate}
    
\end{enumerate}

\end{construction}

\begin{prop}\label{prop:cyc-coin}
Suppose that all processes, regardless of whether they are faulty, always produce some valid response when requested to append their signature to a valid partial chain in construction \ref{construction:cycle-coin}. That is, they always comply with an extension request, or publish a "proof" that they are not required to extend a particular request.\footnote{Technically, we only require that the processes on the payment path $P_{M,P_{n}}$ behave in this way.} Then construction \ref{construction:cycle-coin} is a solution to the marker problem which tolerates $f<N-1$ faults. If at round $i$, $M$ sends the marker to $P_{n}$, then the number of steps taken and number of messages sent until $P_{n}$ decides that it is the marked process for round $i+1$ is $\mathcal{O}(|P_{M,P_{n}}|)$. 

\ 

Moreover, $M$ has a proof it has paid $P_{n}$ in the following sense: If a third party Alice would like a proof $M$ paid $P_{n}$ at round $i$, then $M$ sends to Alice the chain $c$ of length $i+1$ ending at $P_{n}$ it used to pay $P_{n}$. Alice then sends $c$ to $P_{n}$. After Alice sends $c$ to $P_{n}$, if $P_{n}$ has not yet decided it is the marked process at round $i+1$, then $P_{n}$ has a chain or partial chain $\gamma$ with $weight(\gamma)=weight(c)$ which does not equal $c$, which proves $M$ did not pay $P_{n}$ or $M$ is dishonest\footnote{In the future, we will use the term dishonest to refer to both (a) $M$ not following the prescribed protocol, and (b) $M$ lying about paying $P_{b}$}. In particular, $P_{n}$ cannot produce such a proof if $M$ is honest. 
\end{prop}

\newpage

\begin{proof}
We need to check the three conditions: 

\begin{enumerate}
    \item Consistency: Suppose for the sake of contradiction that two distinct honest processes $P_{n},P_{n'}$ both decide in round $i$ they are marked for round $i+1$. Then both processes possess corresponding chains $c_{n},c_{n'}$ of length $i+1$ which end at $P_{n},P_{n'}$ and are extended by $P_{ext},P_{ext'}$ resp. Without loss of generality, assume $weight(c_{n})>weight(c_{n'})$. But then it must be the case that $P_{n'}$ signed a prefix chain/partial chain $\gamma$ of $c_{n}$ with $weight(\gamma)= weight(c_{n'})$. By construction, $P_{n'}$ does not do this in round $r\leq i$ if it receives $c_{n'}$ before being requested to make this signature; thus $P_{n'}$ must have decided it was marked after signing $\gamma$, but this is a contradiction. 
    
    \item Liveness: When $M$ decides it is marked by $c$, there exist no other chains or partial chains of greater weight (otherwise $M$ would have had to sign them). By construction, this means when $M$ messages a process $P_j \in P_{M,P_{n}}$ with partial chain $c'$ in an attempt to extend $c'$, $P_{j}$ can only give the response $(c')_{P_j}$. This is because $M$ is required to sign the end of every partial chain extended by $M$, and so during this process the only partial chains of maximum weight that have been produced by the network are the partial chains $M$ explicitly creates, and all chains constructed up until this time have weight $\leq weight(c)$ (because $M$ never appended $y)_{M}$ to a partial chain with weight $\geq weight(c)-1$). Thus after $M$ receives a response from the last process in the path, $M$ will have a string $\gamma$ which it can sign to produce a completed chain $\tilde{c}=(\gamma \cdot y)_{M}$, where $\tilde{c}$ is extended by $M$, ends at $P_{n}$, and contains $c$ as a prefix. $M$ then sends $\tilde{c}$ to $P_{n}$. By construction, $P_{n}$ will not have previously seen a chain of greater or equal weight ending at $P_{n}$, nor will $P_{n}$ have signed a partial chain of weight $\geq weight(\tilde{c})$. If $P_{n}$ is honest, then $P_{n}$ correctly decides that $M$ was the previously marked process and that $P_{n}$ is the next marked process. 
    
    \item Non-impersonation: By the previous paragraph, the only way for an honest process $P_{n}$ to accept a payment via chain $c$, and decide $c$ was extended by honest process $P_{a}$, is if $P_{a}$ was indeed the extender and appended its signature to the end of $c$. 
    
\end{enumerate}

Note the following: if $M$ is honest and marked via chain $c$ during round $i$, then all chains and partial chains of weight greater than $weight(c)$ produced by the network will have $c$ as a prefix: suppose not, and that $\gamma$ is a chain/partial chain which does not contain $c$ as a prefix. Let $\gamma'$ be the prefix chain/partial chain of $\gamma$ with $weight(\gamma')=weight(c)$, $\gamma'\not = c$. Then $(\gamma')_{M}$ is a prefix of $\gamma$, but by construction if $M$ decides it is marked by chain $c$ in round $i$, it will never sign and never has signed $\gamma'$, contradicting that $\gamma$ exists. With regards to the proof of payment, note that if $M \in \mathcal{H}$ is marked by chain $c$, and did indeed pay $P_{n}$ with chain $c''$ extended by $M$ from $c$, then all partial chains $\gamma$ which do not contain $c$ as a prefix satisfy $weight(\gamma)\leq weight(c)<weight(c'')$, and all chains $\gamma$ which do not contain $c$ as a prefix satisfy $weight(\gamma) \leq weight(c)<weight(c'')$, so $P_{n}$ cannot prove $M$ is dishonest. Conversely, if chain $c''$ is not a valid chain which causes $P_{n}$ to accept a payment from $M$ in round $i$, then $P_{n}$ has seen a chain/partial chain $\gamma\not = c''$ with $weight(\gamma)= weight(c'')$.   

\end{proof}

It remains to explain how we can deal with the issue that a dishonest process $P_{n}$ may refuse to respond altogether, instead of validly responding to an extension request. We can solve this problem by running Byzantine Broadcast as a subprocess. Suppose a process $P_{a}$ (honest or not) sends an extension request to process $P_{b}$ for the extension of a partial chain $c'$, and $P_{b}$ does not respond validly within the correct time. We then allow $P_{a}$ to broadcast to all processes that it would like the network to \textit{decide} on the response $P_{b}$ gives to query $c'$. Importantly, note that we are not saying anything about deciding the \textit{guilt} of $P_{b}$ (for example, it could be that $P_{a}$ simply behaved dishonestly and ignored $P_{b}$'s valid response). Assuming all of this works, all honest processes then reach an agreement on \textit{some} valid extension $\gamma$ of $c'$ which is produced by $P_{b}$, or the default null value $\perp$. If all processes agree on a valid extension $\gamma$, then in particular $P_{a}$ decides on a valid response $\gamma$ from $P_{b}$ and the problem is resolved\footnote{$\gamma$ is necessarily a response from $P_{b}$, because it ends with $P_{b}$'s signature which cannot be forged}. Note this case occurs if $P_{b}$ is honest. If $P_{a}$ instead decides on $\perp$($P_{b}$ is dishonest), then we have all honest processes "imagine" that $P_{b}$ is deleted from the network, and $P_{b}$'s signature is no longer required to construct chains. 

\ 

We need to be slightly careful about how we implement this idea. We give concise details below: 

\begin{construction}\label{proof-of-response}
(Proof of response for construction \ref{construction:cycle-coin})

\

During round $i$, the marked node $M$ tries to extend $c$ by asking for extensions along the path $P_{M,n}$. We break this up into $N$ time periods (the maximum number of extensions needed) each of a predefined number of steps, where during the $k$th time period: 

\begin{enumerate}
    \item At step $1$, $M$ sends an extension request $c'$ to $P_{j}$, the $k$th process in $P_{M,n}$ (this time period is empty if $k>|P_{M,n}|$)
    
    \item At step $1$ in this time period, all honest processes each initialize $N^2$ instances of Byzantine Broadcast in parallel, the implementation of which is detailed in the proof of proposition \ref{prelim:prop:BB-det}. Each instance $(a,b) \in [N]$ is a Byzantine Broadcast protocol for $P_{a}$ to issue a network query for $P_{b}$ about a query request $c'$. The essential property of this particular implementation of Byzantine Broadcast is that no messages are sent if the broadcaster does not broadcast anything. 
    
    \item At step $2$, if $P_{j}$ does not respond, $M$ broadcasts $(c',P_{j})$ to the network. 
    
    \item At step $z=\mathcal{O}(f)$, all the $N^2$ Byzantine Broadcast protocols have terminated. If $M$ is honest and broadcasts $(c',P_{j})$, then all honest processes will decide this. Each honest process checks that the request $c'$ is valid, otherwise the request is ignored.
    
    \item For a further $\mathcal{O}(f)$ steps, all honest processes run Byzantine Broadcast to reach agreement on any network queries of responsiveness for any process $P_{b}$. If a response of $c'$ is requested for $P_{b}$, then Byzantine Broadcast is run with $P_{b}$ as the leader, and $P_{b}$'s broadcast is taken as its response. By the end of these steps, all honest processes agree which queries were made, and the responses to these queries. 
\end{enumerate}

In total, each payment now takes $T=\mathcal{O}(fN)$ steps to complete. However, the number of messages is still only $\mathcal{O}(|P_{M,n}|)$ per round when all processes behave honestly, because no messages are sent by the Byzantine Broadcast protocols in this case. When processes behave dishonestly, the number of messages is $\mathcal{O}(N \times N^2 \times Nf)=\mathcal{O}(N^5)$ per round in the worst case.
\end{construction}

As a sanity check, note that construction \ref{construction:cycle-coin} satisfies the lower bound of proposition \ref{message-complexity:prop:mclb} and is tight as well: when all processes behave honestly, $\sum_{n \in [N]} z_n = \sum_{n \in [N-1]} \Omega(n)=\Omega(Nf)$ since the construction tolerates $f=N-2$ faults. 

\ 

We now build on the ideas in this construction to give a fully fledged payment system which is robust to any $f<N-1$ failures. Moreover, if we have some control over the distribution of income and spending patterns, then when all processes behave honestly, the message complexity is at most $\mathcal{O}(\log(N))$. This marks a distinct shift from constructions of payment systems which are based on consensus mechanisms which reach agreement on all transactions, where $\Omega(N)$ messages per transaction are inherently required. The central idea is that if we could only "hop around" the parts of cycles which are large, then we could avoid the payment paths which require $\Omega(N)$ messages.

\newpage

\part{Extending The Fault Tolerant Model}\label{part:trusted-parties}

\section{Locality Through Trust}

We ended the last section by giving a simple construction of a payment system which has some degree of \textit{locality}: in some transactions, only a few neighbouring processes are messaged in order for a transaction to occur. In this section, we will show how by expanding our model in a natural way, we can significantly bootstrap this locality to construct deterministic payment systems whose payments are \textit{highly local}, and give best case message complexity significantly better than that of randomized consensus solutions, while still tolerating an arbitrary number of faults. 

\ 

\subsection{Trusted Anonymous Third Parties For Indirection}

The idea of having a trusted third party to mediate interactions between processes in a network has been used successfully to design a number of efficient protocols\footnote{for example, secret sharing and anonymous messaging.}. However, we argue that there is a certain kind of trusted third party model which is particularly natural to payment systems. We give a story to motivate this kind of trust: 

\ 

Suppose Alice would like to pay Bob $\$100$ for the comic books he gave her. However, unfortunately Bob's account is held with bank B, while Alice's account is held with bank A, and it is known that there is a very high transaction fee for transfers between these two banks. Luckily, Alice finds an advertisement on the internet from a fellow comic book enthusiast Charlie. Charlie's account is held with bank C, which happens to have very good transaction rates with both banks A and B. Alice hatches a plan: if she can ask Charlie to pay Bob on her behalf, then she could pay Charlie back the difference. With all the savings, she might even be willing to give Charlie an extra transaction fee, and then everyone (except the banks) would get a profit. However, Alice and Charlie have never met, and so Alice is cautious about trusting Charlie. To begin with, she sends Charlie $\$1$. Charlie, being the honest comic book enthusiast he is, happily sends across Alice's $\$1$ to Bob. Alice then asks Bob to confirm the transaction went through. Being a little more confident now, Alice decides to send through $\$2$. Over time, Alice and Charlie might begin to develop a trusting relationship. If Charlie ever cheats and doesn't pay Bob on Alice's behalf, then Alice will know and stop using Charlie as an intermediary. Moreover, she will only lose at most the maximum amount she transacted through Charlie at any single time. Charlie probably doesn't want to cheat: he would lose a profitable revenue stream of future payments from Alice, and Alice might start ruining Charlie's reputation. In reality, we are more likely to have a free market situation: multiple agents like Charlie will try to build a business as efficient intermediaries between Alice and Bob. If an intermediary ever cheats, there will be many more intermediaries to happily take its place. The key properties of this scenario which make third party trust natural are\footnote{We note that this model of trust is also the \textit{de facto} model under which Bitcoin has been operating under in certain contexts: participants will pay goods and service providers in digital currency on the belief that these participants will send them physical goods or provide services in return. Sometimes these providers can be completely anonymous, sometimes the goods themselves are legally questionable, and sometimes the service provides are foreign cryptocurrency startups with unclear legal regulation. There is therefore little opportunity for legal recourse if a trust assumption proves to be invalid.}

\begin{enumerate}
    \item Alice can detect when Charlie is cheating, and get a guaranteed bound on her loss.
    
    \item Alice can stop using Charlie after Charlie cheats once.
    
    \item There are plausible financial incentives for there to exist many choices of good intermediaries Alice can turn to instead. 
\end{enumerate}

Notice that it also makes sense to allow for multiple intermediaries, provided the number of intermediaries is small: Alice may ask Charlie to pay Bob. Instead of Charlie paying Bob directly, it may be cheaper for him to pay Bob via Sam, and so forth. We distill this discussion of trusted payment intermediaries into the following definition: 

\begin{dfn}\label{df:trusted-intermediary}
In a \textbf{trusted payment intermediary} model of payment systems where process $P_{a}$ would like to send a coin to $P_{b}$, we add the optional functionality of allowing process $P_{a}$ to form voluntary agreements with a subset of intermediary processes $\{P_{i_j}\}_{j \in [Z]}$ with the semantics "I, $P_{i_{j}}$, promise to facilitate the payment of $P_{a}$ to $P_{b}$".  We say the payment system in this model is valid if, when such such agreements are used to transact, they have the following properties:

\begin{enumerate}
    \item If $P_{a},P_{b},\{P_{i_j}\}_{j \in [Z]}$ are honest, then the transaction from $P_{a}$ to $P_{b}$ is completed: $P_{b}$ decides $P_{a}$ send a coin to $P_{b}$ in the appropriate round.
    
    \item If $P_{a}$ is honest but $P_{b}$ claims it was not paid, then $P_{a}$ \textit{decides} that at least one process in $\{P_{i_j}\}_{j \in [Z]} \cup \{P_{b}\}$ \textit{cheated}.\footnote{We only stipulate that $P_{a}$ can decide only one intermediary is honest for the following reason: it will be the case, as in the bank analogy, that only one intermediary is required to be dishonest for the transaction to fail: thus only one intermediary is really "responsible for" causing the transaction to fail, and it becomes technically messy to talk about multiple intermediaries being dishonest. Moreover, if multiple intermediaries have ill intentions towards $P_{a}$, they could always collude so that only one of them needs to behave dishonestly at each round, while still preventing payment at each round. Of course, $P_{a}$ may then decide to use an entirely different set of intermediaries all together if a payment fails, and different intermediaries may form their own preferences about which other intermediaries are the most reliable to work with.}
    
    \item If $P_{c} \in \{P_{i_j}\}_{j \in [Z]} \cup \{P_{b}\}$ is honest, then $P_{a}$ does not decide that $P_{c}$ cheated.
\end{enumerate}

\ 

We give no guarantees that the payment will succeed or that $P_{a}$ will keep its coin if $P_{a}$ uses these agreements to transact and any of the parties behave dishonestly. 
\end{dfn}

From the previous section, we know the following payment system exists:

\begin{cor}\label{cycle-payment-systems-cor}
Given any complete directed cycle $C$ on the participant processes, there exists a payment system in the single transaction per round, deterministic model tolerating any $f<N-1$ faults which supports any initial coin distribution and has the following additional properties:

\begin{enumerate}
    \item If any process $P_{c}$ falsely claims to pay some other process $P_{b}$ in round $i$, then $P_{b}$ has a proof to any third party $P_{a}$ in the sense of proposition \ref{prop:cyc-coin} that $P_{c}$ is dishonest.
    
    \item If a process $P_{a}$ pays a process $P_{b}$ in round $i$, and all processes processes are honest, then the message complexity of this round is $\mathcal{O}(|P_{P_{a},P_{b}}|)$.
\end{enumerate}

\end{cor}

\begin{proof}
This follows immediately from the cycle coin solution to the marker problem (construction \ref{construction:cycle-coin}) and the reduction given from payment systems to the Marker Problem (proposition \ref{prop:maker-is-general}), because the cycle coin construction sends no messages when a marker pays itself in a round. 
\end{proof}

Using this building block, we are now ready to give the construction of a payment system which uses indirection between cycles to reduce message complexity. The key idea is the following: construct a payment system by stacking together many cycle payment systems consisting of cycles of different permutations. Now imagine that $P_{a}$ would like to pay $P_{b}$. Since $P_{a}$ has non-zero balance, it has a coin in at least one of the subcycles $C$ which make up the payment system. Now $P_{a}$ has a choice: $P_{a}$ could pay $P_{b}$ using the subpayment system $C$. But if $P_{a}$ is very far away from $P_{b}$ on the cycle, $P_{a}$ might be lucky by finding another process $P_{c}$ such that (a) $P_{a}$ is close to $P_{c}$ on cycle $C$, and (b) $P_{c}$ is close to $P_{b}$ on some other cycle $C'$. If $P_{c}$ has a coin on cycle $C'$, then the following can happen: $P_{a}$ pays $P_{c}$ on cycle $C$, and $P_{c}$ promises to ensure that a payment gets to $P_{b}$. $P_{c}$ then pays $P_{b}$ on cycle $C'$, and sends $P_{a}$ a proof of payment equal to the chain it used to pay $P_{b}$. By "hopping between cycles" we are able to reduce the number of messages we need to send per transaction. In general, we might make multiple hops for a single payment.

\begin{prop}\label{prop:msg-vs-graph-diam}
Consider a collection of $K$ cycle payment systems of the form described in corollary \ref{cycle-payment-systems-cor}, with associated cycles $C_{1},\dots,C_{K}$\footnote{we will interchangably refer to payment systems through their associated cycles}. Concatenate these payment systems together to form a new payment system $PS'$ as in the proof of proposition \ref{prop:maker-is-general} which tolerates $f<N-1$ faults, and let $\mathcal{F} \subset \{P_n\}_{n \in [N]}$ be a set of "trusted intermediary" processes. 

\ 

At round $i$, let $Value(C_{k},n)$ denote the value of process $P_{n}$ in the payment system associated with the simulation of payment system $C_{k}$, and define the following graph $G_{i}$ on the vertex set $[K] \times [N]$ and edge set $E$: 

$$((k,P_{a}),(k',P_{b})) \in E \textit{ if } k=k' \land  (P_{a},P_{b}) \in Edges(C_{k})$$

(cycle step) or

$$((k,P_{a}),(k',P_{b})) \in E \textit{ if } Value(C_{k'},b)>0 \land a=b \in \mathcal{F}$$

(cycle hop)

\ 

Let $\mathcal{W}=\{ k \in [K]|Value(C_{k},a)>0\}$ and $D_{a,b}=\min_{k \in \mathcal{W},k' \in [K]} dist_{G_{i}} ((k,P_{a}),(k',P_{b}))$. Then in the trusted payment intermediary, single payment per round model, we can define the payment system $PS'$ to have the following property: 

\ 

At round $i$ when honest process $P_{a}$ makes a payment to honest process $P_{b}$, if $P_{a}$ can find a path $P$ of length $L$ in $G_{i}$ by using $M$ messages, and all processes behave honestly, then the message complexity of a transaction from $P_{a}$ to $P_{b}$ is $\mathcal{O}(M+L)$. In particular, if $P_{a}$ knows $G_{i}$ at the start of round $i$, then the message complexity is $\mathcal{O}(D_{a,b})$. Moreover, if $\mathcal{O}(|P|)$ trusted intermediaries on the path $P$ behave honestly, then $P_{a}$'s payment is guaranteed to go through to $P_{b}$. 
\end{prop}

We comment that it is a straightforward extension to consider a different set of trusted intermediates $\mathcal{F}_{n}$ for each different process $P_{n}$, however the core ideas are captured by this statement. 

\begin{proof}
Using the primitives we have developed, we describe one potential way for $P_{a}$ to send a payment to $P_{b}$: 

\begin{enumerate}
    \item Let the path from $P_{a}$ to $P_{b}$ be $P$. $P$ consists of a number of cycle hops and cycle steps. Let $(P_{i_1},\dots,P_{i_Z})$ be the intermediate processes which facilitate the cycle hops in path $P$ (see figure \ref{figure:cycle-merge}). 
    
    \item We define a \textit{macro round} which consists of a large number of \textit{micro rounds}. Each macro round corresponds to a single round in $PS'$. Each micro round corresponds to a single round in the simulated cycle payment systems. We might fix the number of micro rounds per macro round to be some upper bound on the diameter of the graph $G_{i}$. 
    
    \item At macro round $i$, $P_{a}$ receives input $b \in [N]$ during the execution of payment system $PS'$. $P_{a}$ now sends $\mathcal{O}(|P|)$ messages to the intermediate processes $P_{i_1},\dots,P_{i_Z}$ asking for a promise of the form "I, $P_{i_{j}}$, will pay $P_{i_{j+1}}$ in micro round $j+1$, if I receive a payment from $P_{i_{j-1}}$ in round $j$". We notate $P_{i_{0}}:=P_{a}, P_{i_{Z+1}}:=P_{b}$. 
    
    \item At microstep $j=0,\dots,Z$, $P_{i_{j}}$ pays a coin to $P_{i_{j+1}}$ using the path $P_{i_{j}} \rightarrow P_{i_{j+1}}$ along the appropriate cycle $C_{k}$ determined by $P$ (cycle steps). Such an action is possible because (we assume without loss of generality) $P$ does not contain two cycle hops from the same process, and $Value(C_{k},P_{i_{j}})>0$ at the beginning of macro round $i$. By corollary \ref{cycle-payment-systems-cor}, all of these actions combined take $\mathcal{O}(|P|)$ messages to perform assuming all processes behave honestly. 
    
    \item At the end of all the microsteps, $P_{a}$ asks for proofs from all the intermediate processes that their payment obligations were satisfied.
    
    \item Conceptually, we have $P_{a}$ sign a message to $P_{b}$ at the beginning of macro round $i$ indicating that it intends to send a payment, and sends the signed promises of the intermediate processes $P_{i_{j}}$ to $P_{b}$ as well. If $P_{b}$ receives a payment from $P_{i_{Z}}$, it then decides that $P_{a}$ paid $P_{b}$ in macro round $i$. 
\end{enumerate}

To see that this gives a valid payment system in the trusted payment intermediary, single transaction per round model, we need to check a few conditions. First note that $S1$ (non-duplication) in Definition \ref{def:full-PS} always holds, because it holds for each simulated payment system for any number of faults $f<N-1$. We assert that when no intermediaries are used, non-impersonation, self consistency and liveness in round $i$ hold by the construction given in the proof of Proposition \ref{prop:maker-is-general}. The new content we need to check is what happens when intermediaries are used in round $i$. We will check non-impersonation still holds, liveness holds when the intermediaries are honest, and that in the case $P_{b}$ claims it is not paid, the conditions of definition of the trusted payment intermediary model (definition \ref{df:trusted-intermediary}) are met. 

\begin{enumerate}
    
    \item non-impersonation: $P_{b}$ only decides that $P_{a}$ paid $P_{b}$ if it explicitly receives a signature from $P_{a}$ indicating it will pay $P_{b}$ at round $i$. Therefore if $P_{a}$ is honest and does not pay $P_{b}$ at round $i$, it cannot be impersonated. 
    
    \item Liveness and proofs of honesty: If $P_{a},P_{b},\{P_{i_j}\}_{j \in [Z]}$ are honest, then the payment goes through as described. If $P_{b}$ claims it did not receive a payment from intermediary $P_{i_{Z}}$, then $P_{a}$ can iteratively go through each of the processes in the order $P_{b}=P_{i_{Z+1}},P_{i_{Z}},P_{i_{Z-1}},\dots$. Inductively, at step $j$, if $P_{i_{j}}$ claims it was not paid by $P_{i_{j-1}}$, then either $P_{i_{j-1}}$ claims it \textit{did} pay $P_{i_{j}}$ and we have a proof of whether $P_{i_{j-1}}$ is being honest by corollary \ref{cycle-payment-systems-cor}, or $P_{i_{j-1}}$ admits it did not pay $P_{i_{j}}$, because it was not paid by $P_{i_{j-2}}$. In this case, $P_{a}$ recursively moves down to step $j-1$. At step $j$, if $P_{j}$ is honest, then either $P_{j}$ is able to prove to $P_{a}$ that it did pay $P_{j+1}$, or $P_{j-1}$ is unable to prove to $P_{a}$ that it paid $P_{j}$. Thus $P_{a}$ will not decide that $P_{j}$ cheated. Since $P_{b}$ claims it was not paid and $P_{a}$ paid $P_{i_1}$ by assumption, $P_{a}$ will decide that at least one of $P_{b}\cup \{P_{i_z}\}_{z \in [Z]}$ cheated. Thus this protocol satisfies all of the conditions of Definition \ref{df:trusted-intermediary}. 
\end{enumerate}
\end{proof}

\newpage
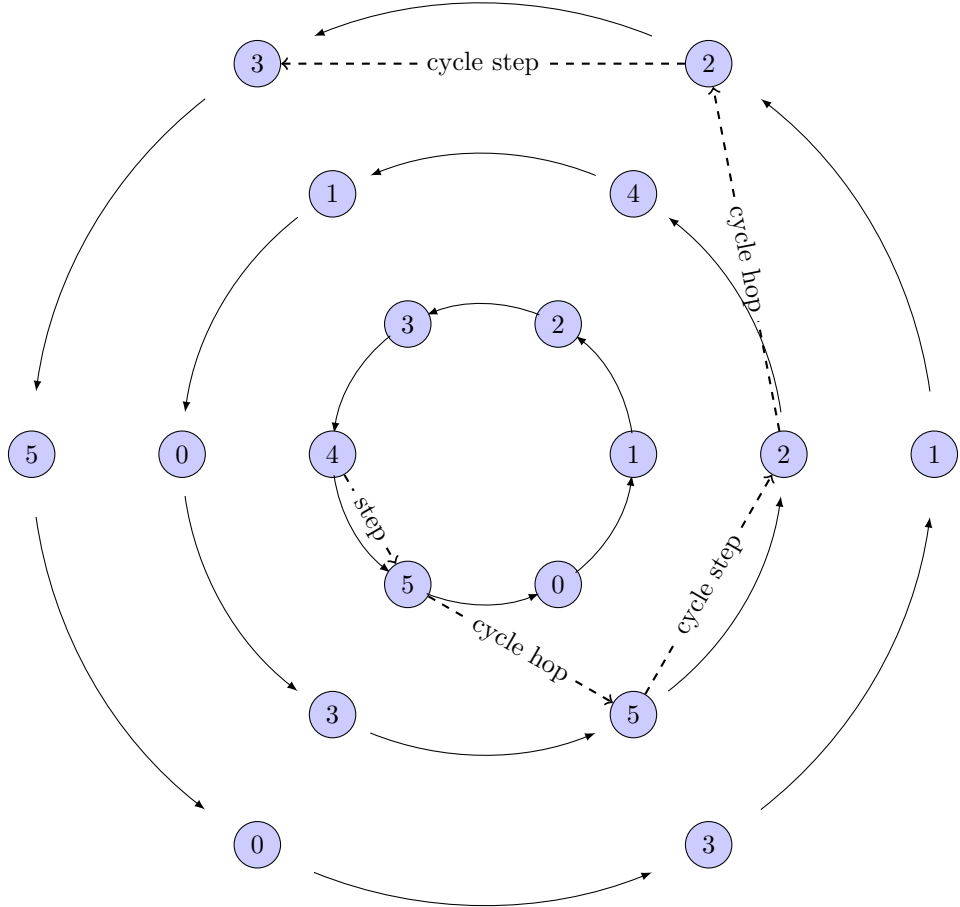
\begin{figure}[!h]
\begin{center}

\begin{tikzpicture}
\def \n {5}
\def \radius {2cm}
\def \margin {8}
\def \s {0}
\def \lab {0}
\node[draw, circle, fill=blue!20] (\s) at ({360/(\n+1) * (\s - 1)}:\radius) {\lab};
  \draw[->, >=latex] ({360/(\n+1) * (\s - 1)+\margin}:\radius) 
    arc ({360/(\n+1) * (\s - 1)+\margin}:{360/(\n+1) * (\s)-\margin}:\radius);
\def \s {1}
\def \lab {1}
\node[draw, circle, fill=blue!20] (\s) at ({360/(\n+1) * (\s - 1)}:\radius) {\lab};
  \draw[->, >=latex] ({360/(\n+1) * (\s - 1)+\margin}:\radius) 
    arc ({360/(\n+1) * (\s - 1)+\margin}:{360/(\n+1) * (\s)-\margin}:\radius);
\def \s {2}
\def \lab {2}
\node[draw, circle, fill=blue!20] (\s) at ({360/(\n+1) * (\s - 1)}:\radius) {\lab};
  \draw[->, >=latex] ({360/(\n+1) * (\s - 1)+\margin}:\radius) 
    arc ({360/(\n+1) * (\s - 1)+\margin}:{360/(\n+1) * (\s)-\margin}:\radius);
\def \s {3}
\def \lab {3}
\node[draw, circle, fill=blue!20] (\s) at ({360/(\n+1) * (\s - 1)}:\radius) {\lab};
  \draw[->, >=latex] ({360/(\n+1) * (\s - 1)+\margin}:\radius) 
    arc ({360/(\n+1) * (\s - 1)+\margin}:{360/(\n+1) * (\s)-\margin}:\radius);
\def \s {4}
\def \lab {4}
\node[draw, circle, fill=blue!20] (\s) at ({360/(\n+1) * (\s - 1)}:\radius) {\lab};
  \draw[->, >=latex] ({360/(\n+1) * (\s - 1)+\margin}:\radius) 
    arc ({360/(\n+1) * (\s - 1)+\margin}:{360/(\n+1) * (\s)-\margin}:\radius);
\def \s {5}
\def \lab {5}
\node[draw, circle, fill=blue!20] (\s) at ({360/(\n+1) * (\s - 1)}:\radius) {\lab};
  \draw[->, >=latex] ({360/(\n+1) * (\s - 1)+\margin}:\radius) 
    arc ({360/(\n+1) * (\s - 1)+\margin}:{360/(\n+1) * (\s)-\margin}:\radius);

\def \n {5}
\def \radius {4cm}
\def \margin {8}
\def \s {6}
\def \lab {5}
\node[draw, circle, fill=blue!20] (\s) at ({360/(\n+1) * (\s - 1)}:\radius) {\lab};
  \draw[->, >=latex] ({360/(\n+1) * (\s - 1)+\margin}:\radius) 
    arc ({360/(\n+1) * (\s - 1)+\margin}:{360/(\n+1) * (\s)-\margin}:\radius);
\def \s {7}
\def \lab {2}
\node[draw, circle, fill=blue!20] (\s) at ({360/(\n+1) * (\s - 1)}:\radius) {\lab};
  \draw[->, >=latex] ({360/(\n+1) * (\s - 1)+\margin}:\radius) 
    arc ({360/(\n+1) * (\s - 1)+\margin}:{360/(\n+1) * (\s)-\margin}:\radius);
\def \s {8}
\def \lab {4}
\node[draw, circle, fill=blue!20] (\s) at ({360/(\n+1) * (\s - 1)}:\radius) {\lab};
  \draw[->, >=latex] ({360/(\n+1) * (\s - 1)+\margin}:\radius) 
    arc ({360/(\n+1) * (\s - 1)+\margin}:{360/(\n+1) * (\s)-\margin}:\radius);
\def \s {9}
\def \lab {1}
\node[draw, circle, fill=blue!20] (\s) at ({360/(\n+1) * (\s - 1)}:\radius) {\lab};
  \draw[->, >=latex] ({360/(\n+1) * (\s - 1)+\margin}:\radius) 
    arc ({360/(\n+1) * (\s - 1)+\margin}:{360/(\n+1) * (\s)-\margin}:\radius);
\def \s {10}
\def \lab {0}
\node[draw, circle, fill=blue!20] (\s) at ({360/(\n+1) * (\s - 1)}:\radius) {\lab};
  \draw[->, >=latex] ({360/(\n+1) * (\s - 1)+\margin}:\radius) 
    arc ({360/(\n+1) * (\s - 1)+\margin}:{360/(\n+1) * (\s)-\margin}:\radius);
\def \s {11}
\def \lab {3}
\node[draw, circle, fill=blue!20] (\s) at ({360/(\n+1) * (\s - 1)}:\radius) {\lab};
  \draw[->, >=latex] ({360/(\n+1) * (\s - 1)+\margin}:\radius) 
    arc ({360/(\n+1) * (\s - 1)+\margin}:{360/(\n+1) * (\s)-\margin}:\radius);
    
\def \n {5}
\def \radius {6cm}
\def \margin {8}
\def \s {12}
\def \lab {3}
\node[draw, circle, fill=blue!20] (\s) at ({360/(\n+1) * (\s - 1)}:\radius) {\lab};
  \draw[->, >=latex] ({360/(\n+1) * (\s - 1)+\margin}:\radius) 
    arc ({360/(\n+1) * (\s - 1)+\margin}:{360/(\n+1) * (\s)-\margin}:\radius);
\def \s {13}
\def \lab {1}
\node[draw, circle, fill=blue!20] (\s) at ({360/(\n+1) * (\s - 1)}:\radius) {\lab};
  \draw[->, >=latex] ({360/(\n+1) * (\s - 1)+\margin}:\radius) 
    arc ({360/(\n+1) * (\s - 1)+\margin}:{360/(\n+1) * (\s)-\margin}:\radius);
\def \s {14}
\def \lab {2}
\node[draw, circle, fill=blue!20] (\s) at ({360/(\n+1) * (\s - 1)}:\radius) {\lab};
  \draw[->, >=latex] ({360/(\n+1) * (\s - 1)+\margin}:\radius) 
    arc ({360/(\n+1) * (\s - 1)+\margin}:{360/(\n+1) * (\s)-\margin}:\radius);
\def \s {15}
\def \lab {3}
\node[draw, circle, fill=blue!20] (\s) at ({360/(\n+1) * (\s - 1)}:\radius) {\lab};
  \draw[->, >=latex] ({360/(\n+1) * (\s - 1)+\margin}:\radius) 
    arc ({360/(\n+1) * (\s - 1)+\margin}:{360/(\n+1) * (\s)-\margin}:\radius);
\def \s {16}
\def \lab {5}
\node[draw, circle, fill=blue!20] (\s) at ({360/(\n+1) * (\s - 1)}:\radius) {\lab};
  \draw[->, >=latex] ({360/(\n+1) * (\s - 1)+\margin}:\radius) 
    arc ({360/(\n+1) * (\s - 1)+\margin}:{360/(\n+1) * (\s)-\margin}:\radius);
\def \s {17}
\def \lab {0}
\node[draw, circle, fill=blue!20] (\s) at ({360/(\n+1) * (\s - 1)}:\radius) {\lab};
  \draw[->, >=latex] ({360/(\n+1) * (\s - 1)+\margin}:\radius) 
    arc ({360/(\n+1) * (\s - 1)+\margin}:{360/(\n+1) * (\s)-\margin}:\radius);

\path (4) -- node[sloped] (text) {step} (5);
  \draw[->,thick,dashed] (4)--(text)--(5);

\path (5) -- node[sloped] (text) {cycle hop} (6);
  \draw[->,thick,dashed] (5)--(text)--(6);
  
\path (6) -- node[sloped] (text) {cycle step} (7);
  \draw[->,thick,dashed] (6)--(text)--(7);
  
\path (7) -- node[sloped] (text) {cycle hop} (14);
  \draw[->,thick,dashed] (7)--(text)--(14);
 
\path (14) -- node[sloped] (text) {cycle step} (15);
  \draw[->,thick,dashed] (14)--(text)--(15);
  
\end{tikzpicture} 
\end{center}
\caption{Illustration of a payment from process $P_{4}$ to process $P_{3}$ consisting of a series of cycle hops and cycle steps. The payment system consists of a concatenation of $3$ cycles: We have $Value(C_{1},P_{4})>0$ in the inner cycle, $Value(C_{2},P_{5})>0$ in the middle cycle, and $Value(C_{2},P_{2})>0$ in the outer cycle. The intermediaries are $P_{5}$ and $P_{2}$.}
\label{figure:cycle-merge}
\end{figure}

We comment that in the above construction, it is not actually necessary for $P_{a}$ to know the entire path to $P_{b}$ in advance, and we can change the semantics of the intermediate trust assumptions as well: for example, $P_{a}$ could simply forward the payment to $P_{b}$ which it thinks is a good intermediate process for this transaction. $P_{b}$ can then sign the promise "I promise to make sure a payment gets to $P_{b}$ by time $t$", and assume full responsibility and trust. Provided $P_{b}$ can figure out the remaining short path and find its own trusted intermediaries, the message complexity can still be made small. 

\

We will refer to payment systems of the form described in proposition \ref{prop:msg-vs-graph-diam} consisting of cycles $C_1,\dots C_K$ as a \textit{cycle payment system} with cycles $C_1,\dots,C_K$\footnote{Thus $C_1$ is a cycle payment system consisting of the cycle $C_1$.}. We now define an unnecessarily strong condition on the spending distribution in a payment system which allows us to state some simple constructions which give low message complexity: 

\newpage

\begin{dfn}\label{condition:balanced}
Consider a cycle payment system consisting of cycles $C_1,\dots,C_K$. We say that the system is \textit{balanced} in the single transaction per round model if, during the beginning of every round, we have that $Value(P_n,C_k)>0$ for all $n \in [N], k \in [K]$.
\end{dfn}

\begin{thm}\label{thrm:short-cyc}
Suppose for simplicity that $\mathcal{F}=[N]$\footnote{Even this extreme case is not entirely unrealistic: we might imagine that the default behavior of participants is to facilitate payments, at the benefit of gaining transaction fees and reputation for being a reliable facilitator.}. Then there exists a deterministic cycle payment system consisting of $2K\geq 4$ cycles $C_1,\dots ,C_K$ and tolerating any $f<N-1$ faults in the trusted payment intermediary model, such that if the payment system is balanced, the best case message complexity is $\mathcal{O}(\log_{K}(N))$ per transaction. The number of trusted intermediaries per transaction is also $\mathcal{O}(\log_{K}(N))$. 
\end{thm}

\begin{proof}
Pick a random, undirected $r=2K\geq 4$ regular graph on $N$ processes, $G_{N,r}$. We use two well known properties about such graphs:

\begin{enumerate}
    \item If $r\geq 4$ is even, then $G_{N,r}$ asymptotically almost surely has a complete Hamiltonian decomposition into edge disjoint Hamiltonian cycles $t_1,\dots,t_{\frac{r}{2}}$. \cite{KW01}
    
    \item Almost every $r$ regular graph has diameter at most $d\geq \mathcal{O}\left(\frac{\log(3rN\log(N))}{\log(r)}\right)=\mathcal{O}(\log_{r}{N})$. \cite{BV82}
\end{enumerate}

Thus, for sufficiently large $N$, we can pick $G_{N,2K}$ such that it has diameter $\mathcal{O}(\log_{K}{N})$, and can be decomposed into undirected Hamiltonian cycles $t_1,\dots,t_K$.

\

For each undirected $t_i$, choose two directed cycles $C_{2i},C_{2i+1}$ for each direction around the cycle $t_i$. Now construct a cycle payment system $PS$ consisting of the cycles $C_{1},\dots,C_{2K}$. Now we invoke proposition \ref{prop:msg-vs-graph-diam}: consider the corresponding graph $G_{i}$ at round $i$ in the statement of the proposition. Because $PS$ is balanced, the cycle hop edges of $G_{i}$ do not change between rounds. Consequently, $G_{i}$ is fixed for all $i$, and determined by the cycles $C_{1},\dots,C_{2K}$; thus we can assume each process $P_{a}$ can compute shortest paths in $G_{i}$ for any round $i$ without the need to send any messages. Moreover, we claim that $\max_{a,b \in [N]} D_{a,b}=\mathcal{O}(\log_{K}{N})$, from which the claim follows by proposition \ref{prop:msg-vs-graph-diam}. To see this, notice that for any $a,b \in [N]$, given a shortest path $P$ in $G_{N,2K}$ connecting $P_{a}$ to $P_{b}$, we can find a corresponding path in $G_{i}$ with at most twice the number of edges: The path $P$ can be decomposed into a sequence of paths $l_1 \rightarrow \dots \rightarrow l_z \rightarrow \dots \rightarrow l_Z$ where each $l_z$ moves in a particular direction around some Hamiltonian cycle $t_{j_{z}}$, and $t_{j_{z}}$ and $t_{j_{z+1}}$ are distinct cycles. For notational convenience, we will say that $l_{z}$ ends at the same vertex $l_{z+1}$ begins (so that adjacent paths share endpoints). By the choices of the cycles $C_{1},\dots,C_{2K}$, we can follow each path $l_z$ in $G_{N,2K}$  with a path $l'_{z}$ of the same length in $G_{i}$, by following the corresponding cycle (either $C_{2j_{z}}$ or $C_{2j_{z}+1}$, depending on orientation) in $G_{i}$: both paths begin and end at the same process. Moreover, by following a cycle hop between cycles and using the fact that the payment system is balanced, we can find an edge $e_{z}$ in $G_{i}$ so that the paths $l_{z} \rightarrow l_{z+1}$ in $G_{N,2K}$ and $l'_{z} \rightarrow e_{z} \rightarrow l'_{z+1}$ both end at the same processes. Inductively, $P'=l'_{1} \rightarrow e_{1} \dots \rightarrow e_{Z-1} \rightarrow l'_{Z}$ has at most twice the path length of $l_{1} \rightarrow \dots \rightarrow l_{Z}$, and $P'$ starts at process $P_{a}$ and ends at process $P_{b}$. It follows by another use of the well balanced condition that $D_{a,b} \leq |P'|$. 
\end{proof}

We emphasize that the statement of \ref{thrm:short-cyc} is, in an important measure, weaker than what can actually be achieved. Given a highly connected graph $G_{i}$, the process $P_{a}$ initiating the payment may have reasonable flexibility in choosing which participants are required to be honest for the message complexity of the transaction to be low. In particular, $P_{a}$ can actively look for processes willing to facilitate the transaction when trying to find a short path through $G_{i}$, rather than being at the mercy of an arbitrary subset of $\mathcal{O}(\log_{K}(N))$ participants to behave honestly. 

\

There are many different choices of cycles one can try to combine to get small graph diameter. Condition \ref{condition:balanced} is overly strong because in practice, especially on random graphs, removing a few edges because we cannot make a cycle jump (say, because $Value(C_k,P_j)=0$ on a cycle $C_k$) will not significantly affect the graph diameter. $P_{a}$ might try a local search along a few different short paths before finding one which works. To illustrate this idea more concretely, we briefly sketch another construction using two cycles: while we don't give any analysis, we leave it to the reader to convince themselves that such a construction allows for small message complexity when the balance condition is satisfied, and that the construction is reasonably robust to removing some between-cycle edges. 

\begin{construction}

\ 

\begin{enumerate}
    \item Start with the standard cycle on $N$ processes, where the edges have the orientation $n \rightarrow n-1$. 
    
    \item Construct a sequence of paths recursively as follows: Start at process $0$. Extend a directed edge to the median $m$ of $[0,N]$. Now recursively (a): continue extending this path to the median of $[m,N]$, and (b) start a new path at $m-1$, extending it to the median of $[0,m-1]$, and so on. At the end of this process, we will have a number of disjoint directed paths. Join them all together to get the second cycle. The picture looks like this:
    
    \begin{figure}[!h]
\begin{center}

 \begin{tikzpicture}
\def \n {28}
\def \radius {5cm}
\def \margin {8}
\foreach \s in {0,...,\n}
{
  \node[draw, circle, fill=blue!20] (\s) at ({360/(\n+1) * (\s - 1)}:\radius) {$\s$};
  \draw[->, >=latex] ({360/(\n+1) * (\s - 1)+\margin}:\radius) 
    arc ({360/(\n+1) * (\s - 1)+\margin}:{360/(\n+1) * (\s)-\margin}:\radius);
}

\path [->,dashed] (0) edge (14);

\path [->,dashed,bend left] (14) edge (21);

\path [->,dashed,bend right] (13) edge (6);

\path [->,dashed,bend left] (6) edge (9);

\path [->,dashed,bend right] (5) edge (3);

\path [->,dashed,bend left] (21) edge (25);

\path [->,dashed,bend right] (20) edge (17);

\path [->,dashed,bend left] (25) edge (27);

\path [->,dashed,bend right] (2) edge (1);

\path [->,dashed,bend right] (24) edge (23);

\path [->,dashed,bend left] (17) edge (18);

\path [->,dashed,bend right] (16) edge (15);

\path [->,dashed,bend left] (9) edge (11);

\path [->,dashed,bend left] (11) edge (12);

\path [->,dashed,bend right] (8) edge (7);

\path [->,dashed,bend left] (3) edge (4);

\path [->,dashed,bend left] (27) edge (28);
  
\end{tikzpicture}
\end{center}
\caption{Recursively constructed "binary search" cycles for $N=29$.}
\end{figure}
\end{enumerate}
\end{construction}

We end this section with a final comment: One may wonder to what extent $3$rd party trust is really necessary for this construction. For example, if $P_{c}$ promises $P_{a}$ it will pay $P_{b}$ in round $i$ but does not, can we not use a similar idea like the one in construction \ref{proof-of-response} (proof of response) to have the network come to a consensus about whether $P_{c}$ cheated? The key issue is making sure that if we catch $P_{c}$ cheating, $P_{c}$ still has some non-zero value it can be forced to pay back to compensate $P_{a}$. Note that a solution of low best case message complexity which doesn't rely on trusted intermediaries doesn't necessarily violate the message complexity lower bound of $\Omega(Nf)$ for payment systems in the best case (Proposition \ref{message-complexity:prop:mclb}), because this was proved for the Marker Problem where the distribution of income is centered only on $P_1$, but our low message constructions rely on the distribution of income being sufficiently spread out with respect to the network topology. Such an idea may therefore work, but the solution constructed in the fault tolerant model may also make heavy use of timing assumptions which would not realistically translate to use over the internet. In practice, one would need to make $P_{c}$ put money in escrow for a certain amount of time while it was behaving as an intermediary, and this may require some kind of global consensus to achieve. The idea of using peer to peer payment indirection, putting money in escrow, and using consensus to resolve peer to peer indirection disputes is very similar in spirit to the Lightning Network \cite{lightning-network}. We find it interesting that by considering a corner case of Lemma \ref{message-complexity:lem:message-complexity-intersection} in an ideal model of payment systems, we have been led down a road which ended with ideas similar to those being experimented with in real world payment systems. We now turn the reader's attention to the appendix, which gives a high level overview of the lightning network and its relation to Part \ref{part:trusted-parties}.  

\subsection{Trusted Third Parties for Coordinating Payment Cancellation}

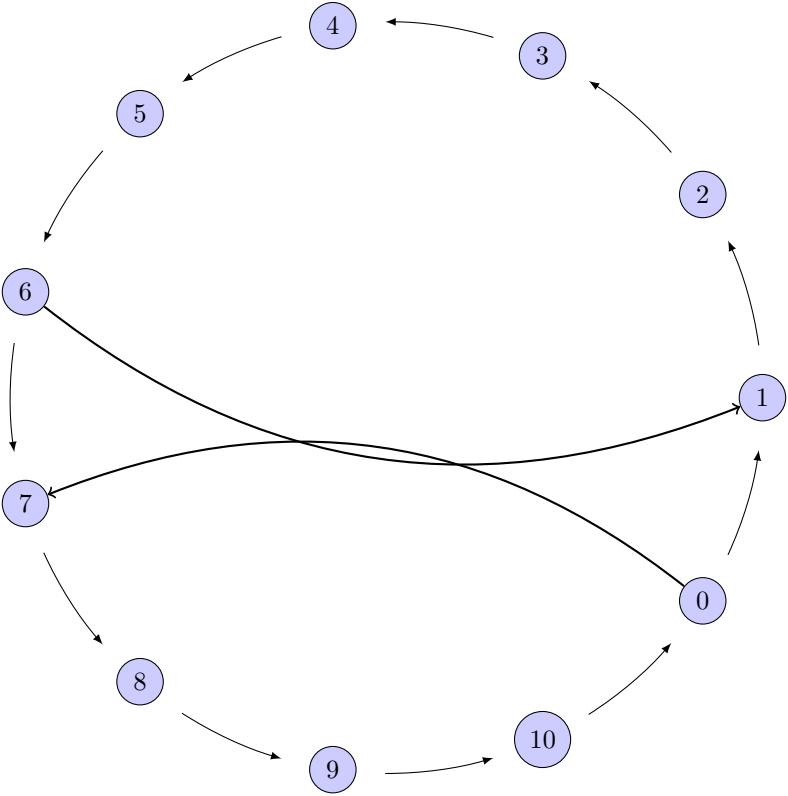
\begin{figure}[!h]
\begin{center}
\begin{tikzpicture}
\def \n {10}
\def \radius {5cm}
\def \margin {8}
\foreach \s in {0,...,\n}
{
  \node[draw, circle, fill=blue!20] (\s) at ({360/(\n+1) * (\s - 1)}:\radius) {$\s$};
  \draw[->, >=latex] ({360/(\n+1) * (\s - 1)+\margin}:\radius) 
    arc ({360/(\n+1) * (\s - 1)+\margin}:{360/(\n+1) * (\s)-\margin}:\radius);
}
    
\path [->,thick, bend right] (0) edge (7);

\path [->,thick, bend right] (6) edge (1);

\end{tikzpicture}
\end{center}
\caption{The payments $P_{0} \rightarrow P_{7}$ and $P_{6} \rightarrow P_{1}$ can be paired together cancel their common payment path intersection, reducing the overall message complexity.}
\end{figure}

In this last subsection, we sketch a simple idea which can also be used to reduce the message complexity in a third party trusted model. Imagine a case where there is a trusted third party which is responsible for coordinating information about payments between processes. If the third party is dishonest, the only consequence is that the message complexity per transaction increases (but no security guarantees are violated). In practice, we imagine a market of such "information intermediaries" which compete to offer services which give the best information/reduction in message complexity. 

\ 

The idea is fairly simple: consider at round $i$, in the multiple payments per round model, of processes making $Q$ payments $\{P_{a_i} \rightarrow P_{b_i}\}_{i \in [Q]}$ on a cycle payment system consisting of a single cycle $C$. If $P_{a_i}$ is much closer to $P_{b_j}$ on $C$ than $P_{a_i}$ is to $P_{b_i}$, and symmetrically, then the central coordinator can pair $P_{a_i},P_{a_j}$ together so that $P_{a_i}$ pays $P_{b_j}$ on behalf of $P_{a_j}$ and $P_{a_j}$ pays $P_{b_i}$ on behalf of $P_{a_i}$ under the same trust model of Definition \ref{df:trusted-intermediary}. In general, multiple processes may cooperate together to collectively "cancel" their payment paths, and be required to collectively trust each other. In practice, we imagine a system where a central coordinator will have clients which routinely join and coordinate with other processes in this group. The membership of a process will be conditional on it never cheating in the group, and the coordinator may offer insurance against such behavior. Better coordinators will receive more clients, and well-behaved clients will be able to join coordinators with larger client pools, leading to reduced message complexity per transaction. 

\

How should the coordinator pair multiple transactions together? Given a collection of payers/sources $\{P_{a_i}\}_{i \in [Q]}$ and recipients/sinks $\{P_{b_i}\}_{i \in [Q]}$, define the cost of the pairing $P_{a_i}$ with $P_{b_j}$, $c_{i,j}$ to be the length of the directed cycle on $C$ from $P_{a_i}$ to $P_{b_j}$, which is proportional to the message complexity associated with such a payment. By repairing difference sources/sinks, we would like to minimize the sums of these costs. Consider the algorithm GREEDY, which at each iteration $i \in [Q]$, picks an arbitrary source/sink, and pairs it with a sink/source not already paired which has the smallest cost associated with its pairing. 

\begin{prop}
GREEDY produces a pairing with minimal total cost.
\end{prop}

\begin{proof}
Induct on the number of sources $Q$, with $Q=1$ being immediate. Now suppose we have a collection $Q$ of sources $\{P_{a_i}\}_{i \in [Q]}$ and sinks $\{P_{b_i}\}_{i \in [Q]}$ which have not already been paired. Let $S$ be a pairing which minimizes the total cost. Without loss of generality, pick an arbitrary source $P_{a_i}$ which is paired to $P_{b_j}$ in solution $S$. If $P_{b_j}$ is already a sink which minimizes the path length $|P_{P_{a_i},P_{b_j}}|$, we remove the pairing $P_{a_i},P_{b_j}$ from the source and sink list and are done by induction. Otherwise, there is some sink $P_{b_{j'}}$ with $|P_{P_{a_i},P_{b_{j'}}}|<|P_{P_{a_i},P_{b_j}}|$. Let $P_{a_{i'}}$ be the source connected to $P_{b_{j'}}$ in solution $S$. Construct a new solution $S'$ which repairs $P_{a_{i}}$ with $P_{b_{j'}}$ and $P_{a_{i'}}$ with $P_{b_{j}}$, hence having the property that $P_{a_{i}}$ is paired with a nearest sink. If we can show $S'$ does not have greater total cost than $S$, then we will be done by induction. 

\ 

Since $|P_{P_{a_i},P_{b_{j'}}}|<|P_{P_{a_i},P_{b_j}}|$, it must be that $P_{b_{j'}}$ lies on the path $P_{P_{a_i},P_{b_j}}$ and occurs strictly before $P_{b_j}$. Consider three cases: (i), $P_{a_{j'}}$ lies on the path $P_{P_{a_i},P_{b_j}}$ and in addition (a) occurs on the path $P_{a}$ to $P_{b_{j'}}$ ($S'$ has the same cost as $S$), (b) occurs on the path $P_{b_{j'}}$ to $P_{b}$ ($S'$ has cost less than or equal to $S$), (ii) $P_{a_{j'}}$ does not lie on the path $P_{P_{a_i},P_{b_j}}$ ($S'$ has the same cost as $S$). 
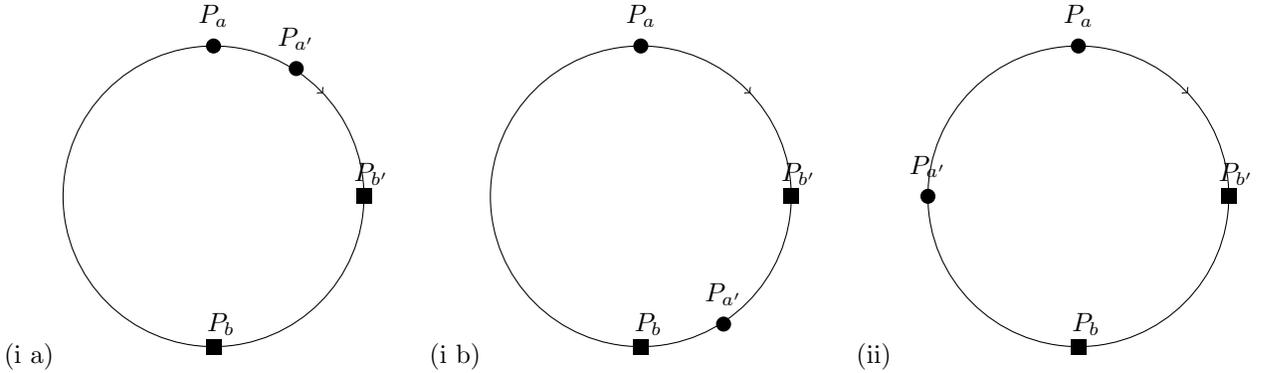
\begin{figure}[!h]
\begin{center}
\begin{tabular}{ccc}
(i a)
 \begin{tikzpicture}

    \draw[ 
        decoration={markings, mark=at position 0.125 with {\arrow{<}}},
        postaction={decorate}
        ]
        (0,0) circle (2);


    \fill (-0.1,-1.9) rectangle (0.1,-2.1) node[text=black, label=$P_{b}$] {}; 
    \fill (0,2) circle (0.1) node[text=black, label=$P_{a}$] {}; 
    \fill (1.9,0.1) rectangle (2.1,-0.1) node[text=black, label=$P_{b'}$] {}; 
    \fill (1.1,1.7) circle (0.1) node[text=black, label=$P_{a'}$] {}; 
    
\end{tikzpicture}  & (i b) 

\begin{tikzpicture}
\draw[ 
        decoration={markings, mark=at position 0.125 with {\arrow{<}}},
        postaction={decorate}
        ]
        (0,0) circle (2);


    \fill (-0.1,-1.9) rectangle (0.1,-2.1) node[text=black, label=$P_{b}$] {}; 
    \fill (0,2) circle (0.1) node[text=black, label=$P_{a}$] {}; 
    \fill (1.9,0.1) rectangle (2.1,-0.1) node[text=black, label=$P_{b'}$] {}; 
    \fill (1.1,-1.7) circle (0.1) node[text=black, label=$P_{a'}$] {}; 
\end{tikzpicture} & (ii) 

\begin{tikzpicture}
\draw[ decoration={markings, mark=at position 0.125 with {\arrow{<}}},
        postaction={decorate}
        ]
        (0,0) circle (2);


    \fill (-0.1,-1.9) rectangle (0.1,-2.1) node[text=black, label=$P_{b}$] {}; 
    \fill (0,2) circle (0.1) node[text=black, label=$P_{a}$] {}; 
    \fill (1.9,0.1) rectangle (2.1,-0.1) node[text=black, label=$P_{b'}$] {}; 
    \fill (-2,0) circle (0.1) node[text=black, label=$P_{a'}$] {}; 
\end{tikzpicture} 
\end{tabular}
\end{center}
\caption{Cases for GREEDY.}
\end{figure}

\end{proof}

While it is not done so here, it might be interesting to explore simple conditions on the distribution of payments in the multiple transaction per round model which cause GREEDY to give low cost solutions. 

\begin{prop}
In the multiple transaction per round, trusted payment intermediary model, there exists a deterministic payment system tolerating any $f<N-1$ faults with best case message complexity in round $i$ of $|GREEDY(\{P_{a_i}\}_{i \in [Q]},\{P_{b_i}\}_{i \in [Q]})|/Q$, where $\{P_{a_i}\}_{i \in [Q]}$ are the processes which make payments to $\{P_{b_i}\}_{i \in [Q]}$ in round $i$. 
\end{prop}

\newpage

\part{Wrapping Up}\label{part:conclusion}

\section{Conclusions}

This thesis began with the aim of interrogating the following assumption:

\begin{assumption*}
Distributed payment systems cannot exist without achieving regular global consensus about which payments have occurred.
\end{assumption*}

By the end of this thesis, we were able to get a clearer idea of the validity of this assumption in the following ways:

\paragraph{In the formal fault tolerance model, payment systems are weaker than consensus:}

\

We showed in proposition \ref{prop:PS-to-BG} that there is a single step back box reduction from payment systems to Byzantine Broadcast, showing that if we can solve the consensus problem, then we can implement a payment system. On the other hand, by using known lower bound results for Byzantine Broadcast and constructing a low round complexity solution to the Marker Problem, we were able to show in Theorem \ref{thm:no-black-box} that there exists no black box, or even "see-through box" reduction from Byzantine Broadcast to payment systems which is significantly better than trivially using a payment system as an unauthenticated messaging channel.

\ 

\paragraph{Under a reasonable trust model, we need not achieve regular global consensus in order to facilitate payments:}

\ 

In section \ref{sec:message-complexity}, Construction \ref{construction:cycle-coin}, we showed how to construct simple solutions to the Marker Problem which did not require global consensus to transact when processors are well behaved. In contrast, we also showed a tight lower bound of $\Omega(Nf)$ for the best case message complexity of the Marker Problem in Proposition \ref{message-complexity:prop:mclb}, indicating that if we wanted to do much better than Construction \ref{construction:cycle-coin}, we would either need to change the model of our problem, or make assumptions about the distribution of income. In part \ref{part:trusted-parties}, we extended our model to allow trusted intermediaries. We showed that there exist highly fault tolerant payment systems in this model which facilitate inherently local transactions in the best case. In Theorem \ref{thrm:short-cyc} we collected these ideas to show that under certain transaction distribution assumptions, we can tolerate any $f<N-1$ faults with best case message complexity  $\mathcal{O}(\log_{d}(N))$ per transaction and $\mathcal{O}(\log_{d}(N))$ trusted intermediaries per transaction. Major cryptocurrencies such as Bitcoin do not tolerate more than $f$ faults for $2f<N$, and require achieving consensus at every round about which transactions have occurred (in our model, which is not the same model as the one a real-world cryptocurrency operates under, this would require sending $\Omega(N)$ messages per transaction in the best case). This result implies that in the model we have chosen, the assumption that payment systems require regular global consensus about which transactions have occurred is not necessarily true: at the least, this degree of consensus is not needed in the best case. 

\

Despite these results, it is also important to point out that all of the constructions given were analyzed in models which are different than the ones in which practical cryptocurrencies operate under (these differences are detailed in section \ref{sec:crypto} part \ref{part:prelim}). While we hope the ideas given in these formal models lead to useful translations in more practical settings, such translations are not always obvious. We state the gaps in this understanding in the form of posing new problems:

\newpage

\paragraph{Part \ref{part:payment-systems}}

\begin{concquestion}
Does there exist a reduction from the $K$ round marker problem to the $2$ round Marker Problem? What about the $1$ Round Marker Problem? (Definition \ref{payment-systems:df:marker-game}).
\end{concquestion}

\ 

Notice that the heavy requirement that "all processes behave honestly" in order to get low message complexity in Construction \ref{construction:cycle-coin}, and hence the construction in Theorem \ref{thrm:short-cyc}, is because the dispute resolution mechanism which forces processes to extend chains uses Byzantine Broadcast as a black box. If anyone requests the network to decide on a response for process $P_{n}$, then the message complexity is dominated by the number of messages required to run Byzantine Broadcast. If we could therefore force processes to be responsive without this consensus mechanism, then we could significantly lower the message complexity even without best case behavior. 

\

\begin{concquestion}
Is there a way to force, or strongly incentivize processes to be responsive in Construction \ref{construction:cycle-coin} without a global consensus mechanism?\footnote{For example, if there is a financial penalty associated with being unresponsive, then we can have a bound on the message complexity in terms of how much an adversary needs to spend to increase the number of messages.}
\end{concquestion}

\paragraph{Part \ref{part:trusted-parties}}

\begin{concquestion}
By using tools for leader election in Randomized Byzantine Agreement, can the ideas presented in part \ref{part:trusted-parties} be extended to the realistic case where the $N$ participants are unknown and may be \textit{online} or \textit{offline}? 
\end{concquestion}

\begin{concquestion}
Can we find realistic, more well motivated conditions on the distribution of transactions which guarantee we can always efficiently find a short payment path in Proposition \ref{prop:msg-vs-graph-diam}?
\end{concquestion}

\begin{concquestion}
Given a transaction distribution, can we always construct a cycle payment system with the guarantee that there exists (with high probability) a short payment path as in Proposition \ref{prop:msg-vs-graph-diam}?
\end{concquestion}

\newpage

\section{Appendix}

\subsection{The Lightning Network and Part \ref{part:trusted-parties}}

 We will briefly describe the core ideas behind the lightning network, and then move to explaining how Part \ref{part:trusted-parties} relates to these ideas.  

\ 

The lightning network is designed to operate as a second layer protocol, on top of Bitcoin. There does not yet seem to be any formal analysis of its performance, although there has been some empirical analysis \cite{lightning-analysis-1}, \cite{lightning-analysis-2}. The idea is the following: suppose Alice and Bob regularly make small payments to each other. Instead of making all of these small transactions on the blockchain, Alice can create a payment channel with Bob. This involves Alice "paying" $1$ Bitcoin into a new payment channel with Bob by putting a transaction on the blockchain, effectively putting her $1$ Bitcoin in escrow. Alice's payment channel is effectively a local ledger between Alice and Bob, where Alice has balance $1$ and Bob has balance $0$. Now when Alice wants to pay Bob 0.5 Bitcoin, she sends it through her payment channel by updating the local ledger. Alice and Bob both sign the update, and Bob is "paid" through the payment channel. Alice and Bob now both have 0.5 Bitcoin on the local payment channel. If Bob wants to now pay Alice, he can send back the 0.5 Bitcoin on the payment channel by signing an updated copy of the local ledger. All of this happens without communicating to the blockchain, besides the initial setup of the payment channel. Now, if Bob has 0.5 Bitcoin in the payment channel with Alice, but would like to pay Charlie with this value, Bob dissolves the payment channel by publishing it to the blockchain. In particular, Bob publishes the most updated version of the local payment channel to the Blockchain. There's a chance Bob can cheat by publishing an outdated version of the payment channel which doesn't contain his payment to Alice (the blockchain cannot tell the difference, because the payment channel only involved communication between Alice and Bob). This problem is solved by putting a timelock on how quickly Bob can dissolve the channel: when he tries to do this, Alice has a few days to publish a "more recent" version of the payment channel to prove Bob is cheating\footnote{Since the more recent version contains Bob's signature in it, this proves that Bob intentionally published an outdated ledger.}. If Bob is cheating and Alice does this, Alice gets all the Bitcoin in the payment channel. Otherwise the payment channel is dissolved, and Alice's $1$ Bitcoin which she originally deposited in escrow to create the payment channel is now split between Alice and Bob on the global blockchain, according to the local ledger of the payment channel. 

\ 

The reason the lightning network is interesting is the following: again, suppose that Bob wanted to pay Charlie. There is another way for Bob to do this without dissolving his payment channel with Alice. In particular, if Alice has a local payment channel with Charlie in which Alice has a positive balance of 0.5 Bitcoin, then Bob can make Alice sign the promise "If Bob pays me 0.5 Bitcoin in the channel between Alice and Bob, I'll pay Charlie 0.5 Bitcoin in the channel between Alice and Charlie". If this happens, then Bob pays Alice and Alice pays Charlie, meaning that Bob effectively pays Charlie. None of this required any messages on the blockchain. The main constraint of this construction is only the capacity of the payment channels\footnote{This constraint also seems to be a major empirical limitation on the ability of the lightning network to facilitate transactions.} (for example, if Alice and Bob's payment channel started with 1 Bitcoin, Alice can only be an intermediary for Bob for a value of at most $1$ Bitcoin), and the topology of the connections of these channels. If Alice cheats and does not pay Charlie, then Bob publishes Alice's promise and proof of her cheating on the blockchain, and Alice loses the balance in her local payment channel. Of course, this idea of paying someone through an intermediary can be done inductively: the lightning network imagines that everyone might form connections and process the majority of small transactions through this network, leaving only large transactions for the blockchain. 

\

The lightning network's idea of using payment indirection between 2 party payment channels is very similar to Part \ref{part:trusted-parties}'s idea to use payment indirection to hop over long cycles. Note however that the lightning network and Part \ref{part:trusted-parties} operate in very different trust and network models:

\begin{enumerate}
    \item The lightning network is designed to be practical over the internet: participants might be temporarily offline\footnote{But they cannot be offline for too long, because then Alice might not be able to catch Bob cheating in time if he publishes an outdated copy of the local payment channel to the blockchain.}, and there are less strong timing assumptions than in a synchronous network. 
    
    \item The lightning network does not trust payment intermediaries (if an intermediary is dishonest, this can be proved and published to the blockchain), while Part \ref{part:trusted-parties}  is motivated in the context of having many intermediaries which can be trusted to behave well due to market incentives. 
\end{enumerate}

Despite these model differences, we believe there is an intuitive way to view the relation between the constructions in Part \ref{part:trusted-parties} and the lightning network: 

\begin{enumerate}
    \item For each local payment channel of one unit of value between $P_n,P_{n'}$ in the lightning network, "deposit" this value into a cycle-coin construction where the only two participants in the cycle are $P_n,P_{n'}$. Local two party payment channels in the lightning network correspond to two-member cycles: if Bob would like to pay Alice in the two-cycle, he signs the current longest chain, which currently ends at Bob, and sends it to Alice. Notice that in Construction \ref{construction:cycle-coin}, Alice's signature is not required to extend the chain in the special case of a 2-cycle: only the sender needs to sign the chain.  
    
    \item Each payment which moves between two payment channels in the lightning network corresponds to a "cycle hop" between 2-cycles in the network. 
    
    \item For Bob to dissolve a local payment channel and redeem a unit of value, he publishes the longest chain (ending at Bob) currently created by the 2-cycle which proves Bob is the marked process in this cycle. Alice has finite time to refute this proof by publishing a chain of greater weight. 
\end{enumerate}

\begin{figure}[!h]
    \centering
    \begin{tikzpicture}
    \def \scale {4}
     \node[shape=circle,draw=black] (B) at (0,0) {B};
     
     \node[shape=circle,draw=black] (A1) at (0,\scale) {A};
     
     \node[shape=circle,draw=black] (A2) at (\scale,\scale) {A};
     
     \node[shape=circle,draw=black] (A3) at (\scale+\scale,\scale+\scale) {A};
     
     \node[shape=circle,draw=black] (D) at (\scale+\scale,\scale) {D};
     
     \node[shape=circle,draw=black] (C) at (\scale,0) {C};
     
     \node[shape=circle,draw=black] (C2) at (\scale+\scale,0) {C};
     
     \node[shape=circle,draw=black] (E) at (\scale+\scale+0.5*\scale,0.5*\scale) {E};

     \path [->,solid,bend left] (B) edge (A1);
     \path [->,solid,bend left] (A1) edge (B);
     
     \path [->,solid,bend left] (A2) edge (C);
     \path [->,solid,bend left] (C) edge (A2);
     
     \path [->,solid,bend left] (A3) edge (D);
     \path [->,solid,bend left] (D) edge (A3);
     
     \path [->,solid,bend left] (C2) edge (E);
     \path [->,solid,bend left] (E) edge (C2);
     
     \path [-,dashed] (A1) edge (A2);
     \path [-,dashed] (A2) edge (A3);
     \path [-,dashed] (A3) edge (A1);
     \path [-,dashed] (C) edge (C2);

    \end{tikzpicture}
    \caption{The lightning network as a collection of 2-cycles connected by cycle-hops. The solid arrows correspond to 2-cycle payment channels, and the dashed lines correspond to cycle hops between 2-cycles.}
\end{figure}
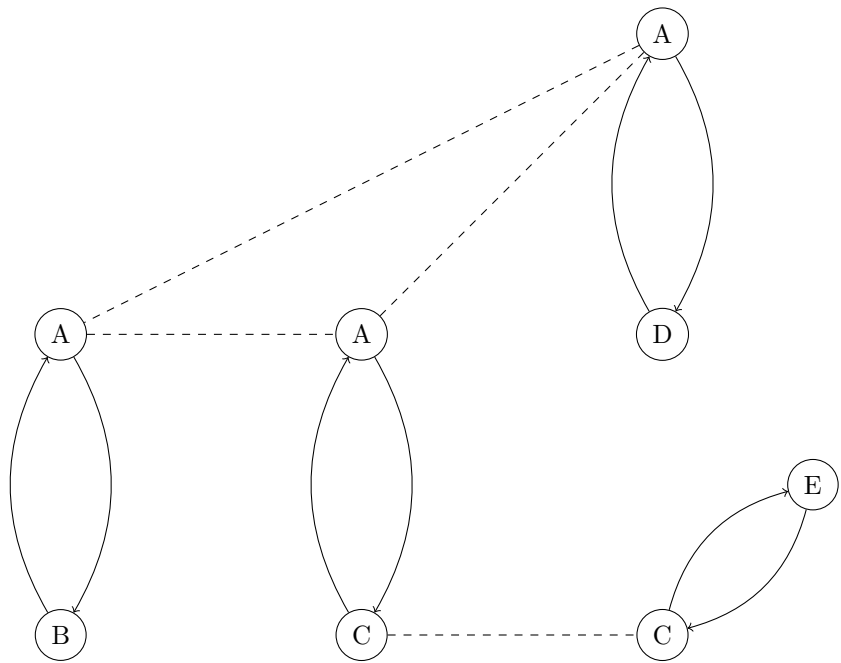

In general, the lightning network corresponds to the network topology of having a large number of 2-cycles of Construction \ref{construction:cycle-coin} connected together by cycle hops. The lightning network chooses to reach consensus on whether participants are honest in these cycle hops, which is why it has stronger security guarantees every time a payment moves through a cycle hop. We justify this correspondence as follows: notice that two party payment systems and 2-cycles have the same security guarantees and mechanism of enforcement. Bob only accepts an updated ledger from Alice if it is an extended ledger from the one they agreed to most recently. Likewise, when redeeming value, Bob proves that his ledger is the most recent by allowing Alice the opportunity to present a counter example of a ledger which is longer (more recent) which contains Bob's signature, and Alice has a finite time in which to do this. In a 2-cycle of Construction \ref{construction:cycle-coin}, Bob only accepts payment from Alice if Alice signs a chain with greater weight than any chain Bob has seen before. Likewise, when redeeming value, Bob can prove that he has the marker in the local payment channel if he publishes a chain ending at Bob, and Alice cannot produce a greater weight chain. This proof of value is identical to the "proof of valid response" in Construction \ref{construction:cycle-coin}, where participants prove that a cycle extension request is invalid by publishing a chain of greater weight. Notice however that in the special case of a 2-cycle construction, there is never a need to ask another participant to extend a cycle if you are the marked process. Thus network consensus is only needed when wanting to prove one particular process is marked in the 2-cycle, but consensus is not needed for in-cycle payments when the marker moves between Alice and Bob, i.e. 2-cycle payments do not require use of the global blockchain. In both a 2-cycle and a payment channel in the lightning network, both Alice and Bob only have finite time to broadcast this a proof or counter example to the network\footnote{For 2-cycles, these need to be broadcast within the current round.}, so both the lightning network and Construction \ref{construction:cycle-coin} use a time locking mechanism. The cycle hops in both systems correspond to using individual participants to act as intermediaries between different payment channels/2-cycles. 

\ 

Thus, when we formally studied the constructions in Part \ref{part:trusted-parties}, we were studying the properties of a structure which is closely related to that of the lightning network, the primary difference being that we used $N$-cycles instead of $2$-cycles. We hope that the formalism developed in parts \ref{part:payment-systems}, \ref{part:trusted-parties} are therefore useful in setting up a way to formally understand the security and efficiency of systems like the lightning network. 

\


\newpage

\begin{thebibliography}{99}

\bibitem{PSL80}
Marshall C. Pease, Robert E. Shostak, Leslie Lamport:
Reaching Agreement in the Presence of Faults. J. ACM 27(2): 228-234 (1980)

\bibitem{LSP82}
Leslie Lamport, Robert E. Shostak, Marshall C. Pease:
\emph{The Byzantine Generals Problem.} ACM Trans. Program. Lang. Syst. 4(3): 382-401 (1982)

\bibitem{FL82}
Michael J. Fischer, Nancy A. Lynch:
A Lower Bound for the Time to Assure Interactive Consistency. Inf. Process. Lett. 14(4): 183-186 (1982)

\bibitem{DF+82}
Danny Dolev, Michael J. Fischer, Robert J. Fowler, Nancy A. Lynch, H. Raymond Strong:
An Efficient Algorithm for Byzantine Agreement without Authentication. Information and Control 52(3): 257-274 (1982)

\bibitem{BV82}
Bollobás, B., Fernandez de la Vega, W. Combinatorica (1982) 2: 125. \url{https://doi.org/10.1007/BF02579310}

\bibitem{DS83}
Danny Dolev, H. Raymond Strong:
Authenticated Algorithms for Byzantine Agreement. SIAM J. Comput. 12(4): 656-666 (1983)

\bibitem{BO83}
Michael Ben-Or:
Another Advantage of Free Choice: Completely Asynchronous Agreement Protocols (Extended Abstract). PODC 1983: 27-30

\bibitem{R83}
Michael O. Rabin:
Randomized Byzantine Generals. FOCS 1983: 403-409

\bibitem{TC84}
Russell Turpin, Brian A. Coan:
Extending Binary Byzantine Agreement to Multivalued Byzantine Agreement. Inf. Process. Lett. 18(2): 73-76 (1984)

\bibitem{FLP85}
Michael J. Fischer, Nancy A. Lynch, Mike Paterson:
\textit{Impossibility of Distributed Consensus with One Faulty Process.} J. ACM 32(2): 374-382 (1985)

\bibitem{DR85}
Danny Dolev, Rüdiger Reischuk:
Bounds on Information Exchange for Byzantine Agreement. J. ACM 32(1): 191-204 (1985)

\bibitem{FLM86}
Michael J. Fischer, Nancy A. Lynch, Michael Merritt:
Easy Impossibility Proofs for Distributed Consensus Problems. Distributed Computing 1(1): 26-39 (1986)

\bibitem{DLS88}
Cynthia Dwork, Nancy A. Lynch, Larry J. Stockmeyer:
Consensus in the presence of partial synchrony. J. ACM 35(2): 288-323 (1988)

\bibitem{CFN88}
David Chaum, Amos Fiat, Moni Naor:
Untraceable Electronic Cash. CRYPTO 1988: 319-327

\bibitem{CMS89}
Benny Chor, Michael Merritt, David B. Shmoys:
Simple constant-time consensus protocols in realistic failure models. J. ACM 36(3): 591-614 (1989)

\bibitem{KW01}
Jeong Han Kim, Nicholas C. Wormald, 
Random Matchings Which Induce Hamilton Cycles and Hamiltonian Decompositions of Random Regular Graphs, Journal of Combinatorial Theory, Series B, Volume 81, Issue 1, January 2001, Pages 20-44

\bibitem{DSU04}
Xavier Défago, André Schiper, Péter Urbán:
Total order broadcast and multicast algorithms: Taxonomy and survey. ACM Comput. Surv. 36(4): 372-421 (2004)

\bibitem{LLR06}
Yehuda Lindell, Anna Lysyanskaya, Tal Rabin:
On the composition of authenticated Byzantine Agreement. J. ACM 53(6): 881-917 (2006)

\bibitem{BBRP07}
Roberto Baldoni, Marin Bertier, Michel Raynal, Sara Tucci Piergiovanni:
Looking for a Definition of Dynamic Distributed Systems. PaCT 2007: 1-14

\bibitem{KK09}
Jonathan Katz, Chiu-Yuen Koo:
On expected constant-round protocols for Byzantine agreement. J. Comput. Syst. Sci. 75(2): 91-112 (2009)

\bibitem{Bit}
Bitcoin: A Peer-to-Peer Electronic Cash System. 2009.

\bibitem{AH10}
Hagit Attiya, Keren Censor-Hillel:
Lower Bounds for Randomized Consensus under a Weak Adversary. SIAM J. Comput. 39(8): 3885-3904 (2010)

\bibitem{SCG+14}
Eli Ben-Sasson, Alessandro Chiesa, Christina Garman, Matthew Green, Ian Miers, Eran Tromer, Madars Virza: Zerocash: Decentralized Anonymous Payments from Bitcoin. IEEE Symposium on Security and Privacy 2014: 459-474

\bibitem{GKL15}
Juan A. Garay, Aggelos Kiayias, Nikos Leonardos:
The Bitcoin Backbone Protocol: Analysis and Applications. EUROCRYPT (2) 2015: 281-310

\bibitem{LBSZR15}
Yoad Lewenberg, Yoram Bachrach, Yonatan Sompolinsky, Aviv Zohar, Jeffrey S. Rosenschein:
Bitcoin Mining Pools: A Cooperative Game Theoretic Analysis. AAMAS 2015: 919-927

\bibitem{LTKS15}
Loi Luu, Jason Teutsch, Raghav Kulkarni, Prateek Saxena:
Demystifying Incentives in the Consensus Computer. ACM Conference on Computer and Communications Security 2015: 706-719

\bibitem{E15}
Ittay Eyal:
The Miner's Dilemma. IEEE Symposium on Security and Privacy 2015: 89-103

\bibitem{SBBR16}
Okke Schrijvers, Joseph Bonneau, Dan Boneh, Tim Roughgarden:
Incentive Compatibility of Bitcoin Mining Pool Reward Functions. Financial Cryptography 2016: 477-498

\bibitem{SSZ16}
Ayelet Sapirshtein, Yonatan Sompolinsky, Aviv Zohar:
Optimal Selfish Mining Strategies in Bitcoin. Financial Cryptography 2016: 515-532

\bibitem{NKMS16}
Kartik Nayak, Srijan Kumar, Andrew Miller, Elaine Shi:
Stubborn Mining: Generalizing Selfish Mining and Combining with an Eclipse Attack. EuroS\&P 2016: 305-320


\bibitem{lightning-network}
Joseph Poon, Thaddeus Dryja.
The Bitcoin Lightning Network: Scalable Off-Chain Instant Payments. 2016. 

\bibitem{CKWN16}
Miles Carlsten, Harry A. Kalodner, S. Matthew Weinberg, Arvind Narayanan:
On the Instability of Bitcoin Without the Block Reward. ACM Conference on Computer and Communications Security 2016: 154-167

\bibitem{KKKT16}
Aggelos Kiayias, Elias Koutsoupias, Maria Kyropoulou, Yiannis Tselekounis:
Blockchain Mining Games. EC 2016: 365-382

\bibitem{PS17}
Rafael Pass, Elaine Shi:
The Sleepy Model of Consensus. ASIACRYPT (2) 2017: 380-409

\bibitem{ZXD+17}
Zibin Zheng, Shaoan Xie, Hongning Dai, Xiangping Chen, Huaimin Wang:
An Overview of Blockchain Technology: Architecture, Consensus, and Future Trends. BigData Congress 2017: 557-564

\bibitem{youtubevid-micali}
Micali, S. (2017, July 25). Dr. Silvio Micali, MIT, DLS [Video file]. Retrieved from \url{https://youtu.be/QNQHbfI3IAQ}

\bibitem{youtubevid-ver}
How Lightning Network Scales For The World - Lightning Network Explained (2018, December 28). Roger Ver [Video file]. Retrieved from \url{https://www.youtube.com/watch?v=Xg_-dz5PqAY}

\bibitem{EE18}
Ittay Eyal, Emin Gün Sirer:
Majority is not enough: bitcoin mining is vulnerable. Commun. ACM 61(7): 95-102 (2018)

\bibitem{SJS+18}
Nicholas Stifter, Aljosha Judmayer, Philipp Schindler, Alexei Zamyatin, Edgar R. Weippl:
\emph{Agreement with Satoshi - On the Formalization of Nakamoto Consensus.} IACR Cryptology ePrint Archive 2018: 400 (2018)


\bibitem{M18}
Silvio Micali
Byzantine Agreement, Made Trivial, 2018


\bibitem{tangle}
Serguei Popov,
The Tangle, 2018

\bibitem{CGMV18}
Jing Chen, Sergey Gorbunov, Silvio Micali, Georgios Vlachos (2018) 
ALGORAND AGREEMENT: Super Fast and Partition Resilient Byzantine Agreement

\bibitem{TE18}
Itay Tsabary, Ittay Eyal:
The Gap Game. ACM Conference on Computer and Communications Security 2018: 713-728

\bibitem{lightning-analysis-1}
István András Seres, László Gulyás, Dániel A. Nagy, Péter Burcsi:
Topological Analysis of Bitcoin's Lightning Network. CoRR abs/1901.04972 (2019)

\bibitem{lightning-analysis-2}
Stefano Martinazzi:
The evolution of Lightning Network's Topology during its first year and the influence over its core values. CoRR abs/1902.07307 (2019)


\bibitem{coinmarketcap}
coinmarketcap.com, a database for cryptocurrency statistics. Last accessed March 2019. 


\end{thebibliography}
\end{document}